\documentclass[letterpaper,11pt,oneside,reqno]{amsart}
\usepackage{amsfonts,amsmath, amssymb,amsthm,amscd}
\usepackage[usenames,dvipsnames]{color}
\usepackage{bm}
\usepackage{mathrsfs}
\usepackage{yfonts}
\usepackage[mpexclude,DIV13]{typearea}
\usepackage{verbatim}
\usepackage{hyperref}
\usepackage{graphicx}
\usepackage[latin1]{inputenc}
\usepackage{latexsym}
\usepackage{lscape}
\usepackage{epsfig}
\usepackage{subfigure}
\usepackage{setspace}
\include{bibtex}

\begin{document}

\bibliographystyle{alpha}
\newcommand{\cn}[1]{\overline{#1}}
\newcommand{\e}[0]{\epsilon}
\newcommand{\EE}{\ensuremath{\mathbb{E}}}
\newcommand{\qq}[1]{(q;q)_{#1}}
\newcommand{\A}{\ensuremath{\mathcal{A}}}
\newcommand{\GT}{\ensuremath{\mathbb{GT}}}
\newcommand{\link}{\ensuremath{Q}}
\newcommand{\PP}{\ensuremath{\mathbb{P}}}
\newcommand{\frakP}{\ensuremath{\mathfrak{P}}}
\newcommand{\frakQ}{\ensuremath{\mathfrak{Q}}}
\newcommand{\frakq}{\ensuremath{\mathfrak{q}}}
\newcommand{\R}{\ensuremath{\mathbb{R}}}
\newcommand{\Rplus}{\ensuremath{\mathbb{R}_{+}}}
\newcommand{\C}{\ensuremath{\mathbb{C}}}
\newcommand{\Z}{\ensuremath{\mathbb{Z}}}
\newcommand{\N}{\ensuremath{\mathbb{N}}}
\newcommand{\Weyl}[1]{\ensuremath{\mathbb{W}}^{#1}}
\newcommand{\Zgzero}{\ensuremath{\mathbb{Z}_{>0}}}
\newcommand{\Zgeqzero}{\ensuremath{\mathbb{Z}_{\geq 0}}}
\newcommand{\Zleqzero}{\ensuremath{\mathbb{Z}_{\leq 0}}}
\newcommand{\Q}{\ensuremath{\mathbb{Q}}}
\newcommand{\T}{\ensuremath{\mathbb{T}}}
\newcommand{\Y}{\ensuremath{\mathbb{Y}}}
\newcommand{\M}{\ensuremath{\mathbf{M}}}
\newcommand{\MM}{\ensuremath{\mathbf{MM}}}
\newcommand{\W}[1]{\ensuremath{\mathbf{W}}_{(#1)}}
\newcommand{\WM}[1]{\ensuremath{\mathbf{WM}}_{(#1)}}
\newcommand{\Zsd}{\ensuremath{\mathbf{Z}}}
\newcommand{\Fsd}{\ensuremath{\mathbf{F}}}
\newcommand{\symBM}{\ensuremath{\mathbf{W}}}
\newcommand{\symFE}{\ensuremath{\mathbf{S}}}

\newcommand{\const}{\ensuremath{\textrm{const}}}

\newcommand{\I}{\ensuremath{\mathbf{i}}}

\newcommand{\Real}{\ensuremath{\mathrm{Re}}}
\newcommand{\Imag}{\ensuremath{\mathrm{Im}}}
\newcommand{\re}{\ensuremath{\mathrm{Re}}}

\newcommand{\Sym}{\ensuremath{\mathrm{Sym}}}

\newcommand{\bfone}{\ensuremath{\mathbf{1}}}

\newcommand{\whitenoise}{\ensuremath{\mathscr{\dot{W}}}}
\newcommand{\alphaW}[1]{\ensuremath{\mathbf{\alpha W}}_{(#1)}}
\newcommand{\alphaWM}[1]{\ensuremath{\mathbf{\alpha WM}}_{(#1)}}
\newcommand{\malpha}{\ensuremath{\hat{\alpha}}}
\newcommand{\walpha}{\ensuremath{\alpha}}
\newcommand{\edge}{\textrm{edge}}
\newcommand{\dist}{\textrm{dist}}

\newcommand{\OO}[0]{\Omega}
\newcommand{\F}[0]{\mathfrak{F}}
\newcommand{\poly}[0]{R}
\def \Ai {{\rm Ai}}
\def \Pf {{\rm Pf}}
\def \sgn {{\rm sgn}}
\def \SS {\mathcal{S}}
\newcommand{\poles}{\mathbb{A}}
\def \ss {\mathcal{X}}
\newcommand{\var}{{\rm var}}

\def \leftp {\stackrel{+}{\longleftarrow}}
\def \leftm {\stackrel{-}{\longleftarrow}}
\def \leftpm {\stackrel{\pm}{\longleftarrow}}
\def \leftmp {\stackrel{\mp}{\longleftarrow}}
\def \rightp {\stackrel{+}{\longrightarrow}}
\def \rightm {\stackrel{-}{\longrightarrow}}
\def \rightpm {\stackrel{\pm}{\longrightarrow}}
\def \rightmp {\stackrel{\mp}{\longrightarrow}}

\newcommand{\Res}[1]{\underset{{#1}}{\mathbf{Res}}}
\newcommand{\Resq}[1]{\underset{{#1}}{\mathbf{Res}^q}}
\newcommand{\Resc}[1]{\underset{{#1}}{\mathbf{Res}^c}}
\newcommand{\Resfrac}[1]{\mathbf{Res}_{{#1}}}
\newcommand{\Sub}[1]{\underset{{#1}}{\mathbf{Sub}}}
\newcommand{\Subq}[1]{\underset{{#1}}{\mathbf{Sub}^q}}
\newcommand{\Subc}[1]{\underset{{#1}}{\mathbf{Sub}^c}}

\newcommand{\ul}[2]{\underline{#1}_{#2}}
\newcommand{\qhat}[1]{\widehat{#1}^{q}}
\newcommand{\La}[0]{\Lambda}
\newcommand{\la}[0]{\lambda}
\newcommand{\ta}[0]{\theta}
\newcommand{\w}[0]{\omega}
\newcommand{\ra}[0]{\rightarrow}
\newcommand{\vectoro}{\overline}
\newtheorem{theorem}{Theorem}[section]
\newtheorem{partialtheorem}{Partial Theorem}[section]
\newtheorem{conj}[theorem]{Conjecture}
\newtheorem{lemma}[theorem]{Lemma}
\newtheorem{proposition}[theorem]{Proposition}
\newtheorem{corollary}[theorem]{Corollary}
\newtheorem{claim}[theorem]{Claim}
\newtheorem{formal}[theorem]{Critical point derivation}
\newtheorem{experiment}[theorem]{Experimental Result}
\newtheorem{prop}{Proposition}

\def\todo#1{\marginpar{\raggedright\footnotesize #1}}
\def\change#1{{\color{green}\todo{change}#1}}
\def\note#1{\textup{\textsf{\Large\color{blue}(#1)}}}

\theoremstyle{definition}
\newtheorem{remark}[theorem]{Remark}

\theoremstyle{definition}
\newtheorem{example}[theorem]{Example}

\theoremstyle{definition}
\newtheorem{definition}[theorem]{Definition}

\theoremstyle{definition}
\newtheorem{definitions}[theorem]{Definitions}

%-----------------------

\newcommand{\FF}{\mathcal R}
\newcommand{\ga}{\gamma}
\newcommand{\ep}{\epsilon}
\newcommand{\s}{\mathcal S}
\newcommand{\sss}{\mathsf{set}}
\newcommand{\p}{\mathsf{prev}}
\newcommand{\aRp}{\xrightarrow[]{+}}
\newcommand{\aRm}{\xrightarrow[]{-}}
\newcommand{\aLp}{\xleftarrow[]{+}}
\newcommand{\aLm}{\xleftarrow[]{-}}
\newcommand{\pp}{\mathcal P}
\newcommand{\q}{\mathcal Q}

\newcommand{\bj}{\mathbf j}

%------------------------------

\begin{abstract}
We study the partition function of two versions of the continuum directed polymer in $1+1$ dimension. In the full-space version, the polymer starts at the origin and is free to move transversally in $\R$, and in the half-space version, the polymer starts at the origin but is reflected at the origin and stays in $\R_{-}$. The partition functions solves the stochastic heat equation in full-space or half-space with mixed boundary condition at the origin; or equivalently the free energy satisfies the Kardar-Parisi-Zhang equation.

We derive exact formulas for the Laplace transforms of the partition functions. In the full-space this is expressed as a Fredholm determinant while in the half-space this is expressed as a Fredholm Pfaffian. Taking long-time asymptotics we show that the limiting free energy fluctuations scale with exponent $1/3$ and are given by the GUE and GSE Tracy-Widom distributions. These formulas come from summing divergent moment generating functions, hence are not mathematically justified.

The primary purpose of this work is to present a mathematical perspective on the polymer replica method which is used to derive these results. In contrast to other replica method work, we do not appeal directly to the Bethe ansatz for the Lieb-Liniger model but rather utilize nested contour integral formulas for moments as well as their residue expansions.
\end{abstract}

\title{Directed random polymers via nested contour integrals}
\author[A. Borodin]{Alexei Borodin}
\address{A. Borodin,
Massachusetts Institute of Technology,
Department of Mathematics,
77 Massachusetts Avenue, Cambridge, MA 02139-4307, USA, and Institute for Information Transmission Problems, Bolshoy Karetny per. 19, Moscow 127994, Russia}
\email{borodin@math.mit.edu}

\author[A. Bufetov]{Alexey Bufetov}
\address{A. Bufetov,
Massachusetts Institute of Technology,
Department of Mathematics,
77 Massachusetts Avenue, Cambridge, MA 02139-4307, USA, and International Laboratory of Representation Theory and Mathematical Physics, National Research University Higher School of Economics, Moscow, Russia}
\email{alexey.bufetov@gmail.com}

\author[I. Corwin]{Ivan Corwin}
\address{I. Corwin, Columbia University,
Department of Mathematics,
2990 Broadway,
New York, NY 10027, USA,
and Clay Mathematics Institute, 10 Memorial Blvd. Suite 902, Providence, RI 02903, USA}
\email{ivan.corwin@gmail.com}

\maketitle

\setcounter{tocdepth}{2}
\tableofcontents
\hypersetup{linktocpage}

\section{Introduction}
The replica method for studying directed polymers, pioneered by Kardar in 1987 \cite{K}, has received a great deal of attention recently (cf. \cite{Dot,CDR,CDprl,ProS1,ProS2,ProSpoComp,ImSa,ImSaKPZ,CDlong,ImSaKPZ2,Dot2,Dot3,Dot4,Dot5,ISS}). In this paper we present a version of this method which ultimately leads to both the GUE and GSE Tracy-Widom distributions (for full-space and half-space polymers, respectively). Our aim is primarily pedagogical. We strive for mathematical clarity and do away with certain assumptions used in previous works (such as completeness of the Bethe ansatz and the evaluation of eigenfunction norms).

In this introduction we briefly highlight the central feature of our version of the polymer replica method -- the use of nested contour integrals and their residue expansions (rather than direct application of the Bethe ansatz) -- and briefly explain its connection to earlier work. The final result of our computations provides formulas characterizing the distributions of the partition function of full-space and half-space continuum directed random polymers.

We have attempted to provide computations which are as mathematically sound as possible (without getting into too many technicalities). Thus, up to a few technical assumptions on uniqueness of solutions to the delta Bose gas, our computations of the polymer partition function moments are mathematically rigorous. In the course of this computation all stated lemmas and propositions are accompanied by proofs.

Unfortunately, the moments of the partition function {\it do not} uniquely characterize the one-point distribution. Therefore, our final step of recovering the partition function's Laplace transform from the (divergent) moment generating function is, at its heart, unjustifiable. In fact, this issue plagues all of the aforementioned replica method works.  We proceed nevertheless and by an illegal application of a Mellin-Barnes summation trick we convert this divergent generating function into a convergent series of integrals. In the case of the full-space polymer this series can be matched to the formula derived independently and in parallel in \cite{ACQ,SaSp,Dot,CDR} and proved in \cite{ACQ} (see also \cite{BCF} for a second proof). In the half-space polymer there are no corresponding rigorous results yet with which to compare the answer. The replica method work of \cite{LD} (see also \cite{LDexpand}) also deals with the half-space polymer in the limit as our parameter $a$ characterizing the interaction of the polymer with the wall goes to infinity (the wall becomes absorbing). We do not take this limit and hence do not compare our results to those derived in \cite{LD}. However, let us remark that in \cite{LD}, the authors recover the predicted GSE Tracy-Widom distribution.

We believe that the most convincing argument for why this unjustified procedure produces the correct answer is that it is a shadow of a totally parallel procedure which can be performed rigorously on a suitable $q$-deformed regularization of our present model. For the full-space case this has been done in \cite{BorCor,BCS, BorCordiscrete, CorPetpush,Corhahn,CSS, CorPetdual, BarCorHahn, BarCorBeta}. Taking a suitable $q\to 1$ limit of the final formulas from these works provides a rigorous derivation of the Laplace transform formula we non-rigorously derive herein. A parallel treatment of the half-space case has not yet been performed (see Section \ref{mathrig} for more on rigorous mathematical work related to this paper).

\subsection{Nested contour integrals}

The polymer replica method relies upon the fact that the joint moments (at a fixed time $t$ and different spatial locations $x_1,\ldots, x_k$) of the partition function for the directed polymer (introduced in Section \ref{modelssec}) satisfy certain closed systems of evolution equations (see Section \ref{mappingsec}). These systems go by the name of the delta Bose gas or the Lieb-Liniger model with two-body delta interaction (we will use both names interchangeably), and variants of them hold in relation to both the full-space and half-space polymers. They are known to be integrable, which means that solving them can be reduced to solving a system of $k$ free one-body evolution equations subject to $k-1$ two-body boundary conditions (see Definitions \ref{Akmbs} and \ref{Bkmbs} for these systems).

The typical approach employed to solve the free evolution equation with $k-1$ two-body boundary conditions is to try to diagonalize the system (for instance, via the Bethe ansatz). The Bethe ansatz produces eigenfunctions, but does not a priori provide the knowledge of the relevant subspace of eigenfunctions on which to decompose the initial data as well as the knowledge of the norms of the eigenfunctions. Though we remark more on this approach below in Section \ref{relationBAsec}, it is not the route we follow. Instead, for delta function initial data we directly solve the system via a single $k$-fold nested contour integral (see Lemmas \ref{solnlemma2} and \ref{solnlemma} in Section \ref{solnsec}). This solution is easily checked to satisfy the desired system by simple residue calculus.

The type of result one hopes to get from the spectral approach is recovered by deforming the nested contours in these formulas to all coincide with the same contour $\I \R$. The eigenfunctions, their norms and the relevant subspace all come immediately out of the residue expansion coming from the poles crossed during these contour deformations (see Section \ref{relationBAsec} for more on this). %Towards this end, let us record the joint moment formula (or equivalently the solutions to the free evolution equation with $k-1$ boundary conditions) for both the full-space and half-space polymer.

In Section \ref{modelssec} we introduce the full-space continuum directed random polymer partition function $Z(t,x)$ as well as the half-space analog $Z^a(t,x)$ where $a$ determines an energetic cost/rewards for the polymer paths reflected at the origin ($a>0$ corresponds to an energetic cost or repulsive interaction and $a<0$ corresponds to an energetic reward or attractive interaction). We define the joint moments for these partition functions as
\begin{equation*}
\bar{Z}(t;\vec{x}) :=\EE\big[Z(t,x_1)\cdots Z(t,x_k)\big],   \qquad\qquad \bar{Z}^a(t;\vec{x}) :=\EE\big[Z^a(t,x_1)\cdots Z^a(t,x_k)\big]
\end{equation*}
where $\vec{x}= (x_1,\ldots, x_k)$ is assumed to be ordered as $x_1\leq \cdots \leq x_k$ for the full-space case and as $x_1\leq \cdots \leq x_k\leq 0$ for the half-space case.

We record equations (\ref{Aknci}) and (\ref{Bknci}) which are nested contour integral formulas for these joint moments. In the full-space case we show that
\begin{equation}\label{Aknciintro}
\bar{Z}(t;\vec{x}) =  \int_{\alpha_1-\I \infty}^{\alpha_1+\I\infty} \frac{dz_1}{2\pi \I} \cdots \int_{\alpha_k-\I \infty}^{\alpha_k+\I\infty} \frac{dz_k}{2\pi \I} \prod_{1\leq A<B\leq k} \frac{z_A-z_B}{z_A-z_B-1} \, \prod_{j=1}^{k} e^{\frac{t}{2} z_j^2 + x_jz_j},
\end{equation}
where we assume that $\alpha_1>\alpha_2 + 1 > \alpha_3 + 2> \cdots > \alpha_k + (k-1)$. And in the half-space case we show that
\begin{equation}\label{Bknciintro}
\bar{Z}^{a}(t;\vec{x}) = 2^k\int_{\alpha_1-\I \infty}^{\alpha_1+\I\infty} \frac{dz_1}{2\pi \I} \cdots \int_{\alpha_k-\I \infty}^{\alpha_k+\I\infty} \frac{dz_k}{2\pi \I} \prod_{1\leq A<B\leq k} \frac{z_A-z_B}{z_A-z_B-1}\, \frac{z_A+z_B}{z_A+z_B-1} \, \prod_{j=1}^{k} e^{\frac{t}{2} z_j^2 + x_jz_j}\frac{z_j}{z_j+a},
\end{equation}
where we assume that $\alpha_1>\alpha_2 + 1 > \alpha_3 + 2> \cdots > \alpha_k + (k-1)$ and $\alpha_k=\max(-a+\e,0)$ for $\e>0$ arbitrary. 

We call the above expressions nested contour integrals since the contours respect a certain infinite version of nesting (so as to avoid poles coming from the denominator). In the full-space case this formula seems to have first appeared in 1985 work of Yudson \cite{Yudson}. As solutions of the delta Bose gas (or Yang's system) with general type root systems (the above formulas correspond with type $A$ and type $BC$ root systems, respectively) such formulas appeared in 1997 work of Heckman-Opdam \cite{HO}.

Using Cauchy's theorem and the residue theorem we may deform the contours in both expressions until they all coincide with $\I \R$. In the course of these deformations we encounter first-order poles coming from the terms $z_A-z_B-1$ (and in the half-space case, also $z_A+z_B-1$ and $z_j+a$ participate) in the denominator. It is this expansion into residue subspaces which replaces the spectral decomposition (see Section \ref{relationBAsec}). For the half-space case, when $a=0$ the faction $\frac{z}{z+a} = 1$ which considerably simplifies the analysis versus $a\neq 0$. We do not presently attempt to work out the $a\neq 0$ residue expansion (see Remark \ref{nota}).

In order to develop the expansion of our nested contour formulas (in both the full-space and half-space cases, as well as in various $q$-deformed cases -- cf. Proposition \ref{321}) there are three parallel steps:

\noindent {\bf Step 1:} We identify the residual subspaces which arise in such an expansion. In terms of a meta-formula (letting NCI represent the words ``Nested contour integral''):
\begin{equation*}
\textrm{NCI} = \sum_{\substack{I\in \textrm{Residual}\\\textrm{subspaces}}} \int_{\I \R}\cdots\int_{\I \R} \Res{I}\big(\textrm{integrand}\big),
\end{equation*}
where $\Res{I}$ represents taking the residue along the residual subspace indexed by $I$, and where the integrals on the right-hand side are over the variables which remain after computing the residues. As an example, consider (\ref{Aknciintro}) for $k=2$. Then if we choose $\alpha_2=0$ from the start, we must deform the $z_1$ contour to $\I \R$. Regarding $z_2$ as fixed along $\I \R$, as we deform the $z_1$ contour we necessarily encounter a pole at $z_1=z_2+c$ and thus the nested contour integral is expanded into two terms -- one in which the integrals of $z_1$ and $z_2$ are both along $\I \R$ and the second in which only the $z_2$ integral remains and the integrand is replaced by the residue at $z_1=z_2+c$.

In general, due to the Vandermonde term in the numerator of the integrands we consider, the residual subspaces we encounter are indexed by certain strings of residues, as well as partitions $\lambda\vdash k$ which identify the sizes of the strings.

\noindent {\bf Step 2:} We show that via the action of the symmetric group (in the full-space case) or the hyperoctahedral group (in the half-space case) we can transform our residual subspaces into a canonical form, only indexed by a partition $\lambda\vdash k$. Moreover, even though only certain elements of these groups arise from such transformations, we readily check that all other group elements lead to zero residue contribution. Thus, we can rewrite our sum over residual subspaces as a sum over $\lambda\vdash k$ and the symmetric or hyperoctahedral group. Using $G$ to denote either of these groups, we arrive at our second meta-formula (we have suppressed certain constants arising from group symmetries)
\begin{equation*}
\textrm{NCI} = \sum_{\lambda\vdash k} \,\sum_{\sigma\in G}\, \int_{\I \R}\cdots\int_{\I \R} \Res{\lambda}\big(\sigma(\textrm{integrand})\big).
\end{equation*}
Now $\Res{\lambda}$ represents taking the residue along the canonical form residual subspace corresponding with $\lambda$.

\noindent {\bf Step 3:} Due to the form of the integrand we may rewrite it as a $G$-invariant function, times a remainder function which does not contain any of the poles presently relevant. This allows us to reach our final meta-formula:
\begin{equation}\label{astar}
\textrm{NCI} = \sum_{\lambda\vdash k} \,\int_{\I \R}\cdots\int_{\I \R} \Res{\lambda}\big(\textrm{$G$-invariant part}\big)\, \Sub{\lambda}\big(\sum_{\sigma\in G} \sigma(\textrm{remainder part})\big).
\end{equation}
Here $\Sub{\lambda}$ represents substitution or restriction of a function to the canonical form residual subspace corresponding with $\lambda$. In the cases we deal with in this paper, the residue term above can be explicitly evaluated, and the substitution term is generally left as is, though simplified considerably in certain cases (such as when the locations $x_i\equiv 0$).

In the full-space case, Proposition \ref{Akcprop} and its proof substantiate the above outlined steps. This type of residue expansion for the delta-Bose gas goes back at least to \cite{HO}, and this particular proposition is already present in \cite[Proposition 3.2.1]{BorCor} and \cite[Lemma 7.3]{BCPS1}. In the half-space case, we provide Conjecture \ref{Bkcprop} which explains what we believe to be the manner through which this expansion works. We additionally provide some evidence for the conjecture in Section \ref{sec:one-string}. This conjecture involves some rather subtle cancelations of residues which at first appear to complicate the situation. %It appears that in the study of BC-type Plancherel theory (see Section \ref{relationBAsec}), this sort of cancelation has also been observed, but has not yet been proved either \cite{OpdamPerComm}.

\subsection{Relation to Bethe ansatz, Plancherel theory, and Macdonald symmetric functions}\label{relationBAsec}

The Bethe ansatz goes back to Bethe's 1931 solution (i.e. diagonalization) of the spin $1/2$ Heisenberg XXX spin chain \cite{Bethe}. Lieb-Liniger diagaonlized their eponymous model (in the repulsive case) in 1963 utilizing this eigenfunction ansatz \cite{LL} (it appears that this was the first application of the ansatz after Bethe's original work and the approach was essentially rediscovered by Lieb-Liniger). Soon after, McGuire formulated the string hypotheses for the attractive case of the Lieb-Liniger model \cite{McGuire}. A great deal of work on this model (and its half-space variant) ensued, most notably in \cite{Gaudin, Yang1, Yang2} (see \cite{GutkinSuther} for a review of this early work).

The (nested) contour integrals for solutions of the Lieb-Liniger model with arbitrary initial conditions play a prominent role in \cite{HO}\footnote{Specifically, the Plancherel formula for the repulsive case is given therein as Theorem 1.3. Page 14 then explains why this formula holds in both repulsive and attractive cases, with nested contours in the attractive case. Theorem 3.13 is the residue computation when the contours come together, but it is not complete as the constants are not computed exactly (Heckman-Opdam did not need the exact values).}. Their form can be traced back to the classical works of Harish-Chandra in the 1950's on harmonic analysis on Riemannian symmetric spaces (see \cite{Hel} and references therein). In \cite{HO} they are used to prove the completeness of the Bethe ansatz eigenfunctions. The space of functions in which they work does not contain the $\delta$ initial condition with which this paper is concerned. However, their type of contour shifting arguments can be extended to this case well, and this extension is central for the present paper. In \cite{BCPS1,BCPS2}, $q$-deformed versions of the nested contour integrals are utilized to prove Plancherel theories for more general classes of eigenfunctions (related to the $q$-Boson and higher-spin six vertex models).

The connection to the Bethe ansatz can also be seen in the formulas above: the $\Sub{\lambda}$ part of (\ref{astar}) (denoted $E_{\vec{x}}^c$ later in the text) turns out to be an eigenfunction of the rewritten Lieb-Liniger model Hamiltonian (see Definitions \ref{Akmbs} and \ref{Bkmbs}), while the $\Res{\lambda}$ term incorporates the inverse squared norm of the eigenfunction and its pairing with the initial condition. Our choice of the $\delta$ initial condition yields simple expressions that eventually allow us to access the asymptotics we need.

In earlier work on Macdonald processes \cite{BorCor}, $(q,t)$-deformed versions of these nested contour integrals arose and independently led to the formulas above (in the full-space case). In \cite{BorCor} these formulas encoded the application of the Macdonald (first) difference operators to multiplicative functions. The reasons why iterating Macdonald difference operators is related to the same nested integrals as arise in \cite{HO} are still unclear to us (see however some exploration into this in \cite[Lemma 6.1]{BorCordiscrete}. %. The connection to the Lieb-Liniger model is partially explored in \cite{BorCordiscrete}. %, though it remains to develop the analog of \cite{HO} at this level (see, however \cite{SasWad,Pov,Tak}).

\subsection{Recent rigorous mathematical work on positive temperature directed polymers}\label{mathrig}

Even though the replica method for the continuum directed polymer falls short on mathematical rigor, it can be put on a solid mathematical ground by appealing to certain discretizations which preserve the model's Bethe ansatz solvability. Let us work for the moment just with full-space systems. Presently, the most general class of such discretizations are known of as the {\it higher-spin vertex models} (see \cite{BorodinR,CorPetdual} for more details). These stochastic interacting particle systems are solved in \cite{CorPetdual} by an analog of the replica method known of as {\it Markov dualities} (see \cite{BorodinPetrovhigher} for a different approach than duality to study these systems). There are many degenerations of these higher-spin vertex models and the duality / replica method applies to all of them. These degenerations include interacting particle systems (such as the stochastic six vertex model \cite{BCG}, ASEP \cite{IS, BCS}, Brownian motions with skew reflection \cite{SaSpBM}, $q$-Hahn TASEP \cite{Pov,Corhahn,BarCorHahn}, discrete time $q$-TASEP \cite{BorCordiscrete}, $q$-TASEP \cite{BCS}, $q$-pushASEP \cite{CorPetpush}) and directed polymer models (such as the Beta polymer \cite{BarCorBeta}, inverse Beta polymer \cite{LeDoussalThierry2}, log-gamma polymer \cite{LeDoussalThierry1}, strict-weak polymer \cite{CSS}, semi-discrete Brownian polymer \cite{BCS}). In some cases  the moments do determine the distribution and the duality / replica method can be rigorously performed, while other cases suffer a similar fate in that higher moments growth too fast, or even become infinite. However, once the distribution (or Laplace transform) is computed rigorously, if the model converges to the continuum polymer (as shown in various cases \cite{BG,ACQ, AKQ, CorTsai}) then taking a limit of the distribution function or Laplace transform provides a rigorous derivation of the continuum formula.

The approach taken in \cite{TW1,TW2,TW3} by Tracy and Widom in studying ASEP involves using Bethe ansatz to directly study transition probabilities and eventually extra marginal distributions from these formulas. While this approach bares some similarity to the duality / replica method, it does not have a clear degeneration to, for instance, the level of the continuum directed polymer. This work did, however, provide the first means to rigorously study the continuum directed polymer \cite{ACQ}. Besides the duality / replica method and Tracy and Widom's approach, the Macdonald process \cite{BorCor} approach has also proved quite fruitful in providing rigorous results on exact formulas for directed polymer models -- for further references, see the review \cite{ICICM} and reference therein. There exist other probabilistic methods which provide scaling exponents (though not exact distributions) for directed polymer models -- see for example, \cite{SeppLog} in the context of the log-gamma polymer.

For half-space systems, much less has been done. Notably, for ASEP with finitely many particles, transition probabilities have been computed \cite{TWhalfspace}, but no one-point marginal distribution formulas have been extracted. For the half-space log-gamma polymer, an analysis (analogous to the full-space case results of \cite{OCon,COSZ}) has been undertaken in \cite{OSZ}. A Laplace transform formula is conjectured therein, though does not seem to be readily accessible to any asymptotics. For zero-temperature polymers (i.e. last passage percolation) there are considerably more results as explained in Section \ref{LPP}.

\subsection{Conventions}\label{partsec}
We write $\Z_{\geq 0} = \{0,1,\ldots\}$ and $\Z_{>0}=\{1,2,\ldots\}$. When we perform integrals along vertical complex contours (such as $\I \R$) we will always assume that the contour is slightly to the right of the specified real part (i.e. $\I\R+\e$) so as to avoid any poles on said contour.

A {\it partition} is a sequence $\lambda=(\lambda_1,\lambda_2,\ldots)$ of non-negative integers with finitely many nonzero entries, such that $\lambda_1\geq\lambda_2\geq \cdots\geq 0$. The {\it length} $\ell(\lambda)$ is the number of non-zero $\lambda_i$ in $\lambda$ and the weight $|\lambda| = \lambda_1+\lambda_2+\cdots$. If $|\lambda|=k$ then we say that $\lambda$ {\it partitions} $k$, which is written $\lambda\vdash k$.  An alternative notation for $\lambda$ is $\lambda= 1^{m_1}2^{m_2}\cdots$ where $m_i$ represents the multiplicity of $i$ in the partition $\lambda$.

Let us also define the $q$-Pochhammer symbol (cf. Chapter 10 of \cite{AAR}) as $(a;q)_n = (1-a)(1-qa) \cdots (1-q^{n-1}a)$ and the Pochhammer symbol (or rising factorial\footnote{While this is the standard notation, we note that Wikipedia records $(a)_n$ as the falling factorial.}) as $(a)_{n} = a(a+1)\cdots (a+n-1)$.

\subsection{Outline}
In Section \ref{modres} we introduce our polymer models (along with some background) and state our main results. In Section \ref{mappingsec} we explain how the problem of computing the joint moments of the polymer partition function is mapped to a many body system. In Section \ref{solnsec} we solve that many body system via nested contour integral formulas. In Section \ref{expsec} we show how these formulas expand into residues when the nested contours are taken together to $\I \R$ (this is a conjecture in the half-space case). In Section \ref{mgfsec} we use the expansion of the moment formulas from earlier to compute the Laplace transform of the full-space and half-space polymer partition functions. We also perform long-time asymptotics and in the half-space case explain how we recognize the GSE Tracy-Widom distribution. Finally, in Section \ref{proofsec} we provide proofs of the residue expansion formulas, with the half-space conjecture being proved modulo Claim \ref{claim:Res-are-zero}.  In Section \ref{partsec} we provide notions concerning partitions and in Section \ref{Asymsec}, we provide background converning the symmetric and hyperoctahedral groups. Evidence for Claim \ref{claim:Res-are-zero} is subsequently provided in Section \ref{sec:one-string}.

\subsection{Acknowledgements}
The authors have benefited from discussions with Pierre Le Doussal, Herbert Spohn and Jeremy Quastel during the Simons Symposium on the KPZ equation. We especially thank Le Doussal for further discussions regarding the Bethe ansatz bound states. A. Borodin was partially supported by the NSF grant DMS-1056390. A. Bufetov was partially supported by the RFBR grant 13-01-12449, and by the Government of the Russian Federation within the framework of the implementation of the 5-100 Programme Roadmap of the National Research University Higher School of Economics.
I. Corwin was partially supported by the NSF through grant DMS-1208998 as well as by Microsoft Research through the Schramm Memorial Fellowship, by the Clay Mathematics Institute through the Clay Research Fellowship, by the Institute Henri Poincare through the Poincare Chair, and by the Packard Foundation through a Packard Fellowship in Science and Engineering.

\section{Models and results}\label{modres}

Models of directed polymers in disordered media in $1+1$ dimensions provide a unified framework for studying a variety of physical and mathematical systems as well as serve as a paradigm in the general study of disordered systems. Physically they provide models of domain walls of Ising type models with impurities \cite{HuHe,LFC}, vortices in superconductors \cite{BFG}, roughness of crack interfaces \cite{HHR}, Burgers turbulence \cite{FNS}, and interfaces in competing bacterial colonies \cite{HHRN} (see also the reviews \cite{HHZ} or \cite{FH} for more applications). Mathematically, they are closely related to stochastic (partial) differential equations \cite{BQS,ACQ}, stochastic optimization problems (including important problems in bio-statistics \cite{SM,HL,MMN,SAY} and operations research \cite{BBSSS}), branching Markov processes in random environments \cite{Bra}, as well as certain aspects of integrable systems and combinatorics (which will be discussed below). The free energy of the polymer partition function is related to the Kardar-Parisi-Zhang (KPZ) equation \cite{KPZ} whose spatial derivative is the stochastic Burgers equation \cite{FNS}. These are representatives of a large universality class of growth models and interacting particle systems (see the review \cite{ICreview}).

The physical motivation of the present work is to study the probability distribution of the free energy (logarithm of the partition function) of the continuum directed polymer in the full-space, as well as in the presence of a hard reflecting / absorbing wall. We compute the Laplace transform (hence also the probability distribution) of the partition function in both geometries (though for the half-space case we only work with $a=0$). Taking a large time limit we show that in both geometries the free energy fluctuations scale with exponent $1/3$, though they display different limiting probability distributions -- the GUE (written $F_{{\rm GUE}}$ or $F_2$) versus GSE (written $F_{{\rm GSE}}$ or $F_4$) Tracy-Widom distributions. In the context of last passage percolation with full and half-space geometries, these scalings and distributions first arose in the work of Johansson \cite{KJ} and Baik-Rains \cite{BaikRains}.

\subsection{The models}\label{modelssec}

Let $\xi(t,x)$ represent a Gaussian space-time white noise in $1+1$ dimension with covariance $\left\langle \xi(s,y)\xi(t,x)\right\rangle = \delta(s-t)\delta(y-x)$ (this is understood in the sense of a generalized function -- see \cite{ACQ,W} for more details) and let $\EE$ represent the expectation with respect to this random noise. We define the full-space and half-space polymer partition functions via the expectation of the exponential of path integrals through this noise field. This definition is properly made sense of either through smoothing the noise or through a chaos series expansion. To indicate either one of these (equivalent) renormalization procedures we use the notation of a Wick exponential $:\exp:$. The polymer partition function may equivalently be defined as the solution to a stochastic heat equation with suitable symmetry constraint in the half-space case\footnote{From a mathematical perspective, the equivalence of these two formulations has only been shown in the full-space case \cite{BC}, however, one expects the same methods to apply in the half-space case.}.

\subsubsection{Full-space polymer}

Let $\mathcal{E}$ represent the expectation of a one dimensional Brownian motion $b(\cdot)$ with $b(0)=0$.  Define the {\it full-space continuum directed random polymer partition function} as
\begin{equation*}
Z(t,x) = \mathcal{E}\left[:\exp:\left\{\int_0^t \xi(s,b(s))ds \right\} \delta(b(t)=x) \right].
\end{equation*}
See Section \ref{fspsec} or \cite{ACQ} for details on how to make mathematical sense of this definition.

Equivalently, $Z(t,x)$ may be defined as the solution to the (well-posed) stochastic heat equation on $\R$ with multiplicative noise and delta initial data
\begin{equation*}
\frac{d}{dt} Z(t,x) = \frac{1}{2} \frac{d^2}{dx^2} Z(t,x) + \xi(t,x) Z(t,x), \qquad Z(0,x) = \delta_{x=0}.
\end{equation*}

\subsubsection{Half-space polymer}

Let $\mathcal{E}^{R}$ represent the expectation of a one dimensional Brownian motion $b(\cdot)$ reflected so as to stay on the left of the origin and started from $b(0)=0$. Then for all $a\in \R$, define the {\it half-space continuum directed random polymer partition function with parameter $a$} as
\begin{equation*}
Z^{a}(t,x) = \mathcal{E}^R\left[:\exp:\left\{\int_0^t \big(\xi(s,b(s)) - a \delta_{b(s)=0}\big)ds \right\} \delta(b(t)=x)\right].
\end{equation*}
The integral of $\delta_{b(s)=0}$ is the local time of the Brownian motion at the origin.

Equivalently, $Z^a(t,x)$ may be defined as the solution to the (well-posed) stochastic heat equation on $(-\infty,0)$ with Robin (i.e. mixed Dirichlet and von Neumann) boundary condition at the origin
\begin{equation*}
\frac{d}{dt} Z^a(t,x) = \frac{1}{2} \frac{d^2}{dx^2} Z^a(t,x) + \xi(t,x) Z^a(t,x), \qquad \left(\frac{d}{dx}+a\right) Z^a(t,x)\big\vert_{x\to 0^-} = 0, \qquad Z^a(0,x) = \delta_{x=0}.
\end{equation*}
The initial data above means that for continuous bounded test functions $f:(-\infty,0)\to \R$,
$$\lim_{t\to 0} \int_{-\infty}^{0} f(x) Z^a(t,x) dx = f(0).$$

The behavior of this partition function should depend on the sign of $a$. When $a$ is positive, the polymer measure will favor paths which tend to avoid touching the origin -- something which can be done at relatively minor entropic cost. However, when $a$ is negative, the polymer paths will be rewarded for staying near the origin. Another way to see this difference is through a Feynman-Kac representation (cf. \cite{Mol,BCLyapunovpaper}) in which when $a$ is positive, the Brownian paths are killed at a rate given by the local time multiplied by $a$, whereas when $a$ is negative, the Brownian paths duplicate at rate given by the local time multiplied by $|a|$.

This difference in path behavior should be reflected in the asymptotic behavior of the partition function. For $a\geq 0$, it is reasonable to expect that the asymptotic behavior of $\log Z^a$ should be independent of the magnitude of $a$. On the other hand, when $a<0$, one expects a localization transition in which the asymptotic growth behavior of $\log Z^a$ changes from that of $\log Z^0$.

In this article we derive moment formulas for $Z^a$ valid for all $a\in \R$, however we restrict our attention to the Laplace transform of the $a=0$ partition function only as $a\neq 0$ introduces some additional complications (see Remark \ref{nota}).

\subsection{Laplace transform formulas}\label{Laplacesec}

The output of the calculations we present are the following Laplace transform formulas. Though our method is not mathematically rigorous, the full-space formula matches the result proved in \cite{ACQ,BCF}. No such rigorous derivation exists for the half-space formula as of yet. The partition function in question are positive random variables\footnote{This fact is proved in \cite{M} and more recently in \cite{Gregorio} for the full-space case. A proof does not seem to exist presently for the half-space case, though one expects the methods should extend without much difficulty.}, hence their Laplace transforms uniquely characterize their distributions.

\subsubsection{Full-space polymer}\label{fullspacelapsec}
We summarize our calculation with the following result. For $\zeta\in \C$ with $\Real(\zeta)<0$
\begin{equation}\label{fullspacedet}
\EE\left[e^{\zeta Z(t,0)}\right] =  1 + \sum_{L=1}^{\infty} \frac{(-1)^L}{L!} \int_{0}^{\infty} dx_1\cdots \int_{0}^{\infty} dx_L \det\left[K_{\zeta}(x_i,x_j)\right]_{i,j=1}^{L}
\end{equation}
where the kernel is given by
\begin{equation*}
K_{\zeta}(x,x') =\int_{-\infty}^{\infty} dr \frac{1}{1 + \exp\left\{(\frac{t}{2})^{1/3}(r + u)\right\}} \Ai(x-r)\, \Ai(x'-r),
\end{equation*}
and where $\zeta$ and $u$ are related according to
\begin{equation*}
\log(-\zeta) = -\left(\frac{t}{2}\right)^{1/3} u + \frac{t}{24}.
\end{equation*}

The right-hand side of (\ref{fullspacedet}) may be written in terms of a Fredholm determinant as
$$\det(I-K_{\zeta})_{L^2([0,\infty))}$$
where the integral operator $K_{\zeta}$ acts on $L^2([0,\infty))$ via its kernel $K_{\zeta}(x,x')$ given above.

As pointed out in \cite{ACQ}, by invariance of space-time white noise under affine shifts, the marginal distribution of $Z(t,x)$ is equal to the marginal distribution of $p(t,x) Z(t,0)$. Here $p(t,x) = (2\pi)^{-1/2} e^{-x^2/2t}$ is the standard heat kernel.

The distribution of $Z(t,0)$ was first discovered independently and in parallel in \cite{ACQ,SaSp,Dot,CDR} with a mathematically rigorous proof of the formula given in \cite{ACQ} (and later a different proof given in \cite{BCF}). The formula above agrees with the rigorously proved result.

\subsubsection{Half-space polymer}\label{halfspacelapsec}
Our half-space polymer calculation for $a=0$ similarly leads to the following result. (It should be noted that along the way, we utilize a residue expansion formula of which we do not have a complete proof. This expansion formula is stated as Conjecture \ref{Bkcprop}.) For $\zeta\in \C$ with $\Real(\zeta)<0$,

\begin{equation}\label{halfspacedet}
\EE\left[e^{\frac{\zeta}{4} Z^0(t,0)}\right] =  1 + \sum_{L=1}^{\infty}\frac{(-1)^L}{L!} \int_{-\infty}^{\infty} dr_1\cdots \int_{-\infty}^{\infty} dr_L \prod_{k=1}^{L} \frac{\zeta}{e^{r_j}-\zeta} \Pf\left[ K(r_i,r_j)\right]_{i,j=1}^{L}
\end{equation}
where the above Pfaffian is of a $2L\times 2L$ sized matrix composed of $2\times 2$ blocks $K^{(t)}(r,r')$ with components
\begin{eqnarray*}
K^{(t)}_{11}(r,r') &=& \frac{1}{4} \int_0^{\infty} dx \int_{-\tfrac{1}{2} -\I\infty}^{-\tfrac{1}{2}+\I\infty} \frac{dw_1}{2\pi \I} \int_{-\tfrac{1}{2} -\I\infty}^{-\tfrac{1}{2}+\I\infty} \frac{dw_2}{2\pi \I} \frac{-w_1+w_2}{w_1 w_2} \frac{1}{F^{(t)}(w_1) F^{(t)}(w_2)} e^{-r w_1-r'w_2} e^{xw_1+xw_2},\\
K^{(t)}_{12}(r,r') &=& \frac{1}{4} \int_0^{\infty} dx \int_{-\tfrac{1}{2} -\I\infty}^{-\tfrac{1}{2}+\I\infty} \frac{dw}{2\pi \I} \int_{\tfrac{1}{4} -\I\infty}^{\tfrac{1}{4}+\I\infty} \frac{ds}{2\pi \I} \frac{-w-s}{w} \frac{F^{(t)}(s)}{F^{(t)}(w)} e^{-r w+r's} e^{xw-xs},\\
K^{(t)}_{22}(r,r') &=& \frac{1}{4} \int_0^{\infty} dx \int_{\tfrac{1}{4} -\I\infty}^{\tfrac{1}{4}+\I\infty} \frac{ds_1}{2\pi \I}\int_{\tfrac{1}{4} -\I\infty}^{\tfrac{1}{4}+\I\infty} \frac{ds_2}{2\pi \I} (s_1-s_2) F^{(t)}(s_1)F^{(t)}(s_2) e^{r s_1+r's_2} e^{-xs_1-xs_2},
\end{eqnarray*}
and $K^{(t)}_{21}(r,r') = - K^{(t)}_{12}(r',r)$. In the above we have used
\begin{equation*}
F^{(t)}(w)  =\frac{\Gamma(w)}{\Gamma(w+\frac{1}{2})} e^{\frac{t}{2} \left(\frac{s^3}{3} - \frac{s}{12}\right)}.
\end{equation*}

This formula can be manipulated further to appear closer to that of (\ref{fullspacedet}). However, unlike in the full-space case, there is no simple transformation which maps this $x=0$ result to a general $x$ result. Though almost all of our work applied to the general $x$ case, we presently do not have a suitable analog of the identity (\ref{Bkfacsym}) for $x\neq 0$.

%Note that the related work of \cite{LD} deals with the case of $a=\infty$ (i.e., pure absorbtion at the origin, as opposed to pure reflection).

\subsection{Long-time asymptotic distributions}\label{asydistsec}

\subsubsection{Full-space polymer}

We may rewrite (\ref{fullspacedet}) in a suggestive manner for taking the $t\to \infty$ asymptotics:

\begin{equation*}
\EE\left[  \exp\left\{-\exp\left\{ \left(\frac{t}{2}\right)^{1/3}\left[\frac{\log Z(t,0)  + \frac{t}{24}}{\left(\frac{t}{2}\right)^{1/3}} \, -u \right]\right\}\right\} \right]   =  \det(I-K_{u})_{L^2([0,\infty))}
\end{equation*}
with kernel
\begin{equation*}
K_{u}(x,x') =\int_{-\infty}^{\infty} dr \frac{1}{1 + \exp\left\{(\frac{t}{2})^{1/3}(r + u)\right\}} \Ai(x-r)\, \Ai(x'-r).
\end{equation*}
Since $e^{-e^{\lambda x}} \to \mathbf{1}_{x<0}$ as $\lambda\to +\infty$ we see that as $t\to \infty$, the left-hand side becomes the expectation of an indicator function (hence a probability), while the right-hand side also has a clear limit since as $t \to \infty$,
$$
\frac{1}{1 + \exp\left\{(\frac{t}{2})^{1/3}(r + u)\right\}} \to \mathbf{1}_{r+u<0}.
$$

The output of these observations is the following limit distribution result\footnote{These asymptotic calculations can be performed in an entirely mathematically rigorous manner  -- see \cite{ACQ,BCF}.}:
\begin{equation*}
\lim_{t\to \infty} \PP\left(\frac{\log Z(t,0)  + \frac{t}{24}}{\left(\frac{t}{2}\right)^{1/3}} \leq u \right)  = \det(I-K_{{\rm Ai}})_{L^2((u,\infty])} = F_{{\rm GUE}}(u)
\end{equation*}
where $F_{{\rm GUE}}(u)$ is the GUE Tracy-Widom distribution, given in terms of the above Fredholm determinant with
\begin{equation*}
K_{{\rm Ai}}(x,x') = \int_{0}^{\infty} dr \Ai(x+r)\Ai(x'+r).
\end{equation*}

\subsubsection{Half-space polymer}\label{hsplim}
Just as in the full-space Laplace transform, we may rewrite the left-hand side of (\ref{halfspacedet}) in a suggestive manner for taking the $t\to \infty$ asymptotics:
\begin{equation*}
\EE\left[  \exp\left\{-\frac{1}{4}\exp\left\{ \left(\frac{t}{2}\right)^{1/3}\left[\frac{\log Z^0(t,0)  + \frac{t}{24}}{\left(\frac{t}{2}\right)^{1/3}} \, -u \right]\right\}\right\} \right].
\end{equation*}
The above expression converges to the limiting probability distribution of $\log Z^0(t,0)$ (centered by $t/24$ and scaled by $(t/2)^{1/3}$) as $t\to \infty$. The calculation of Section \ref{RHSder} shows that the right-hand side of (\ref{halfspacedet}) has a clear limit as well. Combining these facts we find that
\begin{equation}\label{almostGSE}
\lim_{t\to \infty} \PP\left(\frac{\log Z^0(t,0)  + \frac{t}{24}}{\left(\frac{t}{2}\right)^{1/3}} \leq u \right)  = 1 + \sum_{L=1}^{\infty} \frac{(-1)^L}{L!} \int_{u}^{\infty} dr_1 \cdots \int_{u}^{\infty} dr_{L} \Pf\left[K^{\infty}(r_i,r_j)\right]_{i,j=1}^{L}
\end{equation}
where the above Pfaffian is of a $2L\times 2L$ sized matrix composed of $2\times 2$ blocks $K^{\infty}(r,r')$ with components
\begin{eqnarray}\label{Kkern}
\nonumber K^{\infty}_{11}(r,r') &=& \frac{1}{4} \int_0^{\infty} dx \int_{-\tfrac{1}{2} -\I\infty}^{-\tfrac{1}{2}+\I\infty} \frac{dw_1}{2\pi \I} \int_{-\tfrac{1}{2} -\I\infty}^{-\tfrac{1}{2}+\I\infty} \frac{dw_2}{2\pi \I} (w_2-w_1) e^{-\frac{w_1^3}{3} + w_1(r+x)} e^{-\frac{w_2^3}{3}+w_2(r'+x)},\\
K^{\infty}_{12}(r,r') &=& \frac{1}{4} \int_0^{\infty} dx \int_{-\tfrac{1}{2} -\I\infty}^{-\tfrac{1}{2}+\I\infty} \frac{dw}{2\pi \I} \int_{\tfrac{1}{4} -\I\infty}^{\tfrac{1}{4}+\I\infty} \frac{ds}{2\pi \I} \frac{w+s}{s} e^{-\frac{w^3}{3} + w(r+x)} e^{\frac{s^3}{3} - s(r'+x)},\\
\nonumber K^{\infty}_{22}(r,r') &=& \frac{1}{4} \int_0^{\infty} dx \int_{\tfrac{1}{4} -\I\infty}^{\tfrac{1}{4}+\I\infty} \frac{ds_1}{2\pi \I}\int_{\tfrac{1}{4} -\I\infty}^{\tfrac{1}{4}+\I\infty} \frac{ds_2}{2\pi \I} \frac{s_1-s_2}{s_1s_2} e^{\frac{s_1^3}{3}-s_1(r+x)} e^{\frac{s_2^3}{3} - s_2(r'+x)},
\end{eqnarray}
and $K^{\infty}_{21}(r,r') = - K^{\infty}_{12}(r',r)$.
It is shown in Section \ref{RHSdist} that this Fredholm Pfaffian is equal to the GSE Tracy-Widom distribution. Hence we conclude that
\begin{equation*}
\lim_{t\to \infty} \PP\left(\frac{\log Z^0(t,0)  + \frac{t}{24}}{\left(\frac{t}{2}\right)^{1/3}} \leq u \right)  = F_{{\rm GSE}}(u).
\end{equation*}

We expect that as long as the interaction at the origin is repulsive (i.e., $a\geq 0$), its strength does not effect the asymptotic free energy fluctuations. We do not, however, include an analysis of the general $a>0$ case herein due to extra complications which arise in its analysis (see Remark \ref{nota}).

\subsection{Comparison to Johansson and Baik-Rains last passage percolation work}\label{LPP}
Last passage percolation (LPP) is a discrete time and space zero temperature polymer. For geometric weight distributions, the model is exactly solvable via methods of determinantal and Pfaffian point processes. We recover a full-space result due to Johansson \cite{KJ} and a corresponding half-space result due to Baik-Rains \cite{BaikRains}.

For the full-space set up, let $w_{n,x}$ (with $n\in \Z_{\geq 0}$ and $x\in \Z$) be independent geometric random variables with parameter $p\in (0,1)$. Let $\pi(\cdot)$ represent the trajectory of a simple symmetric random walk started from $\pi(0)=0$. Then define the full-space last passage time
$$
L(N) = \max_{\pi:\pi(2N)=0} \left(\sum_{n=0}^{2N} w_{n,\pi(n)}\right).
$$
For the half-space set up, let $w_{n,x}$ (with $n\in \Z_{\geq 0}$ and $x\in \Z_{< 0}$) be independent geometric random variables with parameter $p\in (0,1)$ and let $w_{n,0}$ be independent geometric random variables with parameter $\alpha \sqrt{p}\in (0,1)$. Then define the half-space last passage time with parameter $\alpha$ as
$$
L(N) = \max_{\pi:\pi(2N)=0, \pi(\cdot)\leq 0} \left(\sum_{n=0}^{2N} w_{n,\pi(n)}\right).
$$

The following theorem is due to Johansson \cite{KJ} in the full-space case and Baik-Rains \cite{BaikRains} in the half-space case.
\begin{theorem}
Setting
$$\eta(p) = \frac{2\sqrt{p}}{1-\sqrt{p}}, \qquad \rho(p) = \frac{p^{1/6} (1+\sqrt{p})^{1/3}}{1-\sqrt{p}},$$
then in the full-space case
$$
\lim_{N\to \infty} \PP\left(\frac{L(N)-\eta(p)N}{\rho(p) N^{1/3}}\leq u\right) = F_{{\rm GUE}}(u)
$$
and in the half space case
$$
\lim_{N\to \infty} \PP\left(\frac{L(N)-\eta(p)N}{\rho(p) N^{1/3}}\leq u\right) =
\begin{cases}
F_{{\rm GSE}}(u) & 0\leq \alpha<1\\
F(u;w) & \alpha = 1 - \frac{2 w}{\rho(p) N^{1/3}}\\
0 & \alpha>1.
\end{cases}
$$
\end{theorem}
The definition of the distribution $F(u;w)$ above is given in \cite[Definition 2]{BaikRains}.

The full-space result agrees with our positive temperature analog. The half-space result likewise agrees, as we are in the regime of a repulsive origin (here presented by $\alpha<\sqrt{p}$). There is a regime of $\sqrt{p}\leq \alpha\leq 1$ when the origin provides extra reward over the other weights, but the limiting statistics remain unchanged (this corresponds with the path not being sufficiently rewarded for the entropic cost of staying localized near the origin), and a regime of $\alpha = 1 - \frac{2 w}{\rho(p) N^{1/3}}$ when the origin strength becomes critical. Beyond this, the maximizing path visits the origin a macroscopic proportion of the time and the reason for the 0 limit above is because the law of large number centering by $\eta(p)$ becomes insufficient. In this regime the law of large numbers increases and the fluctuations are Gaussian.

We expect the same sort of transition should occur in the half-space continuum polymer we consider herein. In order to see the critical behavior we must tune $a<0$ appropriately. Since we do not presently deal with asymptotics for this $a<0$ regime, we cannot yet access this behavior.

\section{Mapping to many body system}\label{mappingsec}
The basic observation explained in this section is that joint moments of the polymer partition function solve certain nice closed evolution equations. This at least goes back to the independent work of Kardar \cite{K} and Molchanov \cite{Mol}.% in the years 1986-87.

\subsection{Full-space polymer}\label{fspsec}

Consider $\vec{x} = (x_1,\ldots,x_k) \in \R^k$ and define
$$\bar{Z}(t;\vec{x}) = \EE\left[\prod_{i=1}^{k} Z(t,x_i)\right],$$
where recall that $Z(t;x)$ is the full-space polymer partition function and $\EE$ represents the expectation with respect to the white noise $\xi$. Then the basic observation one makes is that
\begin{equation*}
\frac{d}{dt} \bar{Z}(t;\vec{x}) = H^1 \bar{Z}(t;\vec{x}), \qquad \bar{Z}(0;\vec{x}) = \prod_{i=1}^{k} \delta(x_i=0)
\end{equation*}
where $H^c$, for general $c\in \R$, is given by
\begin{equation*}
H^c = \frac{1}{2} \sum_{i=1}^{k} \frac{d^2}{dx_i^2}  + \frac{c}{2} \sum_{i\neq j} \delta(x_i-x_j=0).
\end{equation*}
Such a statement also holds for solutions of the stochastic heat equation with generic initial data, with $\bar{Z}(0;\vec{x})$ modified accordingly.

The above system is the imaginary time Lieb-Liniger model \cite{LL} with two body delta interaction and coupling constant $-c$. It is also called a delta Bose gas, as one restricts to Bosonic solutions which are symmetric in relabeling the indices of $\vec{x}$. When the coupling constant $-c$ is negative (i.e. $c>0$) the model is called attractive, whereas when it is positive the model is called repulsive. The partition function moments correspond with the attractive case.

The fact that these moments satisfy the above system can be shown through smoothing the noise $\xi$ in space, replicating the path integrals defining the partition function (which now make sense), interchanging the $k$-fold replica expectation with the Gaussian expectation (which can now be evaluated in terms of local times) and finally taking away the smoothing. This procedure is performed in a mathematically rigorous manner in Bertini-Cancrini \cite[Proposition 2.3]{BC} where they show that

$$
\bar{Z}(t;\vec{x}) = \prod_{i=1}^{k} p(t,x_i)  \, \mathcal{E}^{k} \left[ \exp\left\{\frac{1}{2} \sum_{i\neq j} L_t(b_i-b_j)\right\}\right]
$$
where $p(t,x)$ is the transition probability for a Brownian motion to go from 0 to $x$ in time $t$, $\mathcal{E}^k$ is the expectation over $k$ independent Brownian bridges $b_1(\cdot),\ldots, b_k(\cdot)$ which start at 0 at time 0 and end at $x_1,\ldots, x_k$ (respectively), and $L_t(b_i-b_j)$ is the intersection local time of $b_i-b_j=0$ over the time interval $[0,t]$. From this formula and the Feynman-Kac representation, one readily checks that $\bar{Z}(t;\vec{x})$ satisfies the desired evolution equation above\footnote{We note, however, that we are not aware of a uniqueness result for this delta Bose gas which includes the delta initial data we consider.}.

Notice that $\bar{Z}(t;\vec{x}) = \bar{Z}(t;\sigma\vec{x})$ for any permutation $\sigma$ hence it suffices to compute $\bar{Z}(t;\vec{x})$ for $\vec{x} \in W(A_k)$. Here $W(A_k)$ is the type $A_k$ Weyl chamber $x_1\leq x_2\leq \cdots \leq x_k$ (see Section \ref{Asymsec}).

\begin{definition}\label{Akmbs}
We say that $u^c:\R_{\geq 0} \times \R^k\to \R$ solves the {\it type $A_k$ free evolution equation with $k-1$ boundary conditions and coupling constant $c$} if
\begin{enumerate}
\item For all $t>0$ and $\vec{x}\in \R^k$,
$$\frac{d}{dt} u^c(t;\vec{x}) = \frac{1}{2} \sum_{i=1}^{k} \frac{d^2}{dx_i^2} u^c(t;\vec{x});$$
\item For all $t>0$ and $\vec{x}\in W(A_k)$ such that $x_i=x_{i+1}$,
$$ \left(\frac{d}{dx_i} - \frac{d}{dx_{i+1}} -c\right) u^c(t;\vec{x}) = 0;$$
\item For all continuous bounded $L^2$ functions $f:W(A_k)\to \R$,
    $$ \lim_{t\to 0} k!\int_{W(A_k)}f(\vec{x}) u^c(t;\vec{x})  =f(0).$$
\end{enumerate}
\end{definition}
Using the local time representation for $\bar{Z}(t;\vec{x})$ one finds that restricted to $\vec{x}\in W(A_k)$, $$\bar{Z}(t;\vec{x})=u^1(t,\vec{x}).$$
While there is little doubt of the above equality, in fact (to our knowledge) this observation has not been made in a rigorous manner and also relies upon an assumption of uniqueness of solutions for the above system.

This reduction to the form of a free evolution equation with $k-1$ boundary conditions is a hallmark of integrability and will be the starting point for our analysis. In fact, all replica method works relies on this rewriting of the delta Bose gas. It is worth noting that if the two body delta interaction term is replaced by a smoothed version (as corresponds to the case of spatially smoothed noise) there is no reduction to such an integrable form. However, there does exist discrete space versions of the polymer as well as a $q$-deformed discrete version of the polymer for which this integrability persists (see, for example, \cite{BCS}). For these discretizations, the questions of uniqueness are easily overcome.

\subsection{Half-space polymer}

Consider $\vec{x} = (x_1,\ldots,x_k) \in \R^k$ and define
$$\bar{Z}^a(t;\vec{x}) = \EE\left[\prod_{i=1}^{k} Z^a(t,x_i)\right],$$
where recall that $Z^a(t;x)$ is the half-space polymer partition function with parameter $a$. Though this is only defined for $x_i< 0$, we can extend it to $x_i\in \R$ by setting declaring $Z^a(t,x)=Z^a(t,-x)$. Of course, this extended version of $Z^a(t,x)$ solves a stochastic heat equation on $\R$ (with the Robin jump condition  for the derivative at 0), but now with doubled initial data $\prod_{i=1}^{k} 2\delta(x_i=0)$. In fact, (up to the factor of two) this corresponds with considering a full-space polymer with noise $\xi$ which is symmetric through the origin (i.e., $\xi(t,x)=\xi(t,-x)$) and which has an energetic contribution from the polymer path local time at the origin.

Just as in the full-space polymer, these moments solve nice closed evolution equations:
\begin{equation*}
\frac{d}{dt} \bar{Z}^a(t;\vec{x}) = H^{1,a} \bar{Z}(t;\vec{x}), \qquad \bar{Z}^a(0;\vec{x}) = \prod_{i=1}^{k} 2 \delta(x_i=0),
\end{equation*}
where $H^{c,a}$, for general $c,a\in \R$, is given by
\begin{equation*}
H^{c,a} = \frac{1}{2} \sum_{i=1}^{k} \frac{d^2}{dx_i^2}  + \frac{c}{2} \sum_{i\neq j} \delta(x_i-x_j=0) - a \sum_{i=1}^{k} \delta(x_i=0).
\end{equation*}
One restricts to solutions which are invariant with respect to the action of hyperoctahedral  group $BC_k$ (see Section \ref{Asymsec}). Such an extension seems to have been first considered by Gaudin \cite{Gaudin} (see \cite{GutkinSuther,HO} for further developments).

The fact that these moments satisfy the above system can be shown in a similar way as for the full-space polymer. Doing this one finds that
$$
\bar{Z}(t;\vec{x}) = \prod_{i=1}^{k} p^R(t,x_i)  \, \mathcal{E}^{k} \left[ \exp\left\{\frac{1}{2} \sum_{i\neq j} L_t(b_i-b_j)- a \sum_{i=1}^{k} L_t(b_i)\right\}\right]
$$
where $p^R(t,x)$ is the transition probability for a reflected Brownian motion to go from 0 to $x$ in time $t$, $\mathcal{E}^k$ is the expectation over $k$ independent reflected Brownian bridges $b_1(\cdot),\ldots, b_k(\cdot)$ which start at 0 at time 0 and end at $x_1,\ldots, x_k$ (respectively), $L_t(b_i-b_j)$ is the intersection local time of $b_i(\cdot)-b_j(\cdot)=0$ over the time interval $[0,t]$, and $L_t(b_i)$ is the intersection local time of $b_i(\cdot)=0$ over the time interval $[0,t]$. From this formula and the Feynman-Kac representation, one readily checks that $\bar{Z}^a(t;\vec{x})$ satisfies the desired evolution equation above\footnote{We note, however, that we are not aware of a uniqueness result for this delta Bose gas which includes the delta initial data we consider.}.

Notice that $\bar{Z}^a(t;\vec{x}) = \bar{Z}(t;\sigma\vec{x})$ for any $\sigma\in BC_k$. Therefore it suffices to compute $\bar{Z}^a(t;\vec{x})$ for $\vec{x} \in W(BC_k)$. Here $W(BC_k)$ is the type $BC_k$ Weyl chamber $x_1\leq x_2\leq \cdots \leq x_k\leq 0$ (see Section \ref{Asymsec}).

\begin{definition}\label{Bkmbs}
We say that $u^{c,a}:\R_{\geq 0} \times \R^k\to \R$ solves the {\it type $BC_k$ free evolution equation with $k-1$ boundary conditions and coupling constants $c$ and $a$} if
\begin{enumerate}
\item For all $t>0$ and $\vec{x}\in \R^k$,
$$\frac{d}{dt} u^{c,a}(t;\vec{x}) = \frac{1}{2} \sum_{i=1}^{k} \frac{d^2}{dx_i^2} u^{c,a}(t;\vec{x});$$
\item For all $t>0$ and $\vec{x}\in W(BC_k)$ such that $x_i=x_{i+1}$,
$$ \left(\frac{d}{dx_i} - \frac{d}{dx_{i+1}} -c\right) u^{c,a}(t;\vec{x}) = 0;$$
\item For all $t>0$ and $\vec{x}\in W(BC_k)$ such that $x_k=0$,
$$ \left(\frac{d}{dx_k} + a \right) u^{c,a}(t;\vec{x}) = 0;$$
\item For all continuous bounded $L^2$ functions $f:W(BC_k)\to \R$,
    $$ \lim_{t\to 0} k!\int_{W(BC_k)}f(\vec{x}) u^{c,a}(t;\vec{x})  =f(0).$$
\end{enumerate}
\end{definition}
The above many body system is called Yang's system in the work of Heckman-Opdam \cite{HO} due to Yang's work in the late 1960's \cite{Yang1,Yang2}.

Just as in the full-space case, we are not aware of a proof of uniqueness results for the above system. However, assuming this and using the local time representation for $\bar{Z}^a(t;\vec{x})$ one finds (again, we are not aware of a rigorous proof of this) that restricted to $\vec{x}\in W(BC_k)$, $$\bar{Z}^a(t;\vec{x})=u^{1,a}(t,\vec{x}).$$

Unlike the type $A$ case, we do not presently know of any discrete space or $q$-deformed discrete space versions of half-space polymers which display a similar degree of integrability as the continuum version. The work of \cite{TWhalfspace} and \cite{OSZ} are, however, suggestive of such a possibility.

\section{Solution via nested contour integral ansatz}\label{solnsec}

We solve the many body systems given above without appealing to the eigenfunction expansion (such as provided by the Bethe Ansatz). Even though we are presently only interested in the case $c=1$ and $a\geq 0$, the below formula applies for all $c,a\in \R$. %\note{reference sections in the introduction which discuss these type of formulas and their relation to other areas}

\subsection{Full-space polymer}

\begin{lemma}\label{solnlemma2}
The type $A_k$ free evolution equation with $k-1$ boundary conditions and coupling constant $c\in \R$ (see Definition \ref{Akmbs}) is solved by
$$ u^c(t;\vec{x}) = \int_{\alpha_1-\I \infty}^{\alpha_1+\I\infty} dz_1 \cdots \int_{\alpha_k-\I \infty}^{\alpha_k+\I\infty} dz_k \prod_{1\leq A<B\leq k} \frac{z_A-z_B}{z_A-z_B-c}\, \prod_{j=1}^{k} e^{\frac{t}{2} z_j^2 + x_jz_j},$$
where for $c\geq 0$ we assume that $\alpha_1>\alpha_2 + c > \alpha_3 + 2c> \cdots > \alpha_k + (k-1)c$, and for $c<0$, all $\alpha_i\equiv 0$.
\end{lemma}

Thus (up to assuming uniqueness of solutions to the many body system of Definition \ref{Akmbs}) we may conclude that for $x_1\leq x_2\leq \cdots \leq x_k$,
\begin{equation}\label{Aknci}
\bar{Z}(t;\vec{x}) =  \int_{\alpha_1-\I \infty}^{\alpha_1+\I\infty} \frac{dz_1}{2\pi \I} \cdots \int_{\alpha_k-\I \infty}^{\alpha_k+\I\infty} \frac{dz_k}{2\pi \I} \prod_{1\leq A<B\leq k} \frac{z_A-z_B}{z_A-z_B-1} \, \prod_{j=1}^{k} e^{\frac{t}{2} z_j^2 + x_jz_j},
\end{equation}
for any $\alpha_1>\alpha_2 + 1 > \alpha_3 + 2> \cdots > \alpha_k + (k-1)$.

\begin{proof}[Proof of Lemma \ref{solnlemma2}]
A proof of this result can be found in \cite{BorCor}, Proposition 6.2.3. As this proof is a straightforward modification of the proof of Lemma \ref{solnlemma} given below, we do not reproduce it presently.
\end{proof}

\subsection{Half-space polymer}

\begin{lemma}\label{solnlemma}
The type $BC_k$ free evolution equation with $k-1$ boundary conditions and coupling constant $c, a \in \R$ (see Definition \ref{Bkmbs}) is solved by
$$ u^{c,a}(t;\vec{x}) = 2^k \int_{\alpha_1-\I\infty}^{\alpha_1+\I\infty} \frac{dz_1}{2\pi \I} \cdots \int_{\alpha_k-\I\infty}^{\alpha_k+\I\infty} \frac{dz_k}{2\pi \I} \prod_{1\leq A<B\leq k} \frac{z_A-z_B}{z_A-z_B-c}\, \frac{z_A+z_B}{z_A+z_B-c} \, \prod_{j=1}^{k} e^{\frac{t}{2} z_j^2 + x_jz_j} \frac{z_j}{z_j+a},$$
where for $c\geq 0$ we assume that $\alpha_1>\alpha_2 + c > \alpha_3 + 2c> \cdots > \alpha_k + (k-1)c$ and $\alpha_k=\max(-a+\e,0)$ for $\e>0$ arbitrary, and for $c<0$, all $\alpha_i\equiv \max(-a+\e,0)$ for $\e>0$ arbitrary. %\note{Alexei: I am slightly unclear about this $c<0$ case whether it is necessary to take the $\alpha_i$ as here, or whether they can all be 0}
\end{lemma}

Thus (up to assuming uniqueness of solutions to the many body system of Definition \ref{Bkmbs}) we may conclude that for $x_1\leq x_2\leq \cdots \leq x_k\leq 0$,
\begin{equation}\label{Bknci}
\bar{Z}^{a}(t;\vec{x}) =2^k\int_{\alpha_1-\I \infty}^{\alpha_1+\I\infty} \frac{dz_1}{2\pi \I} \cdots \int_{\alpha_k-\I \infty}^{\alpha_k+\I\infty} \frac{dz_k}{2\pi \I} \prod_{1\leq A<B\leq k} \frac{z_A-z_B}{z_A-z_B-1}\, \frac{z_A+z_B}{z_A+z_B-1} \, \prod_{j=1}^{k} e^{\frac{t}{2} z_j^2 + x_jz_j}\frac{z_j}{z_j+a},
\end{equation}
where we assume that $\alpha_1>\alpha_2 + 1 > \alpha_3 + 2> \cdots > \alpha_k + (k-1)$ and $\alpha_k=\max(-a+\e,0)$ for $\e>0$ arbitrary.

\begin{proof}[Proof of Lemma \ref{solnlemma}]
We sequentially check that $u^{c,a}(t;\vec{x})$ given in the statement of the lemma satisfies conditions (1), (2), (3) and (4) of the many body system in Definition \ref{Bkmbs}.

\noindent (1): This follows immediately from Leibniz's rule and the fact that $$\frac{d}{dt} e^{\frac{t}{2} z^2 +zx} = \frac{1}{2} \frac{d^2}{dx^2} e^{\frac{t}{2} z^2 +zx}.$$\newline
\noindent (2): Applying $\left(\frac{d}{dx_i} - \frac{d}{dx_{i+1}} -c\right)$ to the integrand of $u^{c,a}(t;\vec{x})$ changes the integrand by simply bringing out a factor of $(z_i-z_{i+1}-c)$.  This cancels the corresponding term in the denominator from the product over $1\leq A<B\leq k$. As a result, the integrand no longer has a pole corresponding to $z_i -z_{i+1} = c$ which means that without crossing any poles, we can deform the $z_{i+1}$ contour to coincide with the $z_{i}$ contour (the Gaussian decay at infinity justifies the deformation of these infinite contours, and the fact that all of the contours lie to the right of $-a$ implies no pole from the denominator $z_j+a$ is crossed). Note that if $c<0$ such a deformation is not necessary.

Since $x_i=x_{i+1}$, the integrand can now be rewritten in the form $$\int dz_i \int dz_{i+1} (z_i-z_{i+1}) G(z_i,z_{i+1})$$ where $G$ is symmetric in its two variables and all of the other integrations have been absorbed into its definition. As the contours of integrals coincide, it follows immediately that this integral is 0.\newline\newline
\noindent (3): For $\vec{x}$ with $x_k=0$, we find that
\begin{eqnarray*}
\left(\frac{d}{dx_k}+a\right) u^{c,a}(t;\vec{x})  \qquad\qquad\qquad\qquad\qquad\qquad\qquad\qquad\qquad\qquad\qquad\qquad\qquad\qquad\qquad\qquad\qquad&   \\
=2^k\int_{\alpha_1-\I \infty}^{\alpha_1+\I\infty} dz_1 \cdots \int_{\alpha_k-\I \infty}^{\alpha_k+\I\infty} dz_k \prod_{1\leq A<B\leq k} \frac{z_A-z_B}{z_A-z_B-c}\, \cdot \frac{z_A+z_B}{z_A+z_B-c} \,  \prod_{j=1}^{k-1} e^{\frac{t}{2} z_j^2 + x_jz_j} \frac{z_j}{z_j+a} \, e^{\frac{t}{2} z_k^2} z_k. &
\end{eqnarray*}
We are now free to deform the $z_k$ contour from $\alpha_k+\I \R$ to $\I\R$ without crossing any poles. The above integrand is invariant under the change of variables $z_k \mapsto -z_k$ as is the contour $\I \R$ (up to a factor of $-1$ coming from the change in orientation of $\I \R$). This, however, implies that the integral above is equal to the integral with the orientation of the $z_k$ contour reversed, which, in turn, implies that the integral is 0.\newline\newline
\noindent (4): First note that
$$
\frac{z_A-z_B}{z_A-z_B-c} = 1 + \frac{c}{z_A-z_B-c}, \qquad \textrm{and} \qquad \frac{z_A+z_B}{z_A+z_B-c} = 1 + \frac{c}{z_A+z_B-c}.
$$
We can write $u^{c,a}(t;\vec{x}) = \sum u^{c,a}_{K_1,K_2}(t;\vec{x})$, where the sum is over all $K_1, K_2\subset \left\{(A,B):1\leq A<B\leq k\right\}$, and where
\begin{eqnarray*}
u^{c,a}_{K_1,K_2}(t;\vec{x}) \qquad\qquad\qquad\qquad\qquad\qquad\qquad\qquad\qquad\qquad\qquad\qquad\qquad\qquad\qquad\qquad\qquad\qquad     &\\
= 2^k \int_{\alpha_1-\I\infty}^{\alpha_1+\I\infty} \frac{dz_1}{2\pi \I} \cdots \int_{\alpha_k-\I\infty}^{\alpha_k+\I\infty} \frac{dz_k}{2\pi \I} \prod_{(A,B)\in K_1} \frac{c}{z_A-z_B-c} \prod_{(A,B)\in K_2} \frac{c}{z_A+z_B-c} \,\prod_{j=1}^{k} e^{\tfrac{t}{2} z_j^2 + z_j x_j}\frac{z_j}{z_j+a}.&
\end{eqnarray*}
For use later, let $\ell = |K_1|+|K_2|$.

Define the unit vector in the direction of $\vec{x}$ as $v = \vec{x} / ||\vec{x}||$ and note that for $x_1< x_2< \cdots< x_k< 0$, it follows that $-v$ is strictly positive. On account of this fact, we can deform the $z_j$ contours to lie along $-t^{1/2} v_j + \I \R$ (this is because of the ordering of the elements of $\vec{x}$ and the Gaussian decay which justifies the deformation of contours). Let $w_j = t^{1/2} z_j$, then
\begin{eqnarray*}
u^{c,a}_{K_1,K_2}(t;\vec{x}) &=&2^k c^\ell  t^{\ell/2} t^{-k/2} \int_{-v_1-\I\R}^{-v_1+\I \R} \frac{dw_1}{2\pi \I}\cdots \int_{-v_k-\I\R}^{-v_k+\I \R} \frac{dw_k}{2\pi \I}\\
& & \times \, \prod_{(A,B)\in K_1} \frac{1}{w_A-w_B-ct^{1/2}} \prod_{(A,B)\in K_2} \frac{1}{w_A+w_B-ct^{1/2}}  e^{t^{-1/2} \vec{x}\cdot \vec{w}}\, \prod_{j=1}^{k} e^{\tfrac{1}{2} w_j^2}\frac{w_j}{w_j+t^{1/2}a}.
\end{eqnarray*}
Due to the choice of contours, $\Real(\vec{x}\cdot \vec{w}) = - ||\vec{x}||$.
Taking absolute values we find that
\begin{eqnarray*}
\big\vert u^{c,a}_{K_1,K_2}(t;\vec{x})\big\vert &=& 2^k c^\ell   t^{\ell/2} t^{-k/2} e^{-t^{-1/2} ||x||}  \int_{-v_1-\I\R}^{-v_1+\I \R} \frac{dw_1}{2\pi \I}\cdots \int_{-v_k-\I\R}^{-v_k+\I \R} \frac{dw_k}{2\pi \I}\\
& &\times\,\left|\prod_{(A,B)\in K_1} \frac{1}{w_A-w_B-ct^{1/2}} \prod_{(A,B)\in K_2} \frac{1}{w_A+w_B-ct^{1/2}}\, \prod_{j=1}^{k} e^{\tfrac{1}{2} w_j^2} \frac{w_j}{w_j+t^{1/2}a} \right|.
\end{eqnarray*}
The integral above is now clearly convergent, independent of $x$, and bounded by a constant $C$ which is uniform in $t<t_0$. Thus we have established that
\begin{equation}\label{ubound}
\big\vert u^{c,a}_{K_1,K_2}(t;\vec{x})\big\vert \leq C' t^{\ell/2} t^{-k/2} e^{-t^{-1/2} ||x||}
\end{equation}
for some $C'>0$.

For any bounded continuous function $f$ we may compute the integral against $u^{c,a}_{K_1,K_2}$. The boundedness of $f$ means that  $|f(x)|<M$ for some constant $M$ as $x$
varies. By the triangle inequality
\begin{equation*}
\left|\int_{W(BC_k)} u^{c,a}_{K_1,K_2}(t;\vec{x}) f(\vec{x})d\vec{x}\right| \leq \int_{W(BC_k)} \left|u^{c,a}_{K_1,K_2}(t,\vec{x})\right| M d\vec{x} \leq C'' t^{\ell/2}.
\end{equation*}
The second inequality follows by substituting (\ref{ubound}) and performing the change of variables $y_i=t^{-1/2} x_i$ in order to bound the resulting integral (this change of variables results in the cancelation of the $t^{-k/2}$ term). Observe that $\ell\geq 1$ for all choices of $K_1$ and $K_2$ which are not simultaneously empty, and as $t\to 0$ these terms disappear.

The only term in the expansion of $u^{c,a}$ which does not contribute negligibly as $t\to 0$ corresponds to when $K_1=K_2=\emptyset$ are empty. In that case the desired result can be seen in the following manner. We showed that
\begin{equation*}
k! \int_{W(BC_k)} f(x) u^{c,a}(t;\vec{x}) d\vec{x} = o(1) + 2^k k! \int_{\alpha_1-\I\infty}^{\alpha_1+\I\infty} \frac{dz_1}{2\pi \I} \cdots \int_{\alpha_k-\I\infty}^{\alpha_k+\I\infty} \frac{dz_k}{2\pi \I} \, \tilde{f}(\vec{z}) \prod_{j=1}^{k} e^{\frac{t}{2}z_j^2} \frac{z_j}{z_j+a},
\end{equation*}
where
\begin{equation*}
\tilde{f}(\vec{z}) =  \int_{W(BC_k)} f(x) \prod_{j=1}^{k} e^{x_j z_j},
\end{equation*}
and where the $o(1)$ term comes from all other terms in the expansion of $u^{c,a}$ and goes to $0$ as $t\to 0$.
Writing $\frac{z}{z+a}$ as $1+ \frac{a}{z+a}$ and using decay of $\tilde{f}(\vec{z})$ as the real part of $z_j$ goes to infinity, it is possible to show that
\begin{equation*}
\lim_{t\to 0} 2^k k! \int_{W(BC_k)} f(x) u^{c,a}(t;\vec{x}) d\vec{x} = 2^k k! \int_{\alpha_1-\I\infty}^{\alpha_1+\I\infty} \frac{dz_1}{2\pi \I} \cdots \int_{\alpha_k-\I\infty}^{\alpha_k+\I\infty} \frac{dz_k}{2\pi \I} \, \tilde{f}(\vec{z}).
\end{equation*}
Observe that $\tilde{f}(\vec{z})$ is the Fourier transform of a function which is $f(\vec{x})$ for $\vec{x}\in W(BC_k)$ and zero outside, we readily see that
\begin{equation*}
2^k k! \int_{\alpha_1-\I\infty}^{\alpha_1+\I\infty} \frac{dz_1}{2\pi \I} \cdots \int_{\alpha_k-\I\infty}^{\alpha_k+\I\infty} \frac{dz_k}{2\pi \I} \, \tilde{f}(\vec{z}) = f(0),
\end{equation*}
which is exactly as desired.

\end{proof}

\section{Expansion of nested contour integrals into residues}\label{expsec}

At this point we have seen two types of nested contour integral formulas -- those related to type $A$ (the full-space polymer) and those related to type $BC$ (the half-space polymer). In this section we present formulas for each type which shows the effect of deforming the nested contours to all lie upon the imaginary axis $\I \R$. This amounts to computing the residue expansion coming from the various poles crossed during this deformation.

The general scheme followed in these computations are quite similar between the two types (though in the BC case we are presently unable to complete the proof, thus the result remains a conjecture). The first step is to identify the residual subspaces resulting from the deformation. The second step is to show that (in both cases) the residual subspaces can be brought (via the respective action of the symmetric or hyperoctahedral  groups) into a certain canonical form indexed only by partition $\lambda\vdash k$. And, the third step is to compute the residue corresponding to a given partition $\lambda$. This splits into computing the residue of the portion of the integrand containing poles and computing the substitution of the residual subspace relations into the remaining portion of the integrand. The observation of the first step, once made, is relatively easy to confirm, and computing residues and substitutions is straightforward as well (if not slightly technical). In the BC case, there is a subtlety involving contour deformations which we encounter in step 2. Essentially, in order to bring our expressions into a canonical form we end up getting contours of integration which need to be shifted back to $\I \R$. During this shifting we encounter a number of poles which we believe, but have not been able to prove, all cancel. This is the content of Claim \ref{claim:Res-are-zero}.
The basic structure of the computation behind these results comes from the proof of Theorem 3.13 of \cite{HO}, see also Section \ref{expsec} for its non-technical description.
 %That computation, however, is not complete since constants did not need to be computed exact for the application in \cite{HO} (see \note{reference introduction for discussion on HO's application and other intrigue around this}).

Proposition \ref{Akcprop} and Conjecture \ref{Bkcprop} are stated below, though their proofs are held off until Section \ref{proofsec}. Immediately after presenting these residue expansions we show how they are applied to the full-space and half-space polymer partition function joint moment formulas. In the case of the half-space polymer, Conjecture \ref{Bkcprop} is only applied for $a\geq 0$ (repulsive interaction at the origin). The source of this restriction has to do with the occurrence of additional poles in the contour deformations when $a<0$.

In the statements and use of these propositions we use some notation related to partitions, the symmetric group $S_k$, the hyperoctahedral  group $BC_k$, as well as Pochhammer symbols $(a)_{n}$, which all can be found in Section \ref{partsec} or \ref{Asymsec}.

\subsection{Type $A_k$ expansion and full-space polymer application}
We state Proposition \ref{Akcprop} and then immediately apply it to expand the nested contour integral formula for $\bar{Z}(t;\vec{x})$ calculated earlier in equation (\ref{Aknci}). The proof of the below proposition can be found in Section \ref{proofsec}.

\begin{proposition}\label{Akcprop}
Fix $k\geq 1$ and $c\in (0,\infty)$. Given a set of real numbers $\alpha_1,\ldots,\alpha_k$ and a function $F(z_1,\ldots ,z_k)$ which satisfy
\begin{enumerate}
\item For all $1\leq j \leq k-1$, $\alpha_j>\alpha_{j+1}+c$;
\item For all $1\leq j \leq k$ and $z_1,\ldots,z_k$ such that $z_i\in \alpha_i+\I \R$ for $1\leq i<j$ and $z_i\in \alpha_k+\I \R$ for $j<i\leq k$, the function $z_j\mapsto F(z_1,\ldots ,z_j,\ldots, z_k)$ is analytic in the complex domain $\{z: \alpha_k\leq \Real(z) \leq \alpha_j\}$ and is bounded in modulus on that domain by $\const\, \Imag(z_j)^{-1-\delta}$ for some constants $\const,\delta>0$ (depending on $z_1,\ldots,z_{j-1},z_{j+1},\ldots,z_k$ but not $z_j$).
\end{enumerate}
Then we have the following residue expansion identity:
\begin{align*}
\int_{\alpha_1-\I \infty}^{\alpha_1+\I \infty} \frac{dz_1}{2\pi \I} \cdots \int_{\alpha_k-\I \infty}^{\alpha_k+\I \infty} \frac{dz_k}{2\pi \I} \prod_{1\leq A<B\leq k} \frac{z_A-z_B}{z_A-z_B-c} F(\vec{z})
\qquad\qquad\qquad\qquad\qquad\qquad&\\
=c^k \sum_{\substack{\lambda\vdash k\\ \lambda= 1^{m_1}2^{m_2}\cdots}} \frac{1}{m_1! m_2!\cdots} \, \int_{\alpha_k-\I \infty}^{\alpha_k+\I \infty} \frac{dw_1}{2\pi \I} \cdots \int_{\alpha_k-\I \infty}^{\alpha_k+\I \infty} \frac{dw_{\ell(\lambda)}}{2\pi \I}\, \det\left[\frac{1}{w_i +c\lambda_i -w_j}\right]_{i,j=1}^{\ell(\lambda)}&\\
\times\,E^c(w_1,w_1+c,\ldots, w_1 + c(\lambda_1-1), \ldots, w_{\lambda_{\ell(\lambda)}},w_{\lambda_{\ell(\lambda)}}+c, \ldots, w_{\lambda_{\ell(\lambda)}}+c(\lambda_{\ell(\lambda)}-1)),&
\end{align*}
where
\begin{equation*}
E^c(z_1,\ldots, z_k) := \sum_{\sigma\in S_k} \prod_{1\leq B<A\leq k} \frac{z_{\sigma(A)}-z_{\sigma(B)}-c}{z_{\sigma(A)}-z_{\sigma(B)}} F(\sigma(\vec{z})).
\end{equation*}
\end{proposition}

We may apply this proposition to the formula for $\bar{Z}(t,\vec{x})$ provided in equation (\ref{Aknci}). The hypotheses of the proposition are easily checked and the outcome is the following expansion:
\begin{eqnarray*}
\bar{Z}(t;\vec{x}) = \sum_{\substack{\lambda\vdash k\\ \lambda= 1^{m_1}2^{m_2}\cdots}} \frac{1}{m_1! m_2!\cdots} \, \int_{-\I \infty}^{\I \infty} \frac{dw_1}{2\pi \I} \cdots \int_{\alpha_k-\I \infty}^{\alpha_k+\I \infty} \frac{dw_{\ell(\lambda)}}{2\pi \I} \det\left[\frac{1}{w_i +\lambda_i -w_j}\right]_{i,j=1}^{\ell(\lambda)}&\\
\times\,E^1_{\vec{x}}(w_1,w_1+1,\ldots, w_1 + (\lambda_1-1), \ldots, w_{\lambda_{\ell(\lambda)}},w_{\lambda_{\ell(\lambda)}}+1,\ldots, w_{\lambda_{\ell(\lambda)}}+(\lambda_{\ell(\lambda)}-1)),&
\end{eqnarray*}
where
\begin{equation*}
E^1_{\vec{x}}(z_1,\ldots, z_k) := \sum_{\sigma\in S_k} \prod_{1\leq B<A\leq k} \frac{z_{\sigma(A)}-z_{\sigma(B)}-1}{z_{\sigma(A)}-z_{\sigma(B)}} \prod_{j=1}^{k} e^{\frac{t}{2} z_{\sigma(j)}^2 + x_j z_{\sigma(j)}}.
\end{equation*}

As we are interested in $\EE[Z(t,0)^k]$ it suffices to consider all $x_i\equiv 0$. Using the symmetrization identity\footnote{As long as all $x_i$ are equal, the identity in equation (\ref{Akfacsym}) can be applied. However, when the $x_i$ differ it is not clear how to simplify $E_{\vec{x}}$ in a similar manner.} given in equation (\ref{Akfacsym}) we find that
\begin{eqnarray*}
E^1_{\vec{0}}(w_1,w_1+1,\ldots, w_1 +(\lambda_1-1), \ldots, w_{\lambda_{\ell(\lambda)}},w_{\lambda_{\ell(\lambda)}}+1,\ldots, w_{\lambda_{\ell(\lambda)}}+(\lambda_{\ell(\lambda)}-1)) &\\
= k! \prod_{j=1}^{\ell(\lambda)} e^{\frac{t}{2}\big(w_j^2 + (w_j+1)^2+\cdots (w_j+\lambda_j-1)^2\big)} = k! \prod_{j=1}^{\ell(\lambda)} \frac{e^{\frac{t}{2} G(w_j+\lambda_j)}}{e^{\frac{t}{2} G(w_j)}},\qquad\qquad\qquad\qquad&
\end{eqnarray*}
where
\begin{equation}\label{Gdef}
G(w)=\frac{w^3}{3}-\frac{w^2}{2}+\frac{w}{6}.
\end{equation}
In this last line we have used the fact that
\begin{equation}\label{Grat}
\frac{G(w+\ell)}{G(w)} = w^2 + (w+1)^2 + \cdots (w+\ell-1)^2.
\end{equation}

Using this, we conclude that
\begin{eqnarray}\label{Asupexp}
\EE[Z(t,0)^k]  =  k! \sum_{\substack{\lambda\vdash k\\ \lambda= 1^{m_1}2^{m_2}\cdots}} \frac{1}{m_1! m_2!\cdots} \qquad\qquad\qquad\qquad\qquad\qquad\qquad\qquad\qquad\qquad\qquad\qquad&\\
\nonumber \times\, \int_{-\I \infty}^{\I \infty} \frac{dw_1}{2\pi \I} \cdots \int_{-\I \infty}^{\I \infty} \frac{dw_{\ell(\lambda)}}{2\pi \I} \det\left[\frac{1}{w_i +\lambda_i -w_j}\right]_{i,j=1}^{\ell(\lambda)} \prod_{j=1}^{\ell(\lambda)} \frac{e^{\frac{t}{2} G(w_j+\lambda_j)}}{e^{\frac{t}{2} G(w_j)}}.
\end{eqnarray}

\subsection{Type $BC_k$ expansion and half-space polymer application}\label{bkcase}
We present an analogous result to Proposition \ref{Akcprop} above, which we then apply to $\bar{Z}^0(t;\vec{x})$. We explain below in Remark \ref{nota} why the residue expansion for the general case of $a\neq 0$ presents some additional complexities (and possibly additional terms).
We actually give a complete proof of the below conjecture modulo Claim \ref{claim:Res-are-zero} which deals with the cancelation of the sum of certain residues arising from contour deformations in the course of our computations. Partial evidence why Claim \ref{claim:Res-are-zero} is correct is given in Section \ref{sec:one-string}, where cancelation of some of the arising residues is shown.

\begin{conj}\label{Bkcprop}
Fix $k\geq 1$ and $c\in (0,\infty)$. Given a set of real numbers $\alpha_1,\ldots,\alpha_k$ and a function $F(z_1,\ldots ,z_k)$ which satisfy
\begin{enumerate}
\item $\alpha_k=0$;
\item For all $1\leq j \leq k-1$, $\alpha_j>\alpha_{j+1}+c$;
\item For $z_1,\ldots, z_k$ in the complex domain $\{x+\I y: -\alpha_1\leq  x\leq \alpha_1, y\in \R\}$, for all $1\leq j\leq k$, the function $z_j\mapsto F(z_1,\ldots ,z_j,\ldots, z_k)$ is analytic and is bounded by a constant (depending on $z_1,\ldots,z_{j-1},z_{j+1},\ldots,z_k$ but not $z_j$) times $\Imag(z_j)^{-1-\delta}$ for some $\delta>0$.
\end{enumerate}
Then we have the following residue expansion identity:
\begin{align}\label{mainform}
&\int_{\alpha_1-\I \infty}^{\alpha_1+\I \infty}\frac{dz_1}{2 \pi\I} \cdots \int_{\alpha_k-\I \infty}^{\alpha_k+\I \infty} \frac{dz_k}{2 \pi\I} \prod_{1\leq A<B\leq k} \frac{z_A-z_B}{z_A-z_B-c} \,\frac{z_A+z_B}{z_A+z_B-c} F(\vec{z})\\
\nonumber &= c^k \sum_{\substack{\lambda\vdash k\\ \lambda= 1^{m_1}2^{m_2}\cdots}} \frac{(-1)^{\ell(\lambda)}}{m_1! m_2!\cdots} \, \int_{\alpha_k-\I \infty}^{\alpha_k+\I \infty} \frac{dw_1}{2 \pi\I} \cdots \int_{\alpha_k-\I \infty}^{\alpha_k+\I \infty} \frac{dw_{\ell(\lambda)}}{2 \pi\I} \prod_{j=1}^{\ell(\lambda)} \frac{1}{4 c}\, \frac{\left(\tfrac{2w_j+c}{2c}\right)_{\lambda_j-1}}{\left(\tfrac{2w_j}{2c}\right)_{\lambda_j}} \Pf\left[\frac{u_i-u_j}{u_i+u_j}\right]_{i,j=1}^{2\ell(\lambda)} \\
\nonumber & E^{c}(w_1,w_1+c,\ldots, w_1+c(\lambda_1-1),\ldots, w_{\ell(\lambda)},w_{\ell(\lambda)}+c,\ldots, w_{\ell(\lambda)} + c(\lambda_{\ell(\lambda)}-1))
\end{align}
where
\begin{equation*}
(u_1,\ldots, u_{2\ell(\lambda)}) :=\big(-w_1+\tfrac{c}{2}, w_1-\tfrac{c}{2}+\lambda_1 c, -w_2+\tfrac{c}{2},w_2-\tfrac{c}{2}+\lambda_2 c, \ldots, -w_{\ell(\lambda)}+\tfrac{c}{2}, w_{\ell(\lambda)}-\tfrac{c}{2} + \lambda_{\ell(\lambda)} c\big)
\end{equation*}
and where
\begin{equation*}
E^{c}(z_1,\ldots, z_k) := \sum_{\sigma\in BC_k} \prod_{1\leq B<A\leq k} \frac{z_{\sigma(A)}-z_{\sigma(B)}-c}{z_{\sigma(A)}-z_{\sigma(B)}}\frac{z_{\sigma(A)}+z_{\sigma(B)}-c}{z_{\sigma(A)}+z_{\sigma(B)}} F(\sigma(\vec{z})).
\end{equation*}
\end{conj}

We apply this conjecture to the formula for $\bar{Z}^0(t,\vec{x})$ provided in equation (\ref{Bknci}). Observe that the formula in (\ref{Bknci}) can be written in the form of the left-hand side of (\ref{mainform}) by taking $c=1$ and
\begin{equation}
\label{eq:F-apply}
F(\vec{z}) = \prod_{j=1}^{k} e^{\frac{t}{2}z_j^2 + x_j z_j}.
\end{equation}
Notice that the contours satisfy the first two hypotheses of Conjecture \ref{Bkcprop} and that the function $F$ satisfies the analyticity and decay conditions of the third hypothesis.

\begin{remark}\label{nota}
The reason we have restricted our attention at this point to the case of $a=0$ is because of the following: For general $a\neq 0$ the function $F(\vec{z})$ contains an extra term
$$\prod_{j=1}^{k} \frac{z_j}{z_j+a}.$$
When $a<0$ the contours from equation (\ref{Bknci}) are such that $\alpha_k=-a+\epsilon>0$. If we want to shift to $\alpha_k=0$ then we cross poles, which contribute additional terms which are not accounted for in the above result. For $a>0$ we have $\alpha_k=0$, however, if $a<(k-1)$, then the third hypothesis is violated since the pole coming from the denominator of the above fraction will be included in the complex domain in which analyticity is desired (see Remark \ref{aexception} for an explanation of how this hypothesis is utilized).
\end{remark}

The result of applying the above conjecture for $a=0$ is that (from here on out, all half-space formulas will assume the validity of Conjecture \ref{Bkcprop})
\begin{eqnarray*}
\bar{Z}^0(t;\vec{x}) = 2^k \sum_{\substack{\lambda\vdash k\\ \lambda= 1^{m_1}2^{m_2}\cdots}} \frac{(-1)^{\ell(\lambda)}}{m_1! m_2!\cdots} \, \int_{-\I \infty}^{\I \infty} \frac{dw_1}{2\pi \I} \cdots \int_{-\I \infty}^{\I \infty} \frac{dw_{\ell(\lambda)}}{2\pi \I}
\prod_{j=1}^{\ell(\lambda)} \frac{1}{4}\, \frac{\left(w_j+\frac{1}{2}\right)_{\lambda_j-1}}{(w_j)_{\lambda_j}} \Pf\left[\frac{u_i-u_j}{u_i+u_j}\right]_{i,j=1}^{2\ell(\lambda)}&\\
\times\,E^{1}_{\vec{x}}(w_1,w_1+1,\ldots, w_1 + (\lambda_1-1), \ldots, w_{\lambda_{\ell(\lambda)}},w_{\lambda_{\ell(\lambda)}}+1,\ldots, w_{\lambda_{\ell(\lambda)}}+(\lambda_{\ell(\lambda)}-1)),&
\end{eqnarray*}
where
\begin{equation}\label{uabove}
(u_1,\ldots, u_{2\ell(\lambda)}) :=\big(-w_1+\tfrac{1}{2}, w_1-\tfrac{1}{2}+\lambda_1 , -w_2+\tfrac{1}{2},w_2-\tfrac{1}{2}+\lambda_2 , \ldots, -w_{\ell(\lambda)}+\tfrac{1}{2}, w_{\ell(\lambda)}-\tfrac{1}{2} + \lambda_{\ell(\lambda)} \big)
\end{equation}
and
\begin{equation*}
E^{1}_{\vec{x}}(z_1,\ldots, z_k) := \sum_{\sigma\in BC_k} \prod_{1\leq B<A\leq k} \frac{z_{\sigma(A)}-z_{\sigma(B)}+c}{z_{\sigma(A)}-z_{\sigma(B)}}\frac{z_{\sigma(A)}+z_{\sigma(B)}+c}{z_{\sigma(A)}+z_{\sigma(B)}} \, \prod_{j=1}^{k} e^{\frac{t}{2} z_{\sigma(j)}^2 + x_j z_{\sigma(j)}}.
\end{equation*}

As we are interested in $\EE[Z^0(t,0)^k]$ it suffices to consider all $x_i\equiv 0$. Using the symmetrization identity\footnote{We are presently only aware of a suitable identity to deal with the case $x_i\equiv 0$ (even $x_i\equiv x$ for $x\neq 0$ is unclear).} given in equation (\ref{Bkfacsym}) we find that
\begin{eqnarray*}
E^{1}_{\vec{0}} (z_1,\ldots, z_k) &=& \prod_{j=1}^{k} e^{\frac{t}{2} z_j^2} \sum_{\sigma\in BC_k} \prod_{1\leq B<A\leq k} \frac{z_{\sigma(A)}-z_{\sigma(B)}+c}{z_{\sigma(A)}-z_{\sigma(B)}}\frac{z_{\sigma(A)}+z_{\sigma(B)}+c}{z_{\sigma(A)}+z_{\sigma(B)}}\\
&=& 2^k k! \prod_{j=1}^{k} e^{\frac{t}{2} z_j^2}.
\end{eqnarray*}
Note that the identity was used in going from the first to second line. The equality in the first line came from the $BC_k$ invariance of the product which is factored out of the summation over $BC_k$.

%Observe that
%\begin{equation*}
%\frac{w^2}{w^2-a^2} \,\, \frac{(w+1)^2}{(w+1)^2-a^2}\, \cdots\, \frac{(w+\ell-1)^2}{(w+\ell-1)^2-a^2} = \frac{(w)_{\ell} \, (w)_{\ell}}{(w+a)_{\ell}\, (w-a)_{\ell}}.
%\end{equation*}
Owing to equation (\ref{Grat}) we may rewrite
\begin{eqnarray*}
E^{1}_{\vec{0}}(w_1,w_1+1,\ldots, w_1 +(\lambda_1-1), \ldots, w_{\lambda_{\ell(\lambda)}},w_{\lambda_{\ell(\lambda)}}+1,\ldots, w_{\lambda_{\ell(\lambda)}}+(\lambda_{\ell(\lambda)}-1)) &\\
= 2^k k! \, \prod_{j=1}^{\ell(\lambda)} \frac{e^{\frac{t}{2} G(w_j+\lambda_j)}}{e^{\frac{t}{2} G(w_j)}}. &
\end{eqnarray*}

Using this, we conclude that
\begin{eqnarray}\label{Bsupexp}
\EE[Z^0(t,0)^k]  = 4^k k! \sum_{\substack{\lambda\vdash k\\ \lambda= 1^{m_1}2^{m_2}\cdots}} \frac{(-1)^{\ell(\lambda)}}{m_1! m_2!\cdots} \qquad\qquad\qquad\qquad\qquad\qquad\qquad\qquad\qquad\qquad\qquad\qquad&\\
\nonumber \times\,\int_{-\I \infty}^{\I \infty} \frac{dw_1}{2\pi \I} \cdots \int_{-\I \infty}^{\I \infty} \frac{dw_{\ell(\lambda)}}{2\pi \I}
\prod_{j=1}^{\ell(\lambda)} \frac{1}{4}\, \frac{\left(w_j+\frac{1}{2}\right)_{\lambda_j-1}}{(w_j)_{\lambda_j}}  \frac{e^{\frac{t}{2} G(w_j+\lambda_j)}}{e^{\frac{t}{2} G(w_j)}} \Pf\left[\frac{u_i-u_j}{u_i+u_j}\right]_{i,j=1}^{2\ell(\lambda)},&
\end{eqnarray}
with the $u$ variables are defined in terms of the $w$ variables in (\ref{uabove}).

\section{Summing the moment generating function}\label{mgfsec}

We seek now to derive concise formulas for the Laplace transforms of $Z(t,0)$ and $Z^a(t,0)$ for $a\geq 0$. Knowledge of the Laplace transform of a positive random variable uniquely identifies the random variable's distribution and can easily be inverted. Unfortunately, one readily checks that the moments of these random variables grow too rapidly (in $k$) to characterize their Laplace transforms and distributions\footnote{One shows that $k^{th}$ moment grows as $\exp(\const \cdot k^3)$ for some $\const>0$ and $k$ and $t$ large, see e.g. \cite{BC}, \cite{BCLyapunovpaper}.}. As such, we proceed formally (i.e. in a mathematically unjustified manner) and still attempt to use the moments to recover the Laplace transform. This means that we write
\begin{equation}\label{mgf}
\EE\left[e^{\zeta Z(t,0)}\right] = \sum_{k=0}^{\infty} \frac{\zeta^k}{k!} \EE[Z(t,0)^k].
\end{equation}
The left-hand side above is clearly finite for all $\zeta$ with $\Real(\zeta)<0$. On the other hand, since $\EE[Z(t,0)^k]$ grows super-exponentially (as can be seen from analysis of (\ref{Asupexp})) the right-hand side is a divergent series for all $\zeta$, and does not converge in any sense. Of course the source of this apparent inconsistency is in the interchange of expectations and infinite summation which is not mathematically justifiable -- hence the divergent right-hand side. Our derivation of the Laplace transform formula, therefore, involves ``summing'' this divergent moment generating function series. This is done via an illegal application of a Mellin-Barnes trick which says that for functions $g(z)$ with suitable decay and analyticity properties (cf. Lemma 3.2.13 of \cite{BorCor}) we have (using the fact that the residue of $\pi/\sin(-\pi z)$ at $z=n$ is $(-1)^{n+1}$) that
\begin{equation}\label{Mellin}
\sum_{n=1}^{\infty} g(n) \zeta^n = \int \frac{dz}{2\pi \I}  \frac{\pi}{\sin(-\pi z)} g(z) (-\zeta)^z.
\end{equation}
The negatively oriented $z$ contour should encircle the positive integers and no poles of $g(z)$, as well as satisfy certain properties related to the growth of $g(z)$ at infinity. Our illegal application neglects consideration of the growth condition (which clearly is not satisfied). However, the result after the application of this trick becomes perfectly well-defined and convergent. The result of this is given below and summarized in Section \ref{Laplacesec} of the introduction.

Despite this divergent summation, in the full-space case the final answer we arrive at for the Laplace transform of $Z(t,0)$ can be compared to the formula proved in \cite{ACQ,BCF} and it is in agreement. Presently we have no alternative proof of the analogous final result for the half-space case. It is quite reasonable to wonder why the formal manipulations of divergent series that we perform (at least in the full-space case) actually lead to the correct Laplace transform formula. This can be seen as a consequence of the fact that there exist $q$-deformed discrete versions of the polymers (both $q$-TASEP and ASEP work) for which every step below can be performed in an analogous, yet totally rigorous manner (see \cite{BorCor, BCS}). Thus, from this particular perspective, the manipulations below can be viewed as shadows of rigorous manipulations at a $q$-deformed level. Presently we do not have any analogous $q$-deformed results for the half-space polymer.

\subsection{Full-space polymer}
We proceed formally -- mimicking the steps in the mathematically rigorous work of \cite{BorCor, BCS}. We seek to sum the moment generating function (\ref{mgf}). Using (\ref{Asupexp}) we find that
\begin{equation*}
\frac{\zeta^k}{k!} \EE[Z(t,0)^k] = \sum_{L=1}^{\infty} \frac{1}{L!} \sum_{\lambda_1=1}^{\infty} \cdots \sum_{\lambda_L=1}^{\infty} \mathbf{1}_{\lambda_1+\cdots + \lambda_L = k} \int_{-\I \infty}^{\I \infty} \frac{dw_1}{2\pi \I} \cdots \int_{-\I\infty}^{\I \infty}  \frac{dw_L}{2\pi \I} I_{L}(\lambda;w;\zeta)
\end{equation*}
where
\begin{equation*}
I_{L}(\lambda;w;\zeta)= \det\left[\frac{1}{w_i+\lambda_i -w_j}\right]_{i,j=1}^{L}  \prod_{j=1}^{L} \zeta^{\lambda_j} \frac{e^{\frac{t}{2} G(w_j+\lambda_j)}}{e^{\frac{t}{2} G(w_j)}}.
\end{equation*}
Notice that we have replaced the summation over partitions with a summation over $L$ and then arbitrary (unordered) $\lambda_1,\ldots, \lambda_L$. This requires insertion of the term $\mathbf{1}_{\lambda_1+\cdots + \lambda_L = k}$ to enforce that the $\lambda_j$ sum to $k$, as well as division by
$$
\frac{L!}{m_1!m_2!\cdots}
$$
from the change between ordered and unordered $\lambda_j$'s. This factor exactly cancels the terms in (\ref{Asupexp}) and yields the above claimed formula.

By summing over $k$ (and performing infinite range reordering within our already divergent series) we find that
\begin{equation}\label{almostfreddet}
\EE\left[e^{\zeta Z(t,0)}\right] = 1+ \sum_{L=1}^{\infty} \frac{1}{L!} \sum_{\lambda_1=1}^{\infty} \cdots \sum_{\lambda_L=1}^{\infty}  \int_{-\I \infty}^{\I \infty}  \frac{dw_1}{2\pi \I} \cdots \int_{-\I\infty}^{\I \infty}  \frac{dw_L}{2\pi \I} I_{L}(\lambda;w;\zeta).
\end{equation}
The indicator function disappeared as a result of summing over $k$. The above equality is only formal since the right-hand side is still divergent at this point.

\subsubsection{The Mellin-Barnes trick}

Notice that for $w_1,\ldots, w_L$ on $\I \R$, the function $I_{L}(\lambda;w;\zeta)$ is analytic in each $\lambda_1,\ldots,\lambda_L$ as long as their real parts stay positive. We use this fact in order to apply the Mellin-Barnes trick explained above in (\ref{Mellin}) to replace the summations of $\lambda_1,\ldots, \lambda_L$ by an $L$-fold integral:
\begin{equation}\label{sumdiv}
\sum_{\lambda_1=1}^{\infty} \cdots \sum_{\lambda_L=1}^\infty I_{L}(\lambda_1,\ldots, \lambda_L ;w;\zeta) = \int_{\frac{3}{4}-\I \infty}^{\frac{3}{4} + \I\infty} \frac{d\lambda_1}{2\pi \I} \cdots \int_{\frac{3}{4}-\I \infty}^{\frac{3}{4} + \I\infty} \frac{d\lambda_1}{2\pi \I} \, \prod_{j=1}^{L} \frac{\pi}{\sin(-\pi \lambda_j)}\, I_{L}(\lambda;w;-\zeta).
\end{equation}
Note that $3/4$ in the integral limits can be replaced by any real number in $(1/2,1)$. The upper bound is from the requirement (in the Mellin-Barnes trick application) that the $\lambda$ contour contain the positive integers including 1, and the lower bound comes later when, after changing variables by subtracting $1/2$, we want to apply the integral representation for the Airy function in equation \ref{Airyintrep}. That application requires being along a positive real part contour, hence the lower bound.

Of course, this equality is only formal (in the $\zeta$ variable) in terms of a residue expansion and does not make mathematical sense due to the growth of the integrand near infinity.

It is, however, via this illegal application of the Mellin-Barnes trick that our divergent series has been ``summed'' and replaced by a convergent integral:
\begin{equation}\label{almostfreddet2}
\EE\left[e^{\zeta Z(t,0)}\right] = 1+ \sum_{L=1}^{\infty} \frac{1}{L!} \int_{-\I \infty}^{\I \infty}  \frac{dw_1}{2\pi \I} \cdots \int_{-\I\infty}^{\I \infty}  \frac{dw_L}{2\pi \I} \int_{\frac{3}{4}-\I \infty}^{\frac{3}{4} + \I\infty} \frac{d\lambda_1}{2\pi \I} \cdots \int_{\frac{3}{4}-\I \infty}^{\frac{3}{4} + \I\infty} \frac{d\lambda_1}{2\pi \I} \, \prod_{j=1}^{L} \frac{\pi}{\sin(-\pi \lambda_j)}\, I_{L}(\lambda;w;-\zeta).
\end{equation}
In fact, the above formula (though arrived at through mathematically non-rigorous means) is equivalent to the rigorously proved formulas in \cite{ACQ,BCF}. To put this in the same form (so as to compare) will require a few more manipulations -- all of which now are mathematically justifiable.

\subsubsection{Manipulating the formula}

We employ two identities (the first valid for $\Real(\zeta)<0$ and the second valid for $\Real(w+\lambda-w')>0$)
$$
\frac{\pi}{\sin(-\pi \lambda)}\, (-\zeta)^{\lambda} = \int_{-\infty}^{\infty}dr \frac{\zeta e^{\lambda r}}{e^{r}-\zeta}, \qquad \qquad \frac{1}{w+\lambda-w'} =  \int_0^{\infty} dx e^{-x(w+\lambda-w')},
$$
perform the change of variables $w_j+\lambda_j = s_j$, and use the linearity properties of the determinant to rewrite
\begin{equation}\label{freddet}
\EE\left[e^{\zeta Z(t,0)}\right] = 1+ \sum_{L=1}^{\infty} \frac{1}{L!} \int_{-\I \infty}^{\I \infty}  \frac{dw_1}{2\pi \I} \cdots \int_{-\I\infty}^{\I \infty}  \frac{dw_L}{2\pi \I} \det\left[K^1_{\zeta}(w_i,w_j)\right]_{i,j=1}^{L} = \det(I+K^1_{\zeta})_{L^2(\I\R,\frac{d\mu}{2\pi \I})}.
\end{equation}
Here $\det(I+K^1_{\zeta})_{L^2(\I\R,\frac{d\mu}{2\pi \I})}$ is the Fredholm determinant of the operator $K^1_{\zeta}$ which acts on functions in $L^2(\I\R,\frac{d\mu}{2\pi \I})$ (with $\frac{d\mu}{2\pi \I}$ representing the measure $\frac{dw}{2\pi \I}$ along $w\in \I \R$). The integral operator $K^1_{\zeta}$ is specified by its kernel
\begin{equation*}
K^1_{\zeta}(w,w') =\int_0^{\infty} dx \int_{-\infty}^{\infty} dr \int_{\frac{3}{4}-\I \infty}^{\frac{3}{4} + \I\infty} \frac{ds}{2\pi \I} \frac{\zeta e^{r(s-w)}}{e^r-\zeta} e^{-x(s-w')} \frac{e^{\frac{t}{2}G(s)}}{e^{\frac{t}{2}G(w)}}.
\end{equation*}
The change of variables $w \mapsto w+1/2$ and $s\mapsto s+1/2$ yields
\begin{equation*}
\EE\left[e^{\zeta Z(t,0)}\right] =  \det(I+K^2_{\zeta})_{L^2(-\frac{1}{2}+ \I\R,\frac{d\mu}{2\pi \I})}
\end{equation*}
with
\begin{equation*}
K^2_{\zeta}(w,w') =\int_0^{\infty} dx \int_{-\infty}^{\infty} dr \int_{\frac{1}{4}-\I \infty}^{\frac{1}{4} + \I\infty} \frac{ds}{2\pi \I} \frac{\zeta e^{r(s-w)}}{e^r-\zeta} e^{-x(s-w')} \frac{e^{\frac{t}{2}G(s+\frac{1}{2})}}{e^{\frac{t}{2}G(w+\frac{1}{2})}}.
\end{equation*}
We may factor the operator $K^2_{\zeta}(w,w') = \int_0^{\infty} dx A(w,x)B(x,w')$ where
\begin{eqnarray*}
A(w,x) &=& \int_{-\infty}^{\infty} dr \int_{\frac{1}{4}-\I \infty}^{\frac{1}{4} + \I\infty} \frac{ds}{2\pi \I} \frac{\zeta}{e^r-\zeta} e^{-xs} \frac{e^{\frac{t}{2}G(s+\frac{1}{2})}}{e^{\frac{t}{2}G(w+\frac{1}{2})}},\\
B(x,w) &=& e^{xw}.
\end{eqnarray*}
Determinants satisfy a fundamental identity (see \cite{Deift} for details and uses) that $\det(I+AB) = \det(I+BA)$. This enables us to rewrite
\begin{equation*}
\EE\left[e^{\zeta Z(t,0)}\right] =  \det(I+K^3_{\zeta})_{L^2([0,\infty))}
\end{equation*}
where
\begin{equation*}
K^3_{\zeta}(x,x') =\int_{-\infty}^{\infty} dr \frac{\zeta}{e^r-\zeta} \int_{-\frac{1}{2}-\I \infty}^{-\frac{1}{2} + \I\infty} \frac{dw}{2\pi \I} \int_{\frac{1}{4}-\I \infty}^{\frac{1}{4} + \I\infty} \frac{ds}{2\pi \I} \frac{F_x(s)}{F_{x'}(w)}
\end{equation*}
and
\begin{equation*}
F_x(s) = \exp\left\{ \frac{t}{2}\,\frac{s^3}{3} + s\left(r-x-\frac{t}{24}\right)\right\}.
\end{equation*}

Finally, we perform the change of variables
\begin{eqnarray}\label{changes1}
s \mapsto \left(\frac{t}{2}\right)^{-1/3} s,  \qquad w \mapsto \left(\frac{t}{2}\right)^{-1/3} w, \qquad x \mapsto \left(\frac{t}{2}\right)^{1/3} x,\qquad x' \mapsto \left(\frac{t}{2}\right)^{1/3}  x', \\
\nonumber r\mapsto \left(\frac{t}{2}\right)^{1/3} r + \frac{t}{24}, \qquad \log(-\zeta) \mapsto -\left(\frac{t}{2}\right)^{1/3} u + \frac{t}{24}.
\end{eqnarray}
Noting that
$$
\frac{\zeta}{e^r-\zeta} = \frac{-1}{1+ e^{r-\log(-\zeta)}}
$$
and taking into account the Jacobian contribution from the change of variables (and rescaling the contours), we arrive at
\begin{equation*}
\EE\left[e^{\zeta Z(t,0)}\right] =  \det(I-K^4_{\zeta})_{L^2([0,\infty))}
\end{equation*}
where
\begin{eqnarray*}
K^4_{\zeta}(x,x') =\int_{-\infty}^{\infty} d r \frac{1}{1 + \exp\left\{(\frac{t}{2})^{1/3}(r + u)\right\}}\qquad\qquad\qquad\qquad\qquad\qquad\qquad\qquad\qquad\qquad\qquad\qquad&\\
\times\,\int_{\frac{1}{4}-\I \infty}^{\frac{1}{4} + \I\infty} \frac{d s}{2\pi \I}  \exp\left\{\frac{s^3}{3}+ s( r-x)\right\}  \int_{-\frac{1}{2}-\I \infty}^{-\frac{1}{2} + \I\infty} \frac{d w}{2\pi \I}  \exp\left\{-\frac{ w^3}{3}- w( r-x')\right\},&
\end{eqnarray*}
and $u$ is related to $\zeta$ as in (\ref{changes1}).

For any $\delta>0$, the Airy function has the following contour integral representation
\begin{equation}\label{Airyintrep}
\Ai(v) = \int_{\delta -\I\infty}^{\delta+\I \infty} \frac{dz}{2\pi \I} \exp\left\{\frac{z^3}{3} - vz\right\},
\end{equation}
from which it follows that
\begin{equation*}
K^4_{\zeta}( x, x') =\int_{-\infty}^{\infty} d  r \frac{1}{1 + \exp\left\{(\frac{t}{2})^{1/3}( r + u)\right\}} \Ai( x-r)\, \Ai( x'- r).
\end{equation*}

At this point we have reached the final form of our formula, cf. Section \ref{fullspacelapsec} above.
 %and comparison with the mathematically rigorous work of \cite{ACQ,BCF} shows that the replica approach has indeed recovered the correct Laplace transform formula.

\subsection{Half-space polymer}
We seek here to compute the Laplace transform of $Z^0(t,0)$. We do so via formally ``summing'' the divergent moment generating function (\ref{mgf}) in a similar way as worked for the full-space polymer. For convenience of our formulas, we replace $\zeta$ in that formula with $\zeta/4$. Using (\ref{Bsupexp}) we find that
\begin{equation*}
\frac{\zeta^k}{4^k k!} \EE[Z^0(t,0)^k] = \sum_{L=1}^{\infty} \frac{(-1)^{L}}{L!} \sum_{\lambda_1=1}^{\infty} \cdots \sum_{\lambda_L=1}^{\infty} \mathbf{1}_{\lambda_1+\cdots + \lambda_L = k} \int_{-\I \R}^{\I \R} \frac{dw_1}{2\pi \I} \cdots \int_{-\I\infty}^{\I \infty}  \frac{dw_L}{2\pi \I} I^0_{L}(\lambda;w;\zeta)
\end{equation*}
where
\begin{equation*}
I^0_{L}(\lambda;w;\zeta)=  \Pf\left[\frac{u_i-u_j}{u_i+u_j}\right]_{i,j=1}^{2L} \prod_{j=1}^{L} \frac{\zeta^{\lambda_j}}{4} \frac{1}{w_j-\tfrac{1}{2}} \frac{\tilde F(w_j+\lambda_j)}{\tilde F(w_j)}
\end{equation*}
where the $u$ variables are written in terms of the $w$ variables as in (\ref{uabove}) and
\begin{equation*}
\tilde F(w) =\frac{\Gamma(w-\frac{1}{2})}{\Gamma(w)} e^{\frac{t}{2} G(w)}
\end{equation*}
with $G(w)$ as in (\ref{Gdef}). In addition to the same considerations as in the full-space case, presently this step also relies upon the fact that the Pochhammer symbol $(x)_i$ can be expressed through Gamma functions as $\Gamma(x+i)/\Gamma(x)$.

By summing over $k$ (and performing infinite range reordering within our already divergent series) we find that
\begin{equation}\label{Balmostfreddet}
\EE\left[e^{\frac{\zeta}{4} Z^0(t,0)}\right] = 1+ \sum_{L=1}^{\infty} \frac{(-1)^L}{L!} \sum_{\lambda_1=1}^{\infty} \cdots \sum_{\lambda_L=1}^{\infty}  \int_{-\I \infty}^{\I \infty}  \frac{dw_1}{2\pi \I} \cdots \int_{-\I\infty}^{\I \infty}  \frac{dw_L}{2\pi \I} I^0_{L}(\lambda;w;\zeta).
\end{equation}
The indicator function disappeared as a result of summing over $k$ and the above formula is still only formal, the right-hand side is divergent.

\subsubsection{The Mellin-Barnes trick}

Notice that for $w_1,\ldots, w_L$ on $\I \R$, the function $I^0_{L}(\lambda;w;\zeta)$ is analytic in each $\lambda_1,\ldots,\lambda_L$ as long as their real part stays positive. We use this fact in order to apply the same Mellin-Barnes trick as in the full-space case in order to replace the summations of $\lambda_1,\ldots, \lambda_L$ by an $L$-fold integral:
\begin{equation}\label{Bsumdiv}
\sum_{\lambda_1=1}^{\infty} \cdots \sum_{\lambda_L=1}^\infty I^0_{L}(\lambda_1,\ldots, \lambda_L ;w;\zeta) = \int_{\frac{3}{4}-\I \infty}^{\frac{3}{4} + \I\infty} \frac{d\lambda_1}{2\pi \I} \cdots \int_{\frac{3}{4}-\I \infty}^{\frac{3}{4} + \I\infty} \frac{d\lambda_1}{2\pi \I} \, \prod_{j=1}^{L} \frac{\pi}{\sin(-\pi \lambda_j)}\, I^0_{L}(\lambda;w;-\zeta).
\end{equation}
As before this is only a formal (in the $\zeta$ variable) residue expansion due to growth of the integrand near infinity.

Also as before, this Mellin-Barnes trick has turned our divergent series into a convergent integral:
\begin{equation}\label{Balmostfreddet2}
\EE\left[e^{\frac{\zeta}{4} Z^0(t,0)}\right] = 1+ \sum_{L=1}^{\infty} \frac{(-1)^L}{L!} \int_{-\I \infty}^{\I \infty}  \frac{dw_1}{2\pi \I} \cdots \int_{-\I\infty}^{\I \infty}  \frac{dw_L}{2\pi \I} \int_{\frac{3}{4}-\I \infty}^{\frac{3}{4} + \I\infty} \frac{d\lambda_1}{2\pi \I} \cdots \int_{\frac{3}{4}-\I \infty}^{\frac{3}{4} + \I\infty} \frac{d\lambda_1}{2\pi \I} \, \prod_{j=1}^{L} \frac{\pi}{\sin(-\pi \lambda_j)}\, I^0_{L}(\lambda;w;-\zeta).
\end{equation}
Both the left-hand and right-hand sides of the above formula make sense now, though we have only shown it in a formal manner. As before, it is useful to perform a few additional manipulations to the right-hand side in order to prepare the formula for large time asymptotics.

\subsubsection{Manipulating the formula}
We employ essentially the same two identities as in the full-space (the first valid for $\Real(\zeta)<0$ and the second valid for $\Real(u_i+u_j)>0$)
$$
\frac{\pi}{\sin(-\pi \lambda)}\, (-\zeta)^{\lambda} = \int_{-\infty}^{\infty}dr \frac{\zeta e^{\lambda r}}{e^{r}-\zeta}, \qquad \qquad \frac{u_i-u_j}{u_i+u_j} = (u_i-u_j) \int_0^{\infty} dx e^{-x(u_i+u_j)},
$$
and perform the change of variables $w_j+\lambda_j = s_j$ to rewrite
\begin{eqnarray}\label{Bfreddet}
\nonumber \EE\left[e^{\frac{\zeta}{4} Z^0(t,0)}\right] &=& 1+ \sum_{L=1}^{\infty} \frac{(-1)^L}{L!} \int_{-\infty}^{\infty} dr_1 \cdots \int_{-\infty}^{\infty} dr_L \int_{-\I \infty}^{\I \infty}  \frac{dw_1}{2\pi \I} \cdots \int_{-\I\infty}^{\I \infty}  \frac{dw_L}{2\pi \I} \int_{\tfrac{3}{4}-\I \infty}^{\tfrac{3}{4}+\I \infty}  \frac{ds_1}{2\pi \I} \cdots \int_{\tfrac{3}{4}-\I\infty}^{\tfrac{3}{4}+\I \infty}  \frac{ds_L}{2\pi \I} \\
&&\times\,\prod_{j=1}^{L} \frac{ \zeta e^{r_j(s_j-w_j)}}{e^{r_j}-\zeta}\, \frac{1}{4} \frac{1}{w_j-\tfrac{1}{2}} \frac{\tilde F(s_j)}{\tilde F(w_j)}  \Pf\left[(u_i-u_j)\int_0^{\infty} dx e^{-x(u_i+u_j)}\right]_{i,j=1}^{2L}.
\end{eqnarray}
Performing the change of variables $w \mapsto w+1/2$ and $s\mapsto s+1/2$ as well as using properties of the Pfaffian to bring the $s$ and $w$ integrals inside the Pfaffian, we find that this can be rewritten as
\begin{equation}\label{preasybform}
\EE\left[e^{\frac{\zeta}{4} Z^0(t,0)}\right] =  1 + \sum_{L=1}^{\infty} \frac{(-1)^L}{L!}\int_{-\infty}^{\infty} dr_1\cdots \int_{-\infty}^{\infty} dr_L \prod_{L=1}^{k} \frac{\zeta}{e^{r_j}-\zeta} \Pf\left[ K^{(t)}(r_i,r_j)\right]_{i,j=1}^{L}
\end{equation}
where $K^{(t)}$ is a $2\times 2$ matrix with components
\begin{eqnarray*}
K^{(t)}_{11}(r,r') &=& \frac{1}{4} \int_0^{\infty} dx \int_{-\tfrac{1}{2} -\I\infty}^{-\tfrac{1}{2}+\I\infty} \frac{dw_1}{2\pi \I} \int_{-\tfrac{1}{2} -\I\infty}^{-\tfrac{1}{2}+\I\infty} \frac{dw_2}{2\pi \I} \frac{-w_1+w_2}{w_1 w_2} \frac{1}{F^{(t)}(w_1) F^{(t)}(w_2)} e^{-r w_1-r'w_2} e^{xw_1+xw_2},\\
K^{(t)}_{12}(r,r') &=& \frac{1}{4} \int_0^{\infty} dx \int_{-\tfrac{1}{2} -\I\infty}^{-\tfrac{1}{2}+\I\infty} \frac{dw}{2\pi \I} \int_{\tfrac{1}{4} -\I\infty}^{\tfrac{1}{4}+\I\infty} \frac{ds}{2\pi \I} \frac{-w-s}{w} \frac{F^{(t)}(s)}{F^{(t)}(w)} e^{-r w+r's} e^{xw-xs},\\
K^{(t)}_{22}(r,r') &=& \frac{1}{4} \int_0^{\infty} dx \int_{\tfrac{1}{4} -\I\infty}^{\tfrac{1}{4}+\I\infty} \frac{ds_1}{2\pi \I}\int_{\tfrac{1}{4} -\I\infty}^{\tfrac{1}{4}+\I\infty} \frac{ds_2}{2\pi \I} (s_1-s_2) F^{(t)}(s_1)F^{(t)}(s_2) e^{r s_1+r's_2} e^{-xs_1-xs_2},
\end{eqnarray*}
and $K^{(t)}_{21}(r,r') = - K^{(t)}_{12}(r',r)$. In the above we have used
\begin{equation*}
F^{(t)}(w)  =\frac{\Gamma(w)}{\Gamma(w+\frac{1}{2})} e^{\frac{t}{2} \left(\frac{w^3}{3} - \frac{w}{12}\right)}.
\end{equation*}

This is the final form of our Laplace transform formula, cf. Section \ref{halfspacelapsec} above.

\subsubsection{Long-time asymptotics}\label{RHSder}
We perform the necessary asymptotics in order to show the results claimed  in Section \ref{hsplim}. In that section we rewrote the half-space Laplace transform in a suggestive manner for taking the $t\to \infty$ asymptotics:
\begin{equation*}
\EE\left[  \exp\left\{-\frac{1}{4}\exp\left\{ \left(\frac{t}{2}\right)^{1/3}\left[\frac{\log Z^0(t,0)  + \frac{t}{24}}{\left(\frac{t}{2}\right)^{1/3}} \, -u \right]\right\}\right\} \right].
\end{equation*}
This amounts to taking $\zeta$ in (\ref{preasybform}) and replacing it by $\log(-\zeta) \mapsto -\left(\frac{t}{2}\right)^{1/3} u + \frac{t}{24}$. Without changing the value of the right-hand side of (\ref{preasybform}) we may perform the following change of variables as well
\begin{eqnarray*}
s \mapsto \left(\frac{t}{2}\right)^{-1/3} s,  \qquad w \mapsto \left(\frac{t}{2}\right)^{-1/3} w, \qquad r\mapsto \left(\frac{t}{2}\right)^{1/3} r + \frac{t}{24},\\
x \mapsto \left(\frac{t}{2}\right)^{1/3} x, \qquad x' \mapsto \left(\frac{t}{2}\right)^{1/3}  x'.
\end{eqnarray*}
We may (using Cauchy's theorem and the decay of the integrand) deform our rescaled contours back to their original locations without crossing any poles. Notice that under the change of variables,
$$
 F^{(t)}(w)e^{(r-x)w} \mapsto  \frac{\Gamma\left(\left(\frac{t}{2}\right)^{-1/3} w\right)}{\Gamma\left(\left(\frac{t}{2}\right)^{-1/3} w + \frac{1}{2} \right)} e^{\frac{w^3}{3}+(r-x)w}.
$$
The decay of the exponential term along the vertical contours of integration implies that it suffices to consider the termwise limits of these Gamma function factors as $t\to\infty$. The function above becomes behaves (for large $t$) like
$$
\frac{\alpha}{\left(\frac{t}{2}\right)^{-1/3} w} e^{\frac{w^3}{3}+(r-x)w}, \qquad \alpha=\frac{1}{\Gamma(\tfrac{1}{2})}=\pi^{-1/2}.
$$
We may combine this limiting behavior with the Jacobian contribution coming from the change of variables to find that

\begin{equation*}
\lim_{t\to \infty} \EE\left[e^{\frac{\zeta}{4} Z^0(t,0)}\right] = 1 + \sum_{L=1}^{\infty} \frac{(-1)^L}{L!} \int_{-\infty}^{-u} dr_1 \cdots \int_{-\infty}^{-u} dr_{L} \Pf\left[\tilde K^{\infty}(r_i,r_j)\right]_{i,j=1}^{L}
\end{equation*}
where $\tilde K^{\infty}$ is a $2\times 2$ matrix with components
\begin{eqnarray*}
\tilde K^{\infty}_{11}(r,r') &=&-\alpha^{-2} \frac{1}{4} \int_0^{\infty} dx \int_{-\tfrac{1}{2} -\I\infty}^{-\tfrac{1}{2}+\I\infty} \frac{dw_1}{2\pi \I} \int_{-\tfrac{1}{2} -\I\infty}^{-\tfrac{1}{2}+\I\infty} \frac{dw_2}{2\pi \I} (-w_1+w_2) e^{-\frac{w_1^3}{3} - w_1(r-x)} e^{-\frac{w_2^3}{3}-w_2(r'-x)},\\
\tilde K^{\infty}_{12}(r,r') &=& -\frac{1}{4} \int_0^{\infty} dx \int_{-\tfrac{1}{2} -\I\infty}^{-\tfrac{1}{2}+\I\infty} \frac{dw}{2\pi \I} \int_{\tfrac{1}{4} -\I\infty}^{\tfrac{1}{4}+\I\infty} \frac{ds}{2\pi \I} \frac{-w-s}{s} e^{-\frac{w^3}{3} - w(r-x)} e^{\frac{s^3}{3} + s(r'-x)},\\
\tilde K^{\infty}_{22}(r,r') &=& -\alpha^2 \frac{1}{4} \int_0^{\infty} dx \int_{\tfrac{1}{4} -\I\infty}^{\tfrac{1}{4}+\I\infty} \frac{ds_1}{2\pi \I}\int_{\tfrac{1}{4} -\I\infty}^{\tfrac{1}{4}+\I\infty} \frac{ds_2}{2\pi \I} \frac{s_1-s_2}{s_1s_2} e^{\frac{s_1^3}{3}+s_1(r-x)} e^{\frac{s_2^3}{3} +s_2(r'-x)},
\end{eqnarray*}
and $\tilde K^{\infty}_{21}(r,r') = - \tilde K^{\infty}_{12}(r',r)$. Note that $\alpha$ is defined above. However, using the expansion to matrix elements for the Pfaffian, it becomes clear that the Pfaffian does not depend on the value of $\alpha$. It will be convenient (in order to match this expression to the known form of the GSE Tracy-Widom distribution) to define a final kernel $K^{\infty}$ in which $\alpha=-1$ and we have taken $r\mapsto -r$. This leads to the kernel $K^{\infty}$ in (\ref{Kkern}) as well as the result recorded in (\ref{almostGSE}).

\subsubsection{Recognizing the GSE Tracy-Widom distribution}\label{RHSdist}

The GSE Tracy-Widom distribution $F_{{\rm GSE}}(u)$ can be expressed via the following formula (e.g., page 124 of \cite{DeiftG}):
\begin{equation*}
F_{{\rm GSE}}(u)  = \sqrt{\det(I-K^{(4)})_{L^2(u,\infty)}}
\end{equation*}
where the determinant above is given by the following infinite series
\begin{equation*}
\det(I-K^{(4)})_{L^2(u,\infty)} = 1+ \sum_{L=1}^{\infty} \frac{(-1)^L}{L!} \int_u^{\infty} dr_1\cdots \int_r^{\infty} dr_L \det\left(K^{(4)}(r_i,r_j)\right)_{i,j=1}^{L}
\end{equation*}
with $K^{(4)}(r,r')$ a $2\times 2$ matrix kernel with entries
\begin{eqnarray*}
K_{11}^{(4)} (r,r') &=& K_{22}^{(4)} (r',r) = \frac{1}{2} K_{\Ai}(r,r') - \frac{1}{4} \Ai(r)\, \int_{r'}^{\infty} \Ai(t) dt\\
K_{12}^{(4)} (r,r') &=& -\frac{1}{2} \frac{d}{dr'} K_{\Ai}(r,r') - \frac{1}{4} \Ai(r) \Ai(r')\\
K_{21}^{(4)} (r,r') &=& -\frac{1}{2} \int_{r}^{\infty} K_{\Ai}(t,r')dt + \frac{1}{4} \int_{r}^{\infty} \Ai(t) dt \, \int_{r'}^{\infty} \Ai(t)dt,
\end{eqnarray*}
and $K_{22}^{(4)} (r,r')= K_{11}^{(4)} (r',r)$. Here $K_{\Ai}$ is the Airy kernel
$$K_{\Ai}(r,r') = \int_{0}^{\infty} dx \Ai(x+r)\Ai(x+r').$$
We note two facts which can be easily checked and which will be useful very soon. The first is that
\begin{equation}\label{kdiff}
\left(\frac{d}{dr}+\frac{d}{dr'}\right)K_{\Ai}(r,r') = -\Ai(r)\Ai(r'),
\end{equation}
and the second is that the different elements of the kernel for $K^{(4)}$ are related via
\begin{equation}\label{k4rel}
K^{(4)}_{12}(r,r') = - \frac{d}{dr'} K^{(4)}_{11}(r,r'),\qquad K^{(4)}_{21}(r,r') = -\int_{r}^{\infty} K^{(4)}_{11}(t,r')dt.
\end{equation}

Recall that if $M$ is  a $2n\times 2n$ dimensional skew symmetric matrix, then $\sqrt{\det(M)} = \Pf(M)$. Letting
$$J = \left(
        \begin{array}{cc}
          0 & 1 \\
          -1 & 0 \\
        \end{array}
      \right)
$$
and noting that $\det(J) =1$ and that $K^{(4)}J$ is skew symmetric, let us perform the following formal manipulations:
$$
F_{{\rm GSE}}(u)= \sqrt{\det(I-K^{(4)})_{L^2(u,\infty)}} = \sqrt{\det(IJ-K^{(4)}J)_{L^2(u,\infty)}} = \Pf(IJ-K^{(4)}J)_{L^2(u,\infty)}.
$$
What is meant by this last expression is the following infinite series expansion
$$
1+ \sum_{L=1}^{\infty} \frac{(-1)^L}{L!} \int_u^{\infty}dr_1\cdots \int_u^{\infty}dr_L \Pf\left(K^{(4)}J(r_i,r_j) \right)_{i,j=1}^{L}.
$$
Though the above manipulation was just formal, it is possible to prove that
$$
F_{{\rm GSE}}(u)=1+ \sum_{L=1}^{\infty} \frac{(-1)^L}{L!} \int_u^{\infty}dr_1\cdots \int_u^{\infty}dr_L \Pf\left(K^{(4)}J(r_i,r_j) \right)_{i,j=1}^{L}
$$
by observing that the right-hand side is the inclusion exclusion formula for the distribution of the top particle of a point process with correlation functions given by the individual Pfaffian expressions. Since these correlation functions correspond with the edge scaling limit of the GSE (cf. \cite[Chapter 9.7]{Forrester}), we find that the right-hand side is exactly the limit distribution of the edge of the GSE -- that is to say, it is $F_{{\rm GSE}}(u)$.

At this point, compare the above formula for $F_{{\rm GSE}}(u)$ to the formula derived in (\ref{almostGSE}). We would like to show that the $2\times 2$ matrix kernel $K^{\infty}(r,r')$ in (\ref{almostGSE}) exactly match the kernel $K^{(4)}J(r,r')$ above. In other words, we seek to check the following four equalities:
\begin{eqnarray}\label{4equals}
K^{\infty}_{11}(r,r') = -K^{(4)}_{12}(r,r'),\qquad K^{\infty}_{12}(r,r') = K^{(4)}_{11}(r,r'),\\
\nonumber K^{\infty}_{21}(r,r') = -K^{(4)}_{22}(r,r'), \qquad K^{\infty}_{22}(r,r') = K^{(4)}_{21}(r,r').
\end{eqnarray}
The middle two are equivalent. We show these by explicitly checking the first equality, and then checking that the relations in (\ref{k4rel}) hold for $K^{\infty}$ in the same way as for $K^{(4)}$.

From the integral representation of the Airy function given in (\ref{Airyintrep}) we see that
\begin{eqnarray*}
K^{\infty}_{11}(r,r') &=& \frac{1}{4} \left(\frac{d}{dr'} - \frac{d}{dr}\right) K_{\Ai}(r,r')\\
             &=& \frac{1}{4} \left(2\frac{d}{dr'} K_{\Ai}(r,r') + \Ai(r)\Ai(r')\right)\\
             &=& -K^{(4)}_{12}(r,r').
\end{eqnarray*}
In going from the first line to the second line we utilized (\ref{kdiff}) and the second to third line was by the definition of $K^{(4)}_{12}(r,r')$ and the symmetry of $K_{\Ai}(r,r')$. It is readily confirmed from (\ref{Kkern}) that
$$
\frac{d}{dr'} K^{\infty}_{12}(r,r') = K^{\infty}_{11}(r,r'),\qquad -\int_r^{\infty} K^{\infty}_{12}(t,r')dt = K^{\infty}_{22}(r,r').
$$
These are the same relations as those satisfied by the elements of $K^{(4)}$, cf. (\ref{k4rel}).
%$$
% - \frac{d}{dr'} K^{(4)}_{11}(r,r') = K_{12}(r,r'),\qquad \int_r^{\infty} K_{11}(r,r') = K_{21}(r,r').
%$$
This implies\footnote{Strictly speaking the above only shows that $\frac{d}{dr'}K^{\infty}_{12}(r,r') = \frac{d}{dr'} K^{(4)}_{11}(r,r')$; hence the two functions of $r'$ could (in principle) differ by a constant. However, it is easy to compare the two functions as $r'$ goes to infinity and see that they are equal there, hence the constant is zero and the equality holds.} that all of the equalities in (\ref{4equals}) are satisfied and hence proves the desired equality of (\ref{almostGSE}) with $F_{{\rm GSE}}(u)$.

\section{Proofs of residue expansion results}\label{proofsec}
In this section we use residue calculus (briefly reviewed in Section \ref{resreview} along with necessary notation for the proofs) to prove Proposition \ref{Akcprop} and provide a partial proof of Conjecture \ref{Bkcprop}, modulo \ref{claim:Res-are-zero}. A similar result as given in Proposition \ref{Akcprop} was proved in \cite[Proposition 3.2.1]{BorCor}, though by a considerably more involved inductive argument. The proofs given below are relatively elementary and in the full-space case (this type of proof was essentially given previously in \cite[Lemma 7.3]{BCPS1}). A similar structure to the proof exists in both cases so we include the full-space proof mostly to help motivate the considerably harder half-space case. The proof also readily extends to a suitable $q$-deformed versions of the results (in fact, \cite[Proposition 3.2.1]{BorCor} and \cite[Lemma 7.3]{BCPS1} were proved for the full-space case at this $q$-deformed level). To illustrate this, in the type $A$ (full-space) case we prove the $q$-deformed proposition and conclude via a limit as $q\to 1$ the desired Proposition \ref{Akcprop}. In the type $BC$ (half-space) case, we prove Conjecture \ref{Bkcprop}, modulo Claim \ref{claim:Res-are-zero}, directly (though the $q$-deformed argument works quite similarly).

\subsection{Residue calculus review and notation}\label{resreview}

As almost all of the formulas within this paper deal with complex contour integrals, it is no surprise that the Cauchy and residue theorems play prominent roles in our calculations. A reader entirely unfamiliar with these tools of complex analysis is referred to \cite[Chapters 2 and 4]{Alhfors}. We briefly recall some consequences of these two theorems (so as to also fix notation for later applications).

For the present, all curves we consider in $\C$ are positively oriented, simple, smooth and closed (e.g. circles) or a infinite straight lines. In many of our applications we consider horizontal translates of the infinite contour $\I \R$ and in order to justify contour deformations as below it is necessary to bound the contribution of the integral near $\pm \I\infty$ as negligible and use Cauchy's theorem on the finite approximation to the contours, before repairing them to be infinite (with a similar tail bound). As this is fairly standard, we do not dwell on it.

%A {\it deformation} (or homotopy) of a curve $\gamma$ to another curve $\gamma'$ is a continuous map $D$ from $[0,1]$ to the space of curves such that $D(0)=\gamma$ and $D(1)=\gamma'$. A deformation crosses a point $p$ if for some unique $s\in[0,1]$, $p\in D(s)$. The set of such points $p$ crossed by a deformation constitutes the {\it area swept out by $D$}.

Given a deformation $D$ of a curve $\gamma$ to another curve $\gamma'$ and a function $f(z)$ which is analytic in a neighborhood the area swept out by the deformation $D$, Cauchy's theorem implies that
\begin{equation*}
\int_{\gamma}\frac{dz}{2\pi \I} f(z) = \int_{\gamma'}\frac{dz}{2\pi \I}   f(z).
\end{equation*}
This means that we can freely deform the contour $\gamma$ to $\gamma'$ without changing the value of the integral of $f(z)$, so long as no points of non-analyticity of $f(z)$ are crossed in the process.

For concreteness, let us now fix that $\gamma$ is a circle and $\gamma'$ is a second circle which is contained within $\gamma$. Consider a function $f(z)$ which is analytic in a neighborhood of the region between these circles, except at a finite collection of singularities (i.e. poles) $a_j$. The residue theorem along with the Cauchy theorem implies that
\begin{equation*}
\int_{\gamma}\frac{dz}{2\pi \I}  f(z)= \int_{\gamma'}\frac{dz}{2\pi \I}  f(z)+ \sum_j \Res{z=a_j}\,  f(z) ,
\end{equation*}
which is to say that we {\it expand} our integral into an integral on the smaller contour and residue terms from crossing the poles. Here $\Res{z=a}f(z)$ is the residue of $f(z)$ at $a$. In general, this is defined as the unique complex number $R$ for which $f(z)-R/(z-a)$ is the derivative of a single valued analytic function in some annulus $0<|z-a|<\delta$. For a simple pole $a$, the residue can be computed by
$$\Res{z=a}\, f(z)  = \lim_{z\to a} (z-a) f(z).$$
In particular, if $f(z) = g(z)/(z-a)$ where $g(z)$ is analytic at $a$, we find that
\begin{equation}\label{ressubeqn}
\Res{z=a}\, f(z)  = \Res{z=a}\, \frac{1}{z-a} \,\, \Sub{z=a}\, g(z) = g(a),
\end{equation}
where we have now also introduced the notation of {\it substitution} of a value into a function.

In our applications we will primarily be interested in functions of several complex variables $z_1,\ldots, z_k$. However, it will be easier to think of all but one of the variables of these functions as being fixed, and then apply the Cauchy / residue theorems to the remaining variable. In this way, we can justify various contour deformations as well as residue expansions. Let us consider the following:
\begin{example}\label{ex1}
Fix $q\in (0,1)$ and consider the integral
$$
\int_{\gamma_1} \frac{dz_1}{2\pi \I} \int_{\gamma_2} \frac{dz_2}{2\pi \I} \,\frac{z_1-z_2}{z_1-qz_2} f(z_1,z_2)
$$
where the $\gamma_2$ contour is along a small circle around $1$ and the $\gamma_1$ contour is another circle around $1$, with radius large enough so as to contain the image of $q$ times the $\gamma_2$ contour. Assume that for every $z_2\in \gamma_2$, the function $z_1\mapsto f(z_1,z_2)$ is analytic in the area between $\gamma_1$ and $\gamma_2$. Notice that for each fixed $z_2\in \gamma_2$ as we deform the $\gamma_1$ contour towards $\gamma_2$, the contour crosses a simple pole at $z_1= qz_2$ (coming from the denominator of the fraction in the integrand). To take this into account, we apply the residue theorem in the $z_1$ variable and find that
\begin{eqnarray*}
\int_{\gamma_1} \frac{dz_1}{2\pi \I} \int_{\gamma_2} \frac{dz_2}{2\pi \I} \frac{z_1-z_2}{z_1-qz_2} f(z_1,z_2) &=&   \int_{\gamma_2} \frac{dz_1}{2\pi \I} \int_{\gamma_2} \frac{dz_2}{2\pi \I} \frac{z_1-z_2}{z_1-qz_2} f(z_1,z_2) + \int_{\gamma_2} \frac{dz_2}{2\pi \I} \Res{z_1=qz_2} \left(\frac{z_1-z_2}{z_1-qz_2} f(z_1,z_2)\right)\\
&=& \int_{\gamma_2} \frac{dz_1}{2\pi \I} \int_{\gamma_2} \frac{dz_2}{2\pi \I} \frac{z_1-z_2}{z_1-qz_2} f(z_1,z_2) + \int_{\gamma_2} \frac{dz_2}{2\pi \I} (qz_2-z_2) f(qz_2,z_2).
\end{eqnarray*}
\end{example}

In what follows we will deal with slightly more involved residue calculations, so it is important to introduce reasonable notation. We will consider residues which occur according to multiplicative strings with parameter $q$ as well as additive strings with parameter $c$. Fix $k\geq 1$ and a partition $\lambda\vdash k$. Then, for a function $f(y_1,\ldots,y_k)$  we define $\Resq{\lambda}f(y_1,\ldots, y_k)$ to be the residue of $f(y_1,\ldots, y_k)$ at
\begin{eqnarray}\label{resqvalues}
\nonumber &y_{\lambda_{1}}=qy_{\lambda_{1}-1},\quad y_{\lambda_{1}-1}=qy_{\lambda_{1}-2}, \quad \ldots, \quad y_{2}=qy_{1}\\
&y_{\lambda_{1}+\lambda_2}=qy_{\lambda_1+\lambda_2-1},\quad  y_{\lambda_{1}+\lambda_2-1}=qy_{\lambda_1+\lambda_2-2}, \quad \ldots, \quad y_{\lambda_{1}+2}=qy_{\lambda_1+1}\\
\nonumber &\vdots &
\end{eqnarray}
%\begin{eqnarray}\label{resqvalues}
%\nonumber &y_{\lambda_{1}}=qy_{\lambda_{1}-1} = \cdots =q^{\lambda_1-1} y_{1}\\
%&y_{\lambda_{1}+\lambda_2}=qz_{\lambda_1+\lambda_2-1}= \cdots =q^{\lambda_2-1} y_{\lambda_1+1}\\
%\nonumber &\vdots &
%\end{eqnarray}
with the output regarded as a function of the terminal variables $\big(y_1,y_{\lambda_1+1},\ldots, y_{\lambda_1+\cdots+\lambda_{\ell(\lambda)-1}-1}\big)$. We call each sequence of identifications of variables a {\it string}. As all poles which we encounter in what follows are simple, the above residue evaluation amounts to
\begin{equation}\label{amts}
\Resq{\lambda}f(y_1,\ldots, y_k) = \lim_{\substack{y_2\to q y_1\\ \cdots\\y_{\lambda_1}\to q y_{\lambda_1-1}}}\prod_{i=2}^{\lambda_1}(y_i-qy_{i-1}) \quad \lim_{\substack{y_{\lambda_1+2}\to q y_{\lambda_1+1}\\ \cdots \\ y_{\lambda_1+\lambda_2}\to q y_{\lambda_1+\lambda_2-1}}}\prod_{i=\lambda_1+2}^{\lambda_1+\lambda_2}(y_i-qy_{i-1})\quad \cdots \,\, f(y_1,\ldots, y_k).
\end{equation}

It is also convenient to define $\Subq{\lambda}f(y_1,\ldots, y_k)$ as the function of $\big(y_1,y_{\lambda_1+1},\ldots, y_{\lambda_1+\cdots+\lambda_{\ell(\lambda)-1}-1}\big)$, which is the result of substituting the relations of (\ref{resqvalues}) into $f(y_1,\ldots, y_k)$.

Under the change of variables $q\mapsto e^{-\e c}$ and $y\mapsto e^{-\e y}$, as $\e\to 0$ the above relations in (\ref{resqvalues}) become
\begin{eqnarray}\label{rescvalues}
\nonumber &y_{\lambda_1} = y_{\lambda_1 - 1} +c, \quad y_{\lambda_1-1} = y_{\lambda_1 - 2} +c, \quad \ldots, \quad y_{2} = y_{1} +c &\\
&y_{\lambda_1+\lambda_2} =  y_{\lambda_1+\lambda_2 - 1} +c, \quad y_{\lambda_1+\lambda_2-1} =  y_{\lambda_1+\lambda_2 - 2} +c  ,\quad \ldots, \quad y_{\lambda_1+2} =  y_{\lambda_1+1} +c&\\
\nonumber &\vdots &
\end{eqnarray}
In the same way as for the $q$ case, we define $\Resc{\lambda}$ and $\Subc{\lambda}$ as the residue and substitution operators with respect to the above strings of relations, which take a function $f(y_1,\ldots, y_k)$ and output a function of $\big(y_1,y_{\lambda_1+1},\ldots, y_{\lambda_1+\cdots+\lambda_{\ell(\lambda)-1}-1}\big)$.

\subsection{Proof of Proposition \ref{Akcprop} (type $A$ expansion)}\label{akproof}

Proposition \ref{Akcprop} can be proved by either taking a limit as $\e\to 0$ of Proposition \ref{321} (stated and proved below) with the change of variables $q\mapsto e^{-\e c}$, $z\mapsto e^{-\e z}$, $w\mapsto e^{-\e w}$, or it can be proved directly by mimicking the proof of Proposition \ref{321}. The two hypotheses of Proposition \ref{321} play analogous roles to those of Proposition \ref{Akcprop}. Additionally, the decay bound of the second hypothesis of Proposition \ref{Akcprop} is necessary because contours are infinite, and in order to apply the Cauchy / residue theorems it is necessary to have suitable decay near infinity. Otherwise, the proof goes through without any significant changes.

Thus, in this section we state and prove the above mentioned $q$-deformed analog of Proposition \ref{Akcprop}. The following proposition is a generalization of Proposition 3.2.1 of \cite{BorCor}, however the proof which we present here is considerably simpler than the inductive one presented therein. This proof also exposes the generality of the result as stated below.

\begin{proposition}\label{321}
Fix $k\geq 1$ and $q\in (0,1)$. Given a set of positively oriented, closed contours $\gamma_1,\ldots,\gamma_k$ and a function $F(z_1,\ldots, z_k)$ which satisfy
\begin{enumerate}
\item For all $1\leq A<B\leq k$, the interior of $\gamma_A$ contains the image of $\gamma_B$ multiplied by $q$;
\item For all $1\leq j\leq k$, there exist deformations $D_j$ of $\gamma_j$ to $\gamma_k$ so that for all $z_1,\ldots, z_{j-1},z_j,\ldots, z_k$ such with $z_i\in \gamma_i$ for $1\leq i<j$, and $z_i\in \gamma_k$ for $j<i\leq k$, the function $z_j\mapsto F(z_1,\ldots ,z_j,\ldots, z_k)$ is analytic in a neighborhood of the area swept out by the deformation $D_j$.
\end{enumerate}
Then we have the following residue expansion identity:
\begin{align}\label{qAK}
(-1)^k q^{\frac{k(k-1)}{2}} \oint_{\gamma_1} \frac{dz_1}{2\pi \I} \cdots \oint_{\gamma_k} \frac{dz_k}{2\pi \I} \prod_{1\leq A<B\leq k} \frac{z_A-z_B}{z_A-q z_B} F(z_1,\ldots, z_k) \qquad\qquad\qquad\qquad\qquad\qquad&\\
\nonumber = \sum_{\substack{\lambda\vdash k\\ \lambda= 1^{m_1}2^{m_2}\cdots}} \frac{(1-q)^{k}}{m_1! m_2!\cdots} \oint_{\gamma_k} \frac{dw_1}{2\pi \I}\cdots \oint_{\gamma_k} \frac{dw_{\ell(\lambda)}}{2\pi \I} \det\left[\frac{1}{w_i q^{\lambda_i} -w_j}\right]_{i,j=1}^{\ell(\lambda)}\qquad\qquad\qquad\qquad&\\
\nonumber \times\, \prod_{j=1}^{\ell(\lambda)} w_j^{\lambda_j} q^{\frac{\lambda_j(\lambda_j-1)}{2}} \, E^q(w_1,qw_1,\ldots, q^{\lambda_1-1}w_1,\ldots, w_{\lambda_{\ell(\lambda)}},q w_{\lambda_{\ell(\lambda)}},\ldots, q^{\lambda_{\ell(\lambda)}-1} w_{\lambda_{\ell(\lambda)}}),&
\end{align}
where
\begin{equation*}
E^q(z_1,\ldots, z_k) =   \sum_{\sigma\in S_k} \prod_{1\leq B<A\leq k} \frac{z_{\sigma(A)}-qz_{\sigma(B)}}{z_{\sigma(A)}-z_{\sigma(B)}} F(\sigma(\vec{z})).
\end{equation*}
\end{proposition}

\begin{proof}[Proof of Proposition \ref{321}]
The meaning of the various terms and notation used below can be found in Section \ref{resreview}. First notice that for $k=1$ the result follows immediately. Hence, in what follows we assume $k\geq 2$.

We proceed sequentially and deform (using the deformation $D_{k-1}$ afforded from the hypotheses of the theorem) the $\gamma_{k-1}$ contour to $\gamma_k$, and then deform (using the deformation $D_{k-2}$ afforded from the hypotheses of the theorem) $\gamma_{k-2}$ contour to $\gamma_k$, and so on until all contours have been deformed to $\gamma_k$. However, due to the $z_A-qz_B$ terms in the denominator of the integrand, during the deformation of $\gamma_A$ we may encounter simple poles at the points $z_A=q z_B$, for $B>A$. The residue theorem (cf. Section \ref{resreview}) implies that the integral on the left-hand side of (\ref{qAK}) can be expanded into a summation over integrals of possibly few variables (all along $\gamma_k$) whose integrands correspond to the various possible residue subspaces coming from these poles.

Our proof splits into three basic steps. First, we identify the residual subspaces upon which our integral is expanded via residues. This brings us to equation (\ref{RHS2}). Second, we show that these subspaces can be brought to a canonical form via the action of some $\sigma\in S_k$. This enables us to simplify the summation over the residual subspaces to a summation over partitions $\lambda\vdash k$ and certain subsets of $\sigma\in S_k$. Inspection of those terms corresponding to $\sigma\in S_k$ not in these subsets shows that they have zero residue contribution and hence the summation can be completed to all of $S_k$. This brings us to equation (\ref{above3}). And third, we rewrite the function whose residue we are computing as the product of an $S_k$ invariant function (which contains all of the poles related to the residual subspace) and a remainder function. We use Lemma \ref{reslemma} to evaluate the residue of the $S_k$ invariant function and we identify the summation over $\sigma\in S_k$ of the substitution into the remainder function as exactly giving the $E^q$ function in the statement of the proposition we are presently proving.

\smallskip
\noindent {\bf Step 1:} It is worthwhile to start with an example. The case $k=2$ is effectively worked out in Example \ref{ex1} so we start with $k=3$. Figure \ref{qcontourshifting} accompanies the example and illustrates the deformations and locations of poles.

\begin{figure}
\begin{center}
\includegraphics[scale=.8]{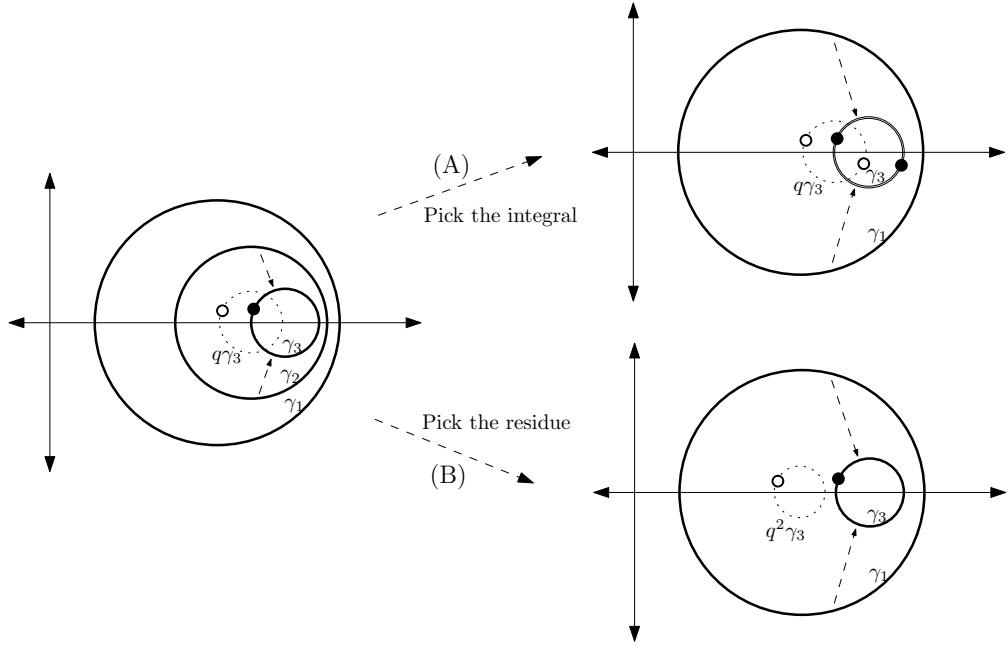}
\end{center}
\caption{The expansion of the $k=3$ nested contour integral as explained in Example \ref{ex2}. On the left-hand side, the $\gamma_2$ contour is deformed to the $\gamma_3$ contour and a pole is crossed along $q\gamma_3$ at the point $qz_3$ (here the point $z_3$ is drawn as a small filled disk and its location is on the solid line circle; and the point $q z_3$ is drawn as a small unfilled disk and its location is on the dotted line circle). On the upper right-hand side the effect of picking the integral is shown. The $\gamma_3$ contour is represented as a (doubled) solid line circle since both $z_3$ and $z_2$ are integrated along it (these correspond to the two filled disks). The unfilled disks are along $q\gamma_3$ (the dotted line circle) and represent $qz_3$ and $qz_2$. As the $\gamma_1$ contour is deformed to $\gamma_3$ these residues must be taken. On the lower right-hand side the effect of picking the initial residue at $z_2=qz_3$ is shown. As the $\gamma_1$ contour is deformed to $\gamma_3$ a pole is encountered along $q^2\gamma_3$ at the location $q^2 z_3$ (as before $z_3$ is the filled disk and $q^2 z_3$ is the unfilled disk).}\label{qcontourshifting}
\end{figure}

\begin{example}\label{ex2}
When $k=3$ the integrand on the left-hand side of (\ref{qAK}) contains the fractions
\begin{equation}\label{cross}
\frac{z_1-z_2}{z_1-qz_2}\, \frac{z_1-z_3}{z_1-qz_3}\,\frac{z_2-z_3}{z_2-qz_3},
\end{equation}
times the function $F(z_1,z_2,z_3)$ which (by the hypotheses of the proposition) does not have poles between $\gamma_j$ and $\gamma_k$ (for $j=1,2,3$). Thinking of $z_3$ as fixed along the contour $\gamma_3$, we begin by deforming the $\gamma_2$ contour to $\gamma_3$. As we proceed, we necessarily encounter a single simple pole at $z_2=qz_3$. The residue theorem implies that our initial integral equals the sum of (A) the integral where $z_2$ is along $\gamma_3$, and (B) the integral with only $z_1$ and $z_3$ remaining and integrand given by taking the residue of the initial integrand at $z_2=qz_3$.

Let us consider separately these two pieces. For (A) we now think of $z_2$ and $z_3$ as fixed along $\gamma_3$ and deform the $\gamma_1$ contour to $\gamma_3$, encountering two simple poles at $z_1=qz_2$ and $z_1=q z_3$. Thus (A) is expanded into a sum of three terms: the integral with $z_1, z_2$ and $z_3$ along $\gamma_3$; and the integral with only $z_2$ and $z_3$ remaining (along $\gamma_3$) and the residue taken at either $z_1=q z_2$ or $z_1=q z_3$.

For (B) we now think of $z_3$ as fixed along $\gamma_3$ and deform the $\gamma_1$ contour to $\gamma_3$, encountering a simple pole at $z_1=q^2 z_3$. This is because the residue of (\ref{cross}) at $z_2=q z_3$ equals
$$
\frac{z_1-z_3}{z_1-q^2 z_3}\,(qz_3-z_3).
$$
Thus (B) is expanded into a sum of two terms: the integral with $z_1$ and $z_3$ (along $\gamma_3$ and with the above expression in the integrand); and the integral with only $z_3$ remaining (along $\gamma_3$) and the residue of the above term taken at $z_1=q^2 z_3$.

Gathering the various terms in this expansion, we see that the residue subspaces we sum over are indexed by partitions of $k$ (here $k=3$) and take the form of geometric strings with the parameter $q$. For example, $\lambda = (1,1,1)$ corresponds to the term in the residue expansion in which all three variables $z_1,z_2$ and $z_3$ are still integrated, but along $\gamma_3$. On the other hand, $\lambda = (3)$ corresponds to the term in which the residue is taken at $z_1=qz_2=q^2z_3$ and the only variable which remains to be integrated along $\gamma_3$ is $z_3$. The partition $\lambda=(2,1)$ corresponds to the three remaining terms in the above expansion in which two integration variables remain. In general, $\ell(\lambda)$ corresponds to the number of variables which remain to be integrated in each term of the expansion.
\end{example}

Let us now turn to the general $k$ case. By the hypotheses of the theorem, the function $F$ has no poles which are encountered during contour deformations -- hence it plays no role in the residue analysis. As we deform sequentially the contours in the left-hand side of (\ref{qAK}) to $\gamma_k$ we find that the resulting terms in the residue expansion can be indexed by partitions $\lambda\vdash k$ along with a list (ordered set) of disjoint ordered subsets of $\{1,\ldots, k\}$ (whose union is all of $\{1,\ldots, k\}$)
\begin{eqnarray}\label{ijvar}
\nonumber &i_{1}< i_{2} < \cdots < i_{\lambda_1}&\\
&j_{1}< j_{2} < \cdots < j_{\lambda_2}&\\
\nonumber & \vdots&
%i_{1,1}< i_{2,1} < \cdots < i_{\lambda_1,1}\\
%i_{1,2}< i_{2,2} < \cdots < i_{\lambda_2,2}\\
%i_{1,\ell(\lambda)}< i_{2,\ell(\lambda)} < \cdots < i_{\lambda_{\ell(\lambda)},\ell(\lambda)}.
\end{eqnarray}
Let us call such a list $I$. For a given partition $\lambda$ call $S(\lambda)$ the collection of all such lists $I$ corresponding to $\lambda$.
For $k=3$ and $\lambda=(2,1)$, in Example \ref{ex2} we saw there are three such lists which correspond with
$$S(\lambda) = \Big\{ \big\{1<2 , 3\big\}, \big\{1<3 ,2\big\}, \big\{2<3 , 1\big\}\Big\}.$$

For such a lists $I$, we write $\Res{I} f(z_1,\ldots, z_k)$ as the residue of the function $f$ at
\begin{eqnarray*}
&z_{i_{1}}=qz_{i_{2}}, z_{i_2} = q z_{i_3}, \quad \ldots, \quad z_{i_{\lambda_1-1}}=q z_{i_{\lambda_1}}&\\
&z_{j_{1}}=qz_{j_{2}}, z_{j_2} = q z_{j_3}, \quad \ldots, \quad z_{j_{\lambda_2-1}}=q z_{j_{\lambda_2}}&\\
&\vdots&
%z_{i_{1,1}}=qz_{i_{2,1}}= \cdots =q^{\lambda_1-1} z_{i_{\lambda_1,1}}\\
%z_{i_{1,2}}=qz_{i_{2,2}}= \cdots =q^{\lambda_2-1} z_{i_{\lambda_2,2}}\\
%\vdots
%z_{i_{1,\ell(\lambda)}}=q z_{i_{2,\ell(\lambda)}} = \cdots = q^{\lambda_{\ell(\lambda)}-1} z_{i_{\lambda_{\ell(\lambda)},\ell(\lambda)}},
\end{eqnarray*}
and regard the output as a function of the terminal variables $(z_{i_{\lambda_1}}, z_{j_{\lambda_2}}, \ldots)$. There are $\ell(\lambda)$ such strings and consequently that many remaining variables (though we have only written the first two strings above).

With the above notation in place, we may write the expansion of the integral on the left-hand side of (\ref{qAK}) as
\begin{eqnarray}\label{RHS2}
{\rm LHS} (\ref{qAK}) = (-1)^k q^{\frac{k(k-1)}{2}} \sum_{\substack{\lambda\vdash k\\ \lambda= 1^{m_1}2^{m_2}\cdots}} \frac{1}{m_1! m_2!\cdots} \sum_{I\in S(\lambda)}\oint_{\gamma_k} \frac{dz_{i_{\lambda_1}}}{2\pi \I} \oint_{\gamma_k} \frac{dz_{j_{\lambda_{2}}}}{2\pi \I} \cdots \\
\nonumber \times\,\Res{I} \left(\prod_{1\leq A<B\leq k} \frac{z_A-z_B}{z_A-q z_B}  F(\vec{z}) \right).
\end{eqnarray}
The factor of $\frac{1}{m_1! m_2!\cdots}$ arose from multiple counting of terms in the residue expansion due to symmetries of $\lambda$. For example, for the partition $\lambda=(2,2,1)$, each $I\in S(\lambda)$, corresponds uniquely with a different $I'\in S(\lambda)$ in which the $i$ and $j$ variables in (\ref{ijvar}) are switched. Since these correspond with the same term in the residue expansion, this constitutes double counting and hence the sum should be divided by $2!$. The reason why our residual subspace expansion only corresponds with strings is because if we took a residue which was not of the form of a string, then for some $A\neq A'$ we would be evaluating the residue at $z_A=qz_B$ and $z_{A'}=qz_B$. However, the Vandermonde determinant in the numerator of our integrand would then necessarily evaluate to zero. Therefore, such possible non-string residues (coming from the denominator) in fact have zero contribution.

\noindent {\bf Step 2:} For each $I\in S(\lambda)$ relabel the $z$ variables as
\begin{eqnarray*}
&(z_{i_{1}}, z_{i_{2}}, \ldots, z_{i_{\lambda_1}}) \mapsto (y_{\lambda_1},y_{\lambda_1-1},\ldots, y_1)&\\
&(z_{j_{1}}, z_{j_{2}}, \ldots, z_{j_{\lambda_2}}) \mapsto (y_{\lambda_1+\lambda_2}, \ldots, y_{\lambda_1+1})&\\
&\vdots&
%(z_{i_{1,1}}, z_{i_{2,1}}, \ldots, z{i_{\lambda_1,1}}) \mapsto (y_{\lambda_1},y_{\lambda_1-1},\ldots, y_1)\\
%(z_{i_{1,2}}, z_{i_{2,2}}, \ldots, z_{i_{\lambda_2,2}}) \mapsto (y_{\lambda_1+\lambda_2}, \ldots, y_{\lambda_1+1})\\
\end{eqnarray*}
and observe that there exists a unique permutation $\sigma\in S_k$ for which $(z_1,\ldots, z_k) = (y_{\sigma(1)},\ldots, y_{\sigma(k)})$. Let us also call $w_j = y_{\lambda_1+\cdots \lambda_{j-1}+1}$, for $1\leq j\leq \ell(\lambda)$.

We rewrite the term corresponding to a partition $\lambda$ in the right-hand side of (\ref{RHS2}) as
\begin{equation}\label{above3}
 \frac{1}{m_1! m_2!\cdots} \sum_{\sigma\in S_k} \oint_{\gamma_k} \frac{dw_1}{2\pi \I} \cdots \oint_{\gamma_k} \frac{dw_{\ell(\lambda)}}{2\pi \I}
\Resq{\lambda} \left(\prod_{1\leq A<B\leq k} \frac{y_{\sigma(A)}-y_{\sigma(B)}}{y_{\sigma(A)}-q y_{\sigma(B)}} F(\sigma(\vec{y})) \right),
\end{equation}
where $\Resq{\lambda}$ is defined in equation (\ref{resqvalues}). Note that the output of the residue operation is a function of the variables $(y_1,y_{\lambda_1+1},\ldots) = (w_1,w_2,\ldots)$.

One should observe that the above expression includes the summation over all $\sigma\in S_k$, and not just those which arise from an $I\in S(\lambda)$ as above. This, however, is explained by the fact that if $\sigma$ does not arise from some $I\in S(\lambda)$, then the residue necessarily evaluates to zero (hence adding these terms is allowed). To see this, observe that in the renumbering of variables discussed above, only permutations with
$$\sigma^{-1}(1)>\sigma^{-1}(2)\cdots>\sigma^{-1}(\lambda_1), \qquad \sigma^{-1}(\lambda_1+1)>\sigma^{-1}(\lambda_1+2)\cdots>\sigma^{-1}(\lambda_1+\lambda_2),\qquad \ldots$$
participated. Any other $\sigma$ must violate one of these strings of conditions. Consider, for example, some $\sigma$ with $\sigma(\lambda_1-1) < \sigma(\lambda_1)$. This implies that the term $y_{\lambda_1-1}-qy_{\lambda_1}$ shows up in the denominator of (\ref{above3}), as opposed to the term $y_{\lambda_1}-qy_{\lambda_1-1}$. Residues can be taken in any order, and if we first take the residue at $y_{\lambda_1} = qy_{\lambda_1-1}$, we find that the above denominator does not have a pole (nor do any other parts of (\ref{above3})) and hence the residue is zero. Similar reasoning works in general.

\noindent {\bf Step 3:} All that remains is to compute the residues in (\ref{above3}) and identify the result (after summing over all $\lambda\vdash k$) with the right-hand side of (\ref{qAK}) as necessary to prove the theorem.

It is convenient to rewrite the following expression as an $S_k$ invariant function, times a function which is analytic at the points in which the residue is being taken:
$$
\prod_{1\leq A<B\leq k} \frac{y_{\sigma(A)}-y_{\sigma(B)}}{y_{\sigma(A)}-q y_{\sigma(B)}}  = \prod_{1\leq A\neq B\leq k} \frac{y_A-y_B}{y_A-qy_B} \prod_{1\leq B<A\leq k} \frac{y_{\sigma(A)} - qy_{\sigma(B)}}{y_{\sigma(A)} - y_{\sigma(B)}}.
$$
Since it is only the $S_k$ invariant function which contains the poles with which we are concerned, it allows us to rewrite (\ref{above3}) as
\begin{eqnarray*}
\frac{1}{m_1! m_2!\cdots} \oint \frac{w_1}{2\pi \I} \cdots \oint \frac{w_{\ell(\lambda)}}{2\pi \I} \Resq{\lambda} \left( \prod_{1\leq A\neq B\leq k} \frac{y_A-y_B}{y_A-qy_B} \right)\\
\nonumber \times\,\Subq{\lambda} \left( \sum_{\sigma\in S_k} \prod_{1\leq B<A\leq k} \frac{y_{\sigma(A)} - qy_{\sigma(B)}}{y_{\sigma(A)} - y_{\sigma(B)}} F(\sigma(\vec{y})) \right),
\end{eqnarray*}
where $\Subq{\lambda}$ is the substitution operator defined in Section \ref{subsec}.
We use Lemma \ref{reslemma} to evaluate the above residue, and we easily identify the substitution on the second line with $E^q(w_1,q w_1,\ldots)$ as in the statement of the proposition.
Combining these two expressions and summing the resulting expression over $\lambda\vdash k$ we arrive at the desired residue expansion claimed in the statement of the proposition.
\end{proof}

\subsection{Proof of Conjecture \ref{Bkcprop} (type $BC$ expansion) modulo Claim \ref{claim:Res-are-zero} }\label{bcproof}

%The proof of Proposition \ref{Akcprop} was given by degeneration of a $q$-deformed version of that proposition. One could easily write down an analogous direct proof of that proposition. For Conjecture \ref{Bkcprop} we write the explicit proof of that proposition. It is possible to formulate and prove an analogous $q$-deformation for this case as well, but as we presently have no application for it, we do not pursue it.

We turn now to Conjecture \ref{Bkcprop}. Note that when $k=1$ the proof follows immediately, hence we assume $k\geq 2$.

The argument follows a similar line to that of Proposition \ref{321}. The type $BC$ symmetry makes the problem more difficult and we must assume at some point a non-trivial residue cancelation in the form of Claim \ref{claim:Res-are-zero}. We do not prove this claim presently, though give a proof of a special case of this in Section \ref{sec:one-string}.  It would, therefore, be advised that the reader first reviews the proof of Proposition \ref{321}, before studying the below argument.

We proceed sequentially and shift the $\alpha_{k-1}+\I \R$ contour to $\alpha_k+\I \R$ which, since $\alpha_k=0$, is just $\I \R$. Then we repeat this procedure for $\alpha_{k-2}$ through $\alpha_1$. Despite these contours being of infinite length, the deformations are easily justified by virtue of the decay to zero of the function $F$ ensured by the hypotheses. We do, however, need to take care in keeping track of the residue contributions, which result from the simple poles when $z_A \pm z_B -c =0$. The residue theorem (see Section \ref{resreview}) implies that the integral on the left-hand side of (\ref{mainform}) can be expanded into a summation over integrals of possibly few variables (all along $\I \R$) whose integrands correspond to the various possible residue subspaces coming from these poles.

Our proof splits into the same three steps as the proof of Proposition \ref{321}, but with the role of $S_k$ replaced by $BC_k$ and with a more complicated collection of residual subspaces. In the first step, equation (\ref{RHS2}) from the proof of Proposition \ref{321} is replaced by equation (\ref{firstclaim}); in the second step equation (\ref{above3}) is replaced by equation (\ref{ref31prime}) -- showing this we appeal to Claim \ref{claim:Res-are-zero}; and in the third step, Lemma \ref{reslemma} is replaced by Lemma \ref{fourthlemma} and the $E^q$ function is replaced by $E^c$ as defined in the statement of Conjecture \ref{Bkcprop}.

\smallskip
\noindent {\bf Step 1:} We start with the example of $k=3$ so as to introduce the basic idea and notation for counting these subspaces. Figure \ref{contourshifting} accompanies the example and illustrates the deformations and locations of poles.

\begin{figure}
\begin{center}
\includegraphics[scale=.57]{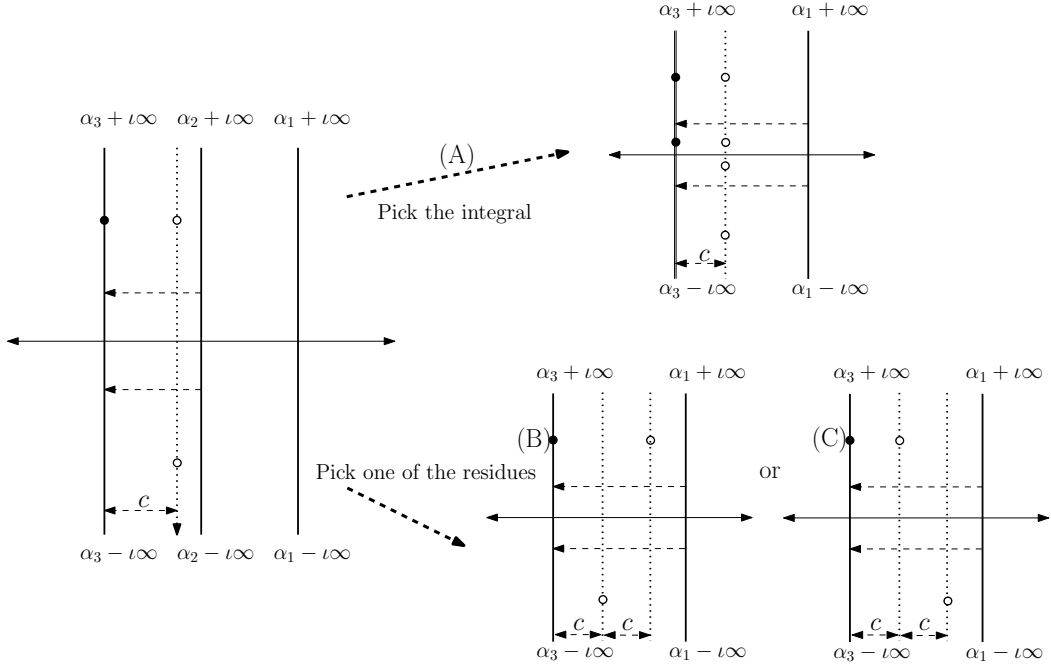}
\end{center}
\caption{The expansion of the $k=3$ nested contour integral as explained in Example \ref{ex3}. Recall that $\alpha_3=0$. On the left-hand side, the contour $\alpha_2+\I \R$ is deformed to $\I \R$ and two poles are crossed at $z_2=\pm z_3+c$ (here the point $z_3$ is drawn as a small filled disk and its location is on the solid line; the points $\pm z_3+c$ are drawn as small unfilled disks and their locations are on the dotted line). On the upper right-hand side the effect of picking the integral is shown (this is case (A) in the example). The $\I \R$ contour on which $z_2$ and $z_3$ are integrated is represented as a doubled solid line and $z_2$ and $z_3$ are represented by the two filled disks. The four unfilled disks are along $c+\I \R$ and represent $\pm z_3+c $ and $\pm z_2+c$. As the $z_1$ contour is deformed from $\alpha_1+\I \R$ to $\I \R$ these residues must be taken. On the lower right-hand side the effect of picking the initial residues at $z_2=z_3+c$ and then at $z_2=-z_3+c$ are shown (these are cases (B) and (C) from the example, respectively). In case (B) as the $z_1$ contour is deformed, poles are encountered at $z_1= z_3 + 2c$ and $z_1=-z_3 +c$, whereas in case (C) poles are encountered at $z_1=-z_3+2c$ and $z_1=z_3+c$.}\label{contourshifting}
\end{figure}

\begin{example}\label{ex3}
For $k=3$, the right-hand side of (\ref{mainform}) contains the fractions
\begin{equation}\label{asabovefrac}
\frac{z_1-z_2}{z_1-z_2-c}\, \frac{z_1+z_2}{z_1+z_2-c}\, \frac{z_1-z_3}{z_1-z_3-c}\, \frac{z_1+z_3}{z_1+z_3-c}\, \frac{z_2-z_3}{z_2-z_3-c}\, \frac{z_2+z_3}{z_2+z_3-c},
\end{equation}
times the function $F(z_1,z_2,z_3)$ which (by the hypotheses of the proposition) does not have poles in the area in which we are deforming contours. Thinking of $z_3$ as fixed along $\I\R$, we begin by deforming the $z_2$ contour, $\alpha_2+\I \R$, to $\I \R$. As we proceed, we necessarily encounter two simple poles at $z_2 = \pm z_3+c$. The residue theorem implies that our initial integral equals the sum of (A) the integral where $z_2$ is along $\I \R$, (B) the integral with only $z_1$ and $z_3$ remaining and integrand given by taking the residue of the initial integrand at $z_2=z_3+c$, and (C) the integral with only $z_1$ and $z_3$ remaining and integrand given by taking the residue of the initial integrand at $z_2=-z_3+c$.

Let us consider separately these three pieces. For (A) the fraction in the integrand remains the same as above in (\ref{asabovefrac}). Thus, as we deform the $z_1$ contour, $\alpha_1+\I \R$, to $\I \R$, we encounter four simple poles at $z_1= \pm z_2 + c$ and $z_1 = \pm z_3+c$. Thus (A) is expanded into five terms (one when none of the residues are picked, and four for each of the residue contributions).

For (B) the fraction in the resulting integrand (after taking the residue at $z_2=z_3+c$ and canceling terms) is given by
$$
c\,\frac{z_1-z_3}{z_1-z_3-2c}\, \frac{z_1+z_3+c}{z_1+z_3-c}\,\frac{2z_3+c}{2z_3}.
$$
Thus, as we deform the $z_1$ contour, $\alpha_1+\I \R$, to $\I \R$, we encounter two simple poles at $z_1= z_3+2c$ and $z_1 = -z_3+c$. Since one should think of $z_2=z_3+c$, the first pole is really at $z_1=z_2+c$. Thus (B) is expanded into three terms.

For (C) the fraction in the resulting integrand (after taking the residue at $z_2=-z_3+c$ and canceling terms) is given by
$$
c\,\frac{z_1+z_3}{z_1+z_3-2c}\, \frac{z_1-z_3+c}{z_1-z_3-c}\,\frac{-2z_3+c}{-2z_3}.
$$
Thus, as we deform the $z_1$ contour, $\alpha_1+\I \R$, to $\I \R$, we encounter two simple poles at $z_1= z_3+c$ and $z_1 = -z_3+2c$. Since one should think of $z_2=-z_3+c$, the second pole is really at $z_1=z_2+c$. Thus (C) is expanded into three terms.

We will be well served by introducing good notation to keep track of the various terms in such a residue expansion. Let us give the example of the notation for the $k=3$ case considered above, and then explain what it means in terms of residues and integrals.

The five terms in (A) correspond with the following diagrams
\begin{equation*}
\substack{\displaystyle 1\\ \\ \displaystyle 2\\ \\ \displaystyle 3} \qquad\qquad
\substack{\displaystyle 1\rightp 2 \\ \\ \displaystyle 3\qquad\,\,\,\,\\ \,\\}\qquad\qquad
\substack{\displaystyle 1\rightm 2 \\ \\ \displaystyle 3\qquad\,\,\,\,\\ \,\\}\qquad\qquad
\substack{\displaystyle 1\rightp 3 \\ \\ \displaystyle 2\qquad\,\,\,\,\\ \,\\}\qquad\qquad
\substack{\displaystyle 1\rightm 3 \\ \\ \displaystyle 2\qquad\,\,\,\,\\ \,\\},
\end{equation*}
while the three terms in (B) correspond with the following diagrams
\begin{equation*}
\substack{\displaystyle 2\rightp 3 \\ \\ \displaystyle 1\qquad\,\,\,\,\\ }\qquad\qquad
\substack{\displaystyle 1\rightp 2 \rightp 3 \\ \\ \\ }\qquad\qquad
\substack{\displaystyle 1\rightm 3 \leftp 2 \\ \\ \\ },
\end{equation*}
while the three terms in (C) correspond with the following diagrams
\begin{equation*}
\substack{\displaystyle 2\rightm 3 \\ \\ \displaystyle 1\qquad\,\,\,\,\\ }\qquad\qquad
\substack{\displaystyle 1\rightp 2 \rightm 3 \\ \\ \\ }\qquad\qquad
\substack{\displaystyle 1\rightp 3 \leftm 2 \\ \\ \\ }.
\end{equation*}

In the above notation, each separate line of a diagram represents a single integral (over the $z_j$ variable where $j$ is the largest number in the line) and each directed arrow from $i$ to $j$ (with $i<j$) represents a residue taken at $z_i =\pm z_j +c$, where the choice of plus or minus is indicated by the label above the arrow. For instance, when $k=4$, the term corresponding with the diagram
\begin{equation*}
\substack{\displaystyle 2\rightm 3\leftp 1 \\ \\ \displaystyle 4\,\,\,\,\,\,\,\qquad\qquad\\ }
\end{equation*}
in the residue expansion of the left-hand side of (\ref{mainform}) is
\begin{equation*}
\int_{-\I \infty}^{\I \infty} \frac{dz_3}{2\pi \I} \int_{-\I \infty}^{\I \infty} \frac{dz_4}{2\pi \I}\, \Res{\substack{z_2 = -z_3+c \\ z_1=z_3+c}} \left(\prod_{1\leq A<B\leq k} \frac{z_A-z_B}{z_A-z_B-c} \,\frac{z_A+z_B}{z_A+z_B-c} F(\vec{z})\right).
\end{equation*}
\end{example}

Let us now turn to the general $k$ case. Using the notation introduced in the above example, let us consider a general set of diagrams. Since it will be useful in proving the residue expansion formula which follows, consider $p\in \{0,\ldots, k\}$. (This number will correspond to the number of contours that have not been deformed to $\I \R$ yet.)
Then for each partition $\lambda\vdash k-p$ we consider diagrams of the form
\begin{equation*}
\substack{\displaystyle i_1\rightp i_2\rightp \cdots \rightp i_{\mu_1-1}\rightpm i_{\mu_1} \leftmp i_{\mu_1+1}\leftp \cdots \leftp i_{\lambda_1}\\ \\
\displaystyle j_1\rightp j_2\rightp \cdots \rightp j_{\mu_2-1}\rightpm j_{\mu_2} \leftmp j_{\mu_2+1}\leftp \cdots \leftp j_{\lambda_2}\\ \\ \displaystyle \vdots}
\end{equation*}
where the following conditions hold: The numbers on each line are pairwise disjoint and the union of the numbers over all lines equals the set $\{p+1,\ldots ,k\}$ (when $p=k$ this is the empty set); the arrows change from pointing right to pointing left only once (call the number at which this change occurs $i_{\mu_1},j_{\mu_2},\ldots$); the directed arrows are always from smaller numbers to larger numbers; within each line all arrows have $+$ above them except for the two (or possibly one if all arrows point in the same direction) surrounding $i_{\mu_1},j_{\mu_2},\ldots$ for which one sign is $+$ and the other is $-$ (this choice may change line to line). We write $S(\lambda)$ for the collection of all such diagrams for the partition $\lambda$, and we write $I$ for one such diagram. %This is the same type of notation we used previously in Section \ref{akproof}. Presently, the type of diagrams (i.e., residual subspaces) we consider becomes slightly more involved, though still manageable.

For such a list $I$ as above, we write $\Res{I} f(z_1,\ldots,z_k)$ as the residue of the function $f$ at
\begin{equation*}
\substack{\displaystyle z_{i_1}= z_{i_2}+c, z_{i_2} =  z_{i_3} +c, \ldots, z_{i_{\mu_1}-1} = \pm z_{i_{\mu_1}} +c,\qquad z_{i_{\lambda_1}} = \mp z_{i_{\lambda_1}-1} +c, \ldots  z_{i_{\mu_1}+1} =  z_{i_{\mu_1}} +c \\ \\
\displaystyle z_{j_1}= z_{j_2}+c, z_{j_2} = z_{j_3} +c, \ldots, z_{j_{\mu_2}-1} = \pm z_{j_{\mu_2}} +c,\qquad z_{j_{\lambda_2}} = \mp z_{j_{\lambda_2}-1} +c, \ldots  z_{j_{\mu_2}+1} =  z_{j_{\mu_2}} +c}
\end{equation*}
\begin{equation}\label{replacements}
\vdots
\end{equation}
and regard the output as a function of the variables $z_1,\ldots, z_p, z_{i_{\mu_1}},z_{j_{\mu_2}},\ldots$. Note that the $\pm$ choices for each line above correspond with the choices in the diagram $I$, as do the strings of equalities. We will thus identify a diagram $I$ with the string of equalities given above.

We claim that for any $p\in \{0,\ldots, k\}$
\begin{eqnarray}\label{firstclaim}
{\rm LHS} (\ref{mainform}) &=& \sum_{\substack{\lambda\vdash k-p\\ \lambda= 1^{m_1}2^{m_2}\cdots}} \frac{1}{m_1! m_2!\cdots} \frac{1}{2^{m_2+m_3+\cdots}} \sum_{I\in S(\lambda)} \int_{-\I \infty}^{\I \infty} \frac{dz_{i_{\mu_1}}}{2\pi \I} \int_{-\I \infty}^{\I \infty} \frac{dz_{j_{\mu_2}}}{2\pi \I} \cdots \\
\nonumber &&\times\, \int_{\alpha_1-\I\infty}^{\alpha_1+\I \infty} \frac{dz_1}{2\pi \I} \cdots \int_{\alpha_p-\I\infty}^{\alpha_p+\I \infty} \frac{dz_p}{2\pi \I}\, \Res{I} \left(\prod_{1\leq A<B\leq k} \frac{z_A-z_B}{z_A-z_B-c}\frac{z_A+z_B}{z_A+z_B-c}  F(\vec{z}) \right).
\end{eqnarray}
Notice that when $p=k$ the partition $\lambda$ is empty and hence this reduces back to the exact expression on the left-hand side of (\ref{mainform}). To see that this holds for general $p$ we proceed inductively (decreasing $p$). As we deform the $z_m$ contour to $\I \R$ we should consider whether there are poles at any of the points of the form $z_m = \pm z_B +c$ for $B>p$. We claim that the only points for which there are actually poles are those which correspond to augmenting the diagram $I$ into another diagram $I'$ (with $p$ replaced by $p-1$) by inserting $p$ on the left or right of any line of $I$, or into a new line. Any arrow direction and sign which leads to $I'$ satisfying the conditions listed earlier may arise. This fact follows from the Vandermonde determinant factor in the numerator, which necessarily evaluates to zero for any other augmentation of $I$ to $I'$. Another way to see this is to observe that diagrams which contain the following two snippets have zero residue:
$$
A\rightm B\rightpm C \qquad C\leftpm B\leftm A \qquad A\rightpm B\leftpm C.
$$
Any other augmentation than described above would necessary include one of the subdiagrams.

There are two symmetries which lead to the combinatorial factors above. The first (which also was present in the proof of Proposition \ref{321}) has to do with the symmetry of the parts of $\lambda$. In particular, there are $m_1!m_2!\cdots$ permutations which interchange lines in $I$ with the same length. The other symmetry is new to the type $BC$ diagrams $I$. For any line of $I$ which has at least one arrow (i.e. correspond with $\lambda_i\geq 2$) it is possible to reflect around the number $\mu_i$ and arrive at a new diagram $I'$. However, the residual subspace associated with $I$ and $I'$ is the same. Hence, for all $i$ with $\lambda_i$ we have a multiplicity of $2$. As there are $m_2+m_3+\cdots$ such $i$, we find the extra factor $2^{m_2+m_3+\cdots}$. Combining these two considerations yields the formula above.
%For instance when inserting when inserting $m$ into a line with only one number, the same residue can be encode via inserting it on the left with a right pointing arrow, or inserting it on the left with a right pointing arrow (this leads to the powers of $2$); or the shape of $I$ and $I'$ might have different

\noindent {\bf Step 2:} Now focus on the above formula with $p=0$. We claim that for any $\lambda=1^{m_1}2^{m_2}\cdots \vdash k$
\begin{eqnarray}\label{ref31prime}
&&\sum_{I\in S(\lambda)} \, \int_{-\I \infty}^{\I \infty} \frac{dz_{i_{\mu_1}}}{2\pi \I} \int_{-\I \infty}^{\I \infty} \frac{dz_{j_{\mu_2}}}{2\pi \I} \cdots \Res{I} \left(\prod_{1\leq A<B\leq k} \frac{z_A-z_B}{z_A-z_B-c}\frac{z_A+z_B}{z_A+z_B-c}  F(\vec{z})) \right) \\
\nonumber &&=\frac{1}{2^{m_1}} \sum_{\sigma\in BC_k} \int_{-\I \infty}^{\I \infty} \frac{dw_{1}}{2\pi \I} \int_{-\I \infty}^{\I \infty} \frac{dw_{2}}{2\pi \I} \cdots \Resc{\lambda} \left(\prod_{1\leq A<B\leq k} \frac{y_{\sigma(A)}-y_{\sigma(B)}}{y_{\sigma(A)}-y_{\sigma(B)}-c}\frac{y_{\sigma(A)}+y_{\sigma(B)}}{y_{\sigma(A)}+y_{\sigma(B)}-c}  F(\sigma(\vec{y})) \right),
\end{eqnarray}
where $w_i = y_{\lambda_1+\cdots +\lambda_{i-1}+1}$.
This can be seen in two pieces (in the first piece we appeal to the as of yet unproved Claim \ref{claim:Res-are-zero}). First, we consider $I\in S(\lambda)$ and replace the $z$ variables as follows. Start with the first line of $I$ (the other lines work similarly). There are two possible forms that this first line may take. The first form is
\begin{equation*}
i_1\rightp i_2\rightp \cdots \rightp i_{\mu_1-1}\rightp i_{\mu_1} \leftm i_{\mu_1+1}\leftp \cdots \leftp i_{\lambda_1},
\end{equation*}
where the arrows neighboring $i_{\mu_1}$ have (respectively) plus and then minus signs. The second form is
\begin{equation*}
i_1\rightp i_2\rightp \cdots \rightp i_{\mu_1-1}\rightm i_{\mu_1} \leftp i_{\mu_1+1}\leftp \cdots \leftp i_{\lambda_1},
\end{equation*}
where the arrows neighboring $i_{\mu_1}$ have (respectively) minus and then plus signs.

If the first line of $I$ is in the first form, then make the following change of variables
$$(z_{i_1},\ldots, z_{i_{\mu_1}-1},z_{i_{\mu_1}},z_{i_{\mu_1}+1},\ldots, z_{i_{\lambda_1}}) \mapsto (y_{\lambda_1}, \ldots, y_{\lambda_1 - \mu_1 +2},y_{\lambda_1 - \mu_1 +1},-y_{\lambda_1 - \mu_1},\ldots, -y_1);$$
whereas if the first line of $I$ is in the second form, then make the following change of variables
$$(z_{i_1},\ldots, z_{i_{\mu_1}-1},z_{i_{\mu_1}},z_{i_{\mu_1}+1},\ldots, z_{i_{\lambda_1}}) \mapsto (y_{\lambda_1}, \ldots, y_{\lambda_1 - \mu_1 +2},-y_{\lambda_1 - \mu_1 +1},-y_{\lambda_1 - \mu_1},\ldots, -y_1).$$

Similarly for the other lines of $I$, make the  analogous change of variables, where for the $i^{th}$ line, the participating $z$-variables are replaced by the variables $y_{\lambda_1+\cdots +\lambda_{i-1}+1}, \ldots, y_{\lambda_1+\cdots +\lambda_i}$.

The above procedure almost uniquely specifies how to change from $z$ to $y$ variables. The exception is for singleton lines. For example, if the first line of $I$ is a singleton $z_{i_1}$, then we may either make the change of variables $z_{i_1} \mapsto y_{1}$ or $z_{i_1}\mapsto -y_1$.

One quickly sees that this change of variables procedure turns $\Resc{I}$ into $\Resc{\lambda}$, and that for a given $I$, there are exactly $2^{m_1}$ different changes of variables which accomplishes this (here $m_1$ is the number of singletons in $I$).

Note that since the variables $z_{i_{\mu_1}},z_{j_{\mu_2}},\ldots$ were being integrated along $\I \R$, the variables $w_i = y_{\lambda_1+\cdots +\lambda_{i-1}+1}$ may end up being integrated along a shift of $\I \R$.
We need to deform these contours back to $\I \R$. However, during this deformation we meet poles of the integrand.

\begin{claim}
\label{claim:Res-are-zero}
For any $\lambda$ the total sum over $S(\lambda)$ of all residues which we need to encounter during this deformation equals $0$.
\end{claim}
We do not give a complete proof of Claim \ref{claim:Res-are-zero} in this paper. In Section \ref{sec:one-string} we give a proof of a partial result towards this claim: We prove that all residues which come from elements of $S(\lambda)$ consisting of one string give zero contribution (that is $\lambda=(k)$).
%The case of several strings is not immediately reduced to the case of one string, and we do not consider it.

\begin{remark}\label{aexception}
When $a\neq 0$ (recall $a$ is the parameter controlling the behavior of the half-space polymer at the origin), there exist poles coming from the term $\frac{z}{z+a}$ which may block this deformation. We do not attempt to keep track of this possibly more complicated residue expansion presently.
\end{remark}

For a given change of variables from $z$'s to $y$'s we can associate an element $\sigma\in BC_k$ so that $(z_1,\ldots, z_k) \mapsto (y_{\sigma(1)},\ldots, y_{\sigma(k)})$. Here an element $\sigma\in BC_k$ acts on a vector $\vec{y}$ in the manner described in Section \ref{Asymsec}. Thus, we have almost shown (\ref{ref31prime}), except that on the right-hand side of that equation, the summation is presently only over those $\sigma$ which arise in the manner described above.

The second piece to showing the claimed equality in (\ref{ref31prime}) is to recognize those $\sigma$ which do not arise from the above replacement scheme, correspond to terms in the right-hand side of (\ref{ref31prime}) which have zero residue. Hence the partial summation over $BC_k$ can be completed to the full sum, as claimed. This second piece follows in the same manner as in step 2 of the proof of Proposition \ref{321} by observing that any $\sigma$ which does not arise as above, leads to the evaluation of a residue at a point with no pole.

Summarizing what we have learned so far:
\begin{eqnarray}\label{31prime}
{\rm LHS} (\ref{mainform}) &=& \sum_{\substack{\lambda\vdash k\\ \lambda= 1^{m_1}2^{m_2}\cdots}} \frac{1}{m_1! m_2!\cdots} \frac{1}{2^{\ell(\lambda)}} \sum_{\sigma\in BC_k} \int_{-\I \infty}^{\I \infty} \frac{dw_{1}}{2\pi \I} \cdots \int_{-\I \infty}^{\I \infty} \frac{dw_{\ell(\lambda)}}{2\pi \I} \\
\nonumber && \times\, \Resc{\lambda} \left(\prod_{1\leq A<B\leq k} \frac{y_{\sigma(A)}-y_{\sigma(B)}}{y_{\sigma(A)}-y_{\sigma(B)}-c}\frac{y_{\sigma(A)}+y_{\sigma(B)}}{y_{\sigma(A)}+y_{\sigma(B)}-c}  F(\sigma(\vec{y})) \right).
\end{eqnarray}

\noindent {\bf Step 3:} It remains to compute the residues and identify the result with the right-hand side of (\ref{mainform}). It is convenient to rewrite the following expression as:
\begin{eqnarray*}
\prod_{1\leq A<B\leq k} \frac{y_{\sigma(A)}-y_{\sigma(B)}}{y_{\sigma(A)}-y_{\sigma(B)}-c}\frac{y_{\sigma(A)}+y_{\sigma(B)}}{y_{\sigma(A)}+y_{\sigma(B)}-c}\qquad\qquad\qquad\qquad\qquad\qquad\qquad\qquad &\\
= \prod_{1\leq A< B \leq k} \frac{y_{A}-y_{B}}{y_{A}-y_{B}-c}\frac{y_{A}+y_{B}}{y_{A}+y_{B}-c} \frac{y_{A}-y_{B}}{y_{A}-y_{B}+c}\frac{y_{A}+y_{B}}{y_{A}+y_{B}+c}\qquad\qquad&\\
\times\,\qquad \prod_{1\leq B<A\leq k} \frac{y_{\sigma(A)}-y_{\sigma(B)}-c}{y_{\sigma(A)}-y_{\sigma(B)}}\frac{y_{\sigma(A)}+y_{\sigma(B)}-c}{y_{\sigma(A)}+y_{\sigma(B)}}.\qquad\qquad\qquad&
\end{eqnarray*}
On the right hand side above there are two products. The first is easily seen to be invariant under the action of $BC_k$, while the second product has no residue at the poles we are considering. These considerations allow us to rewrite (\ref{31prime}) as
\begin{eqnarray}\label{31prime2}
{\rm LHS} (\ref{mainform}) &=& \sum_{\substack{\lambda\vdash k\\ \lambda= 1^{m_1}2^{m_2}\cdots}} \frac{1}{m_1! m_2!\cdots} \frac{1}{2^{\ell(\lambda)}} \int_{-\I \infty}^{\I \infty} \frac{dw_{1}}{2\pi \I} \cdots \int_{-\I \infty}^{\I \infty} \frac{dw_{\ell(\lambda)}}{2\pi \I} \\
&&\nonumber\times\, \Resc{\lambda} \left(\prod_{1\leq A< B \leq k} \frac{y_{A}-y_{B}}{y_{A}-y_{B}-c}\frac{y_{A}+y_{B}}{y_{A}+y_{B}-c} \frac{y_{A}-y_{B}}{y_{A}-y_{B}+c}\frac{y_{A}+y_{B}}{y_{A}+y_{B}+c} \right)\\
&&\nonumber\times\, \Subc{\lambda} \left(\sum_{\sigma\in BC_k} \prod_{1\leq B<A\leq k} \frac{y_{\sigma(A)}-y_{\sigma(B)}-c}{y_{\sigma(A)}-y_{\sigma(B)}}\frac{y_{\sigma(A)}+y_{\sigma(B)}-c}{y_{\sigma(A)}+y_{\sigma(B)}} F(\sigma(\vec{y}))\right).
\end{eqnarray}

We use Lemma \ref{fourthlemma} to evaluate the above residue, and we easily identify the substitution with $E^c(w_1,w_1+c,\ldots)$ as in the statement of the proposition. Combining these two expressions yields the desired residue expansion of the proposition.

\subsection{Evaluation of certain residues and substitutions}\label{subsec}
The purpose of this section is to evaluate certain residues and substitutions which arise in the proofs of Propositions \ref{Akcprop}, \ref{321}, and Conjecture \ref{Bkcprop}. We first state our lemmas in the $q$-deformed setting and then take $q\to 1$. These calculations are straightforward but require some care as they involve large products.

For a partition $\lambda\vdash k$ and $q\in (0,1)$, recall the definitions of $\Resq{\lambda}$ and $\Subq{\lambda}$ from Section \ref{resreview}, and recall that the outcome of these residue or substitution operators on functions $f(y_1,\ldots, y_k)$ are functions of the terminal variables $\big(y_1,y_{\lambda_1+1},\ldots, y_{\lambda_1+\cdots+\lambda_{\ell(\lambda)-1}-1}\big)$. We rename these remaining variables as $w_j = y_{\lambda_1+\cdots \lambda_{j-1}+1}$, for $1\leq j\leq \ell(\lambda)$.

\begin{lemma}\label{reslemma}
For all $k\geq 1$, $\lambda\vdash k$ and $q\in (0,1)$, we have that
\begin{eqnarray}\label{reslemmaeqn}
 \Resq{\lambda} \left(\prod_{1\leq i\neq j\leq k} \frac{y_i-y_j}{y_i-qy_j} \right)\qquad\qquad\qquad\qquad\qquad\qquad\qquad\qquad\qquad\qquad\qquad\qquad\qquad\qquad\qquad&\\
\nonumber =(-1)^k (1-q)^k q^{-\frac{k^2}{2}} \prod_{j=1}^{\ell(\lambda)} w_j^{\lambda_j} q^{\frac{\lambda_j^2}{2}} \frac{\prod_{1\leq i<j\leq \ell(\lambda)}(w_j-w_i)(w_i q^{\lambda_i} - w_j q^{\lambda_j})}{\prod_{i,j=1}^{\ell(\lambda)} (w_i q^{\lambda_i} - w_j)} \qquad\qquad\qquad\,\,&\\
\nonumber =(-1)^k (1-q)^k q^{-\frac{k^2}{2}} \prod_{j=1}^{\ell(\lambda)} w_j^{\lambda_j} q^{\frac{\lambda_j^2}{2}} \det\left[\frac{1}{w_i q^{\lambda_i}-w_j}\right]_{i,j=1}^{\ell(\lambda)}.\qquad\qquad\qquad\qquad\qquad\qquad&
\end{eqnarray}
\end{lemma}

\begin{proof} The second equality of the lemma follows immediately from the Cauchy determinant. Though straightforward, proving the first equality of lemma does require some care in keeping track of large products. The product on the left-hand side of (\ref{reslemmaeqn}) involves terms in which $i$ and $j$ are in the same string of variables in (\ref{resqvalues}) as well as terms in which they are in different strings. We need to compute the residue of the same string terms and multiply it by the substitution of variables into the different string terms.

Let us first evaluate same string residues. Consider variables $y_1,\ldots, y_{\ell}$, $\ell\geq 2$, and observe that
\begin{eqnarray*}
\Res{\substack{y_2=qy_1\\y_3=qy_2\\\cdots\\y_{\ell}= qy_{\ell-1}}}  \left( \prod_{1\leq i\neq j\leq \ell} \frac{y_i-y_j}{y_i-qy_j}\right)\qquad\qquad\qquad\qquad\qquad\qquad\qquad\qquad\qquad\qquad\qquad\qquad\qquad\qquad&\\
\nonumber = \frac{\prod\limits_{1\leq i\neq j\leq \ell} (q^{i-1}y_1 - q^{j-1} y_1)}{\prod\limits_{\substack{1\leq i\neq j\leq \ell\\i\neq j+1}} (q^{i-1}y_1-q^j y_1)} =
 \frac{\prod\limits_{1\leq i\neq j\leq \ell} (q^{i-1}y_1 - q^{j-1} y_1)}{\prod\limits_{\substack{1\leq i \leq \ell \\2\leq j'\leq \ell+1\\i\neq j'-1,\, i\neq j'}} (q^{i-1}y_1-q^{j'-1} y_1)} \qquad\qquad\qquad\qquad\qquad&\\
\nonumber = (-1)^{\ell-1} y_1^{\ell-1} \frac{(q-1)\cdots (q^{\ell-1}-1) (q-1)^{\ell-1} q^{\frac{(\ell-1)(\ell-2)}{2}}}{(1-q^{\ell})(q-q^{\ell})\cdots (q^{\ell-2}-q^{\ell})} = (-1)^{\ell-1} y_1^{\ell-1} \frac{(1-q)^{\ell}}{1-q^{\ell}}. \,\,\,&
\end{eqnarray*}

Now turn to the cross term between two strings of variables. Consider one set of variables  $y_1,\ldots, y_{\ell}$ with $\ell\geq 2$ and a second set of variables $y'_1\ldots, y'_{\ell'}$ with $\ell'\geq 2$. Then
\begin{eqnarray*}
\Sub{\substack{y_2=qy_1\\y_3=qy_2\\\cdots\\y_{\ell}= qy_{\ell-1}}}  \Sub{\substack{y'_2=qy'_1\\y'_3=qy'_2\\\cdots\\y'_{\ell'}= qy'_{\ell'-1}}} \left( \prod_{i=1}^{\ell}\prod_{j=1}^{\ell'} \frac{y_i-y_j}{y_i-qy_j}\right)\qquad\qquad\qquad\qquad\qquad\qquad\qquad\qquad\qquad\qquad\qquad\qquad & \\
\nonumber = \prod_{i=1}^{\ell}\prod_{j=1}^{\ell'} \frac{y_1 q^{i-1}- y'_1 q^{j-1}}{(y_1 q^{i-1} -y'_1 q^j)} = \prod_{i=1}^{\ell} \frac{\prod\limits_{j=1}^{\ell'} (y_1 q^{i-1}- y'_1 q^{j-1})}{\prod\limits_{j'=2}^{\ell'} y_1 q^{i-1} y'_1 q^{j'-1}} = \prod_{i=1}^{\ell} \frac{y_1 q^{i-1}-y'_1}{y q^{i-1} - y'_1 q^{\ell'}}. \qquad\qquad&
\end{eqnarray*}

Since the strings of variables also come interchanged, we should multiply the above expression by the same term with $(y_1,\ell)$ and $(y'_1,\ell')$ interchanged. This gives
\begin{equation*}
\prod_{i=1}^{\ell} \prod_{j=1}^{\ell'} \frac{(y_1 q^{i-1} - y'_1 q^{j-1})(y'_1 q^{j-1} -y_1 q^{i-1})}{(y_1 q^{i-1} - y'_1 q^{j})(y'_1 q^{j} - y_1 q^{i})} = q^{-\ell \ell'} \frac{ (y_1-y'_1) (y_1 q^{\ell}-y'_1 q^{\ell'})}{(y_1 q^{\ell}-y'_1)(y_1 - y'_1 q^{\ell'})}.
\end{equation*}

Returning to the statement of the lemma, we see that we can evaluate the desired residue by multiplying the same string terms over all strings in (\ref{resqvalues}) as well as multiplying all terms corresponding to pairs of different string. Using the above calculations we obtain
\begin{equation*}
\Resq{\lambda} \left(\prod_{1\leq i\neq j\leq k} \frac{y_i-y_j}{y_i-qy_j} \right) = \prod_{j=1}^{\ell(\lambda)} \frac{(1-q)^{\lambda_j}}{(1-q^{\lambda_j})} (-1)^{\lambda_j-1} w_{j}^{\lambda_j-1} \, \prod_{1\leq i<j\leq \ell(\lambda)} q^{-\lambda_i\lambda_j} \frac{(w_i-w_j)(w_i q^{\lambda_i} - w_j q^{\lambda_j})}{(w_i q^{\lambda_i} - w_j)(w_i -w_j q^{\lambda_j})}.\qquad\qquad
\end{equation*}
It is easy now to rewrite the above expression so as to produce the first equality of the lemma, as desired.
\end{proof}

\begin{lemma}\label{sublemma}
For all $k\geq 1$, $\lambda\vdash k$ and $q\in (0,1)$, we have
\begin{eqnarray}\label{sublemmaeqn}
\Subq{\lambda} \left(\prod_{1\leq i<j\leq k} \frac{1-y_i y_j}{q-y_i y_j} \frac{1- y_i y_j}{1-q y_i y_j}\right) \qquad\qquad\qquad\qquad\qquad\qquad\qquad\qquad\qquad\qquad\qquad\qquad\qquad\qquad&\\
\nonumber = q^{-\frac{k^2}{2}}  q^{\frac{k}{2}} \prod_{j=1}^{\ell(\lambda)} \frac{1-q^{\lambda_j} \tfrac{w_j^2}{q}}{1-\tfrac{w_j^2}{q}} \, \frac{(\tfrac{w_j^2}{q};q^2)_{\lambda_j}}{(w_j^2;q^2)_{\lambda_j}}
\prod_{1\leq i<j\leq \ell(\lambda)} \frac{1-q^{\lambda_i} \tfrac{w_i w_j}{q}}{1- \tfrac{w_i w_j}{q}}\,
\frac{1-q^{\lambda_j} \tfrac{w_i w_j}{q}}{1- q^{\lambda_i+\lambda_j}\tfrac{w_i w_j}{q}}.&
\end{eqnarray}
\end{lemma}

\begin{proof} The same considerations as in the proof of Lemma \ref{reslemma} apply here. Let us first evaluate same string residues. Consider variables $y_1,\ldots, y_{\ell}$, $\ell\geq 2$, and observe that

\begin{eqnarray*}
\Sub{\substack{y_2=qy_1\\y_3=qy_2\\\cdots\\y_{\ell}= qy_{\ell-1}}} \left( \prod_{1\leq i<j\leq \ell}  \frac{1-y_i y_j}{q-y_i y_j} \frac{1- y_i y_j}{1-q y_i y_j}\right) &=& \frac{ \prod\limits_{1\leq i<j\leq \ell} (1-q^{i-1+j-1} y_1^2)^2}{\prod\limits_{1\leq i<j\leq \ell} q(1-q^{i-1+j-1-1}y_1^2)(1-q^{i-1+j-1+1}y_1^2)} \\
\nonumber &=& q^{-\frac{\ell(\ell-1)}{2}}  \frac{ \prod\limits_{1=j<j\leq \ell} (1-q^{i+j-2} y_1^2)\, \prod\limits_{1<j=i+1\leq \ell} (1-q^{i+j-2}y_1^2)}{\prod\limits_{\substack{i=0\\1<j\leq \ell}} (1-q^{i+j-2} y_1^2) \, \prod\limits_{1<i=j\leq \ell} (1-q^{i+j-2} y_1^2)}\\
\nonumber &=& q^{-\frac{\ell(\ell-1)}{2}} \frac{\prod\limits_{2\leq j\leq \ell} (1-q^{j-1}y_1^2) \, \prod\limits_{2\leq j\leq \ell} (1-q^{2j-3} y_1^2)}
{\prod\limits_{2\leq j\leq \ell} (1-q^{j-2}y_1^2) \, \prod\limits_{2\leq j\leq \ell} (1-q^{2j-2} y_1^2)}\\
\nonumber &=&  q^{-\frac{\ell(\ell-1)}{2}} \frac{1-q^{\ell-1}y_1^2}{1-y_1^2}\, \frac{(1-q y_1^2)(1-q^3 y_1^2) \cdots (1-q^{2\ell-3}y_1^2)}{(1-q^2 y_1^2)(1-q^4 y_1^2) \cdots (1-q^{2\ell-2}y_1^2)}.
\end{eqnarray*}

Now turn to the cross term between two strings of variables. Consider one set of variables  $y_1,\ldots, y_{\ell}$ with $\ell\geq 2$ and a second set of variables $y'_1\ldots, y'_{\ell'}$ with $\ell'\geq 2$. We also multiply by the same term with the string interchanged. Thus,
\begin{align*}
&\Sub{\substack{y_2=qy_1\\y_3=qy_2\\\cdots\\y_{\ell}= qy_{\ell-1}}}  \Sub{\substack{y'_2=qy'_1\\y'_3=qy'_2\\\cdots\\y'_{\ell'}= qy'_{\ell'-1}}}\left(\prod_{1\leq i<j\leq k} \frac{1-y_i y_j}{q-y_i y_j} \frac{1- y_i y_j}{1-q y_i y_j}\, \frac{1-y_j y_i}{q-y_j y_i} \frac{1- y_j y_i}{1-q y_j y_i}\right) \\
&\nonumber = \prod_{i=1}^{\ell} \prod_{j=1}^{\ell'} \frac{(1-q^{i-1+j-1}y_1y'_1)^2}{q(1-q^{i-1+j-1-1}y_1y'_1)(1-q^{i-1+j-1+1}y_1y'_1)} = q^{-\ell \ell'} \prod_{j=1}^{\ell'} \frac{(1-q^{j-1} y_1 y'_1) (1-q^{\ell -1 +j-1} y_1 y'_1)}{(1-q^{j-2} y_1 y'_1)(1-q^{\ell +j-1} y_1 y'_1)}\\
&\nonumber = q^{-\ell \ell'}\frac{1-q^{\ell'-1}y_1 y'_1}{1-q^{-1} y_1 y'_1} \, \frac{1-q^{\ell-1} y_1 y'_1}{1-q^{\ell+\ell'-1} y_1 y'_1}.
\end{align*}

Returning to the statement of the lemma, we see that we can evaluate the desired residue by multiplying the same string terms over all strings in (\ref{resqvalues}) as well as multiplying all terms corresponding to pairs of different string. Using the above calculations and the replacement $w_j= y_{\lambda_1+\cdots \lambda_{j-1}+1}$ we obtain

\begin{eqnarray*}
\Subq{\lambda} \left(\prod_{1\leq i<j\leq k} \frac{1-y_i y_j}{q-y_i y_j} \frac{1- y_i y_j}{1-q y_i y_j}\right) \qquad\qquad\qquad\qquad\qquad\qquad\qquad\qquad\qquad\qquad\qquad\qquad\qquad\qquad&\\
\nonumber = \prod_{i=1}^{\ell(\lambda)} q^{-\frac{\lambda_i(\lambda_i-1)}{2}} \frac{1-q^{\lambda_i} \frac{w_i^2}{q}}{1-\frac{w_i^2}{q}} \, \frac{\big(\frac{w_i^2}{q};q^2\big)_{\lambda_i}}{\big(w_i^2;q^2\big)_{\lambda_i}} \, \prod_{1\leq i<j\leq \ell(\lambda)} q^{-\lambda_i\lambda_j} \frac{(1-q^{\lambda_i} \frac{w_i w_j}{q})(1-q^{\lambda_j }\frac{w_i w_j}{q})}{(1-\frac{w_i w_j}{q})(1-q^{\lambda_i+\lambda_j}\frac{w_i w_j}{q})}. &
\end{eqnarray*}

The powers of $q$ can be simplified further using ($\lambda_1+\lambda_2+\cdots + \lambda_{\ell(\lambda)}=k$)
\begin{equation*}
\prod_{i=1}^{\ell(\lambda)} q^{-\frac{\lambda_i(\lambda_i-1)}{2}} \prod_{1\leq i<j\leq \ell(\lambda)} q^{-\lambda_i\lambda_j}  = q^{-\frac{k^2}{2}}  q^{\frac{k}{2}}.
\end{equation*}

Using this simplification, the above formula reduces to the right-hand side of (\ref{sublemmaeqn}), thus proving the lemma.
\end{proof}

Combining these two lemmas along with (\ref{ressubeqn}) immediately yields
\begin{lemma}\label{thirdlem}
For all $k\geq 1$, $\lambda\vdash k$ and $q\in (0,1)$, we have that
\begin{eqnarray*}
\Resq{\lambda} \left(\prod_{1\leq i<j\leq k}  \frac{y_i-y_j}{y_i-qy_j} \frac{y_j-y_i}{y_j-qy_i} \frac{1-y_i y_j}{q-y_i y_j} \frac{1- y_i y_j}{1-q y_i y_j}\right)\qquad \qquad \qquad \qquad \qquad \qquad\qquad \qquad \qquad\qquad  &\\
\nonumber = (-1)^k (1-q)^{k} q^{-k^2} \prod_{j=1}^{\ell(\lambda)} w_j^{\lambda_j} q^{\frac{\lambda_j(\lambda_j+1)}{2}}  \frac{1-q^{\lambda_j} \tfrac{w_j^2}{q}}{1-\tfrac{w_j^2}{q}} \, \frac{(\tfrac{w_j^2}{q};q^2)_{\lambda_j}}{(w_j^2;q^2)_{\lambda_j}}\qquad \qquad \qquad\qquad \qquad &\\
\nonumber \times\,\prod_{1\leq i<j\leq \ell(\lambda)} \frac{(w_j-w_i)(w_i q^{\lambda_i} - w_j q^{\lambda_j}) (1-q^{\lambda_i} \tfrac{w_i w_j}{q})(1-q^{\lambda_j} \tfrac{w_i w_j}{q})}{(1- \tfrac{w_i w_j}{q})(1- q^{\lambda_i+\lambda_j}\tfrac{w_i w_j}{q})} \prod_{i,j=1}^{\ell(\lambda)} \frac{1}{w_i q^{\lambda_i} - w_j}.
\end{eqnarray*}
\end{lemma}

Recall the definitions of $\Resc{\lambda}$ and $\Subc{\lambda}$ from Section \ref{resreview}. By taking a limit as $\e\to 0$ of the above results under the change of variables $q\mapsto e^{-\e c}$, $y\mapsto  e^{-\e y}$ and $w\mapsto e^{-\e w}$ we find the following:
\begin{lemma}\label{fourthlemma}
For all $k\geq 1$, $\lambda\vdash k$ and $c\in (0,\infty)$, we have the following
\begin{eqnarray*}
\Resc{\lambda} \left(\prod_{1\leq A<B\leq k} \frac{ y_A-y_B}{y_A-y_B-c} \frac{y_A-y_B}{y_A-y_B+c} \frac{y_A+y_B}{y_A+y_B+c} \frac{y_A+y_B}{y_A+y_B-c}\right) && \\
= c^k \prod_{j=1}^{\ell(\lambda)} \frac{-1}{2c}\, \frac{\left(\tfrac{2w_j+c}{2c}\right)_{\lambda_j-1}}{\left(\tfrac{2w_j}{2c}\right)_{\lambda_j}} \Pf\left[\frac{u_i-u_j}{u_i+u_j}\right]_{i,j=1}^{2\ell(\lambda)},&&
\end{eqnarray*}
where
$$(u_1,\ldots, u_{2\ell(\lambda)}) =\big(-w_1+\tfrac{c}{2}, w_1-\tfrac{c}{2}+c\lambda_1 , -w_2+\tfrac{c}{2},w_2-\tfrac{c}{2}+c\lambda_2, \ldots, -w_{\ell(\lambda)}+\tfrac{c}{2}, w_{\ell(\lambda)}-\tfrac{c}{2} + c\lambda_{\ell(\lambda)} \big).$$
\end{lemma}

\begin{proof}
The immediate consequence of taking the $\e\to 0$ limit of Lemma \ref{thirdlem} (with the change of variables $q\mapsto e^{-\e c}$, $y\mapsto  e^{-\e y}$ and $w\mapsto e^{-\e w}$) is that
\begin{eqnarray}\label{almostthere}
&&\Resc{\lambda} \left(\prod_{1\leq A<B\leq k} \frac{ y_A-y_B}{y_A-y_B-c} \frac{y_A-y_B}{y_A-y_B+c} \frac{y_A+y_B}{y_A+y_B+c} \frac{y_A+y_B}{y_A+y_B-c}\right)\\
&&\nonumber = (-1)^{k-\ell(\lambda)}(-c)^k \prod_{j=1}^{\ell(\lambda)} \frac{ -2 w_i + c - c\lambda_i}{-2 w_i +c} \, \frac{(2w_i-c)(2w_i-c +2c) \cdots (2 w_i -c +2c(\lambda_i-1))}{(2w_i)(2w_i +2c) \cdots (2 w_i +2c(\lambda_i-1))} \\
&&\nonumber \times\,\prod_{1\leq i<j\leq \ell(\lambda)} \frac{ (w_j-w_i) (-w_i-c\lambda_i +w_j + c\lambda_j)(-w_i-w_j+c-c\lambda_i)(-w_i-w_j+c-c\lambda_j)}{(-w_i-w_j+c) (-w_i-w_j+c-c(\lambda_i+\lambda_j))(-w_i-c\lambda_i +w_j)(-w_i+w_j+c\lambda_j)}\\
&& \times\,\prod_{j=1}^{\ell(\lambda)} \frac{1}{-c\lambda_j}.
\end{eqnarray}
Note that the $(-1)^{k-\ell(\lambda)}$ factor arose from the limit of $\Resq{\lambda}$ to $\Resc{\lambda}$.
%out of the fact that when we take the $\Resq{\lambda}$, we are using the limit procedure in (\ref{amts}). The terms we must multiply by before taking the limit are of the form $(y_2-qy_1)(y_3-qy_2)\ldots$ (with the exception of terms $y_i-qy_{i+1}$ in which $i=\lambda_{j}$ for some $j$). Under the above change of variables, as $\e\to 0$ these terms become $(-1)(y_{i+1}-y_i-c)$, and since there are exactly $k-\ell(\lambda)$ of these terms, this explains the extra factor of $(-1)^{k-\ell(\lambda)}$. \note{improve explanation}

We can recognize a Pfaffian in this expression. Take  $u$ as in the statement of the lemma and use the Pfaffian identity (cf. \cite[III.8]{M})
$$
\Pf\left[\frac{u_i-u_j}{u_i+u_j}\right]_{i,j=1}^{2\ell(\lambda)} = \prod_{1\leq i<j\leq \ell(\lambda)} \frac{u_i-u_j}{u_i+u_j}.
$$
Then right-hand side of the above expression evaluates to
$$
\prod_{1\leq i<j\leq \ell(\lambda)} \frac{(-w_i+w_j) (-w_i-w_j +c - c\lambda_j) ( w_i+w_j -c + c\lambda_i) (w_i -w_j +c\lambda_i-c\lambda_j)}{(-w_i-w_j +c) (-w_i+w_j + c\lambda_j) (w_i -w_j +c\lambda_i) (w_i+w_j -c +c(\lambda_i+\lambda_j))} \, \prod_{i=1}^{\ell(\lambda)} \frac{ -2w_i +c -c\lambda_i}{c \lambda_i}
$$
which, compared with (\ref{almostthere}), yields the desired result.
\end{proof}

\subsection{Type $A$ and $BC$ symmetry}\label{Asymsec}

The symmetry group of type $A_k$ is identified with the permutation group on $k$ elements, written $S_k$. We write $\sigma\in S_k$ to identify a group element as a permutation. A permutation $\sigma\in S_k$ acts on a vector $\vec{z}=(z_1,\ldots, z_k)$ by permutating indices, yielding $\sigma(\vec{z}) = (z_{\sigma(1)},\ldots, z_{\sigma(k)})$.

A fundamental domain for the action of $S_k$ is a type $A_k$ Weyl chamber, defined here as
\begin{equation*}
W(A_k) = \left\{\vec{x} = (x_1\leq x_2\leq \cdots \leq x_k)\right\}.
\end{equation*}

An example of an $S_k$ invariant function (which is utilized in this work) is
\begin{equation*}
\prod_{1\leq i\neq j\leq k} \frac{z_A-z_B}{z_A-q z_B}.
\end{equation*}

The following symmetrization identity is proved in \cite[III.1]{M} as a special case of computing the Hall-Littlewood polynomial normalization:
\begin{equation*}
\sum_{\sigma\in S_k} \prod_{1\leq B<A\leq k}   \frac{z_{\sigma(A)}-tz_{\sigma(B)}}{z_{\sigma(A)}-z_{\sigma(B)}} = \prod_{j=1}^{k} \frac{1-t^j}{1-t}.
\end{equation*}

Changing variables as $t\mapsto e^{-\e c}$ and $z \mapsto e^{-\e z}$ and taking the limit $\e \to 0$, we find that this becomes
\begin{equation}\label{Akfacsym}
\sum_{\sigma\in S_k} \prod_{1\leq B<A\leq k}   \frac{z_{\sigma(A)}-z_{\sigma(B)}-c}{z_{\sigma(A)}-z_{\sigma(B)}} = k!.
\end{equation}

The symmetry group of type $BC_k$ is called the {\it hyperoctahedral  group} and is identified with the {\it signed} permutation group on $k$ elements (we also denote this group by $BC_k$). The group is sometimes written as $S_k\ltimes \Z_2^k$ as it is a wreath product of $S_k$ and $\Z_2^k$. An element $\sigma\in BC_k$ maps $\{\pm 1,\ldots, \pm k\}$ onto itself with the condition that $\sigma(-i) = -\sigma(i)$. As such, it suffices to specify the images of $\{1,\ldots ,k\}$ under $\sigma$. For example for $k=3$, $\sigma\in BC_3$ could map $1\mapsto \sigma(1)=-2$, $2\mapsto \sigma(2)=3$ and $3\mapsto\sigma(3)=-1$ (and hence $-1\mapsto 2$, $-2\mapsto -3$ and $-3\mapsto 1$). A signed permutation $\sigma\in BC_k$ can act on a vector $\vec{z}=(z_1,\ldots,z_k)$ either {\it multiplicatively} or {\it additively}. If $\sigma$ acts multiplicatively, then
$$\sigma(\vec{z}) = (z_{|\sigma(1)|}^{\sgn(\sigma(1))},\ldots,z_{|\sigma(k)|}^{\sgn(\sigma(k))})$$
and if $\sigma$ acts additively, then
$$\sigma(\vec{z}) = (\sgn(\sigma(1)) z_{|\sigma(1)|},\ldots,\sgn(\sigma(k)) z_{|\sigma(k)|}).$$

A fundamental domain for the additive action of $BC_k$ is a type $BC_k$ Weyl chamber, defined here as
\begin{equation*}
W(BC_k) = \left\{ \vec{x} = (x_1\leq x_2\leq \cdots \leq x_k \leq 0)\right\}.
\end{equation*}

The following identity is proved in \cite{Vidya} as a special case of computing the $q=0$ Koornwinder polynomial normalization (type $BC$ analogs of Hall-Littlewood polynomials). In particular by combining equations (1), (5) and Theorem 2.6 of \cite{Vidya} and choosing $\lambda=\emptyset$ we have that for arbitrary $a$, $b$, $t$,
\begin{equation*}
\sum_{\sigma\in BC_k} \prod_{1\leq A<B\leq k} \frac{1-t\frac{z_{\sigma(B)}}{z_{\sigma(A)}}}{1-\frac{z_{\sigma(B)}}{z_{\sigma(A)}}}\frac{1-t\frac{1}{z_{\sigma(A)}z_{\sigma(B)}}}{1-\frac{1}{z_{\sigma(A)}z_{\sigma(B)}}} \prod_{j=1}^{k} \frac{\left(1-a \frac{1}{z_{\sigma(j)}}\right)\left(1-b \frac{1}{z_{\sigma(j)}}\right)}{1-\frac{1}{z_{\sigma(j)}^2}} = \prod_{j=1}^{k} \frac{1-t^j}{1-t} (1-ab t^{j-1}).
\end{equation*}
In the above, the action of $BC_k$ is taken to be multiplicative. Note also that the product over $1\leq A<B\leq k$ above can also be taken over $1\leq B<A\leq k$, which simply amounts to renaming the variables $z_j\mapsto z_{k+1-j}$, $1\leq j\leq k$.

Changing variables as $a\mapsto e^{-\e a}$, $t\mapsto e^{-\e c}$, $z_j\mapsto e^{-\e z_j}$, and setting $b\equiv 0$, we find that the $\e\to 0$ limit of the above identity yields (now with additive action of $BC_k$)
\begin{equation}\label{Bkfacsym}
\sum_{\sigma\in BC_k} \prod_{1\leq B<A\leq k} \frac{z_{\sigma(A)}-z_{\sigma(B)}-c}{z_{\sigma(A)}-z_{\sigma(B)}}\frac{z_{\sigma(A)}+z_{\sigma(B)}-c}{z_{\sigma(A)}+z_{\sigma(B)}} \prod_{j=1}^{k} \frac{z_{\sigma(j)}-a}{z_{\sigma(j)}} = 2^k k!.
\end{equation}

\section{The case of one string in Claim \ref{claim:Res-are-zero}}\label{sec:one-string}

%\textbf{The case $a=0$}

%For this function the value of $F(\sigma (y))$ in the right-hand side of \eqref{ref31prime} is the same for any choice of $\sigma \in BC_k$. In the case of one string, all substitutions also lead to the same expression for any choice of this string. Therefore, we can ``forget'' about this function and consider poles which come from the factor

The goal of this section is to present a proof of a partial result towards Claim \ref{claim:Res-are-zero}. It is given by equation \eqref{12} below. We deal with the case when the element of $I$ consists of one string only.
%$$
%R (z_1, \dots, z_k) := \prod_{i<j} \frac{(z_i-z_j) (z_i+z_j)}{(z_i-z_j-c) (z_i+z_j-c)}
%$$
%in the right-hand side of \eqref{ref31prime} only. The cancelation of the residues for this function implies the necessary cancelation of residues with our function $F$ present.

Let us denote by $BC^{I(\lambda)}_k$ the set of elements of $BC_k$ which come from the diagrams with the partition type $\lambda$. Equation \eqref{ref31prime} can be written as
\begin{eqnarray}\label{ref31prime-2}
&&\sum_{I\in S(\lambda)} \, \int_{-\I \infty}^{\I \infty} \frac{dz_{i_{\mu_1}}}{2\pi \I} \int_{-\I \infty}^{\I \infty} \frac{dz_{j_{\mu_2}}}{2\pi \I} \cdots \Res{I} \left(\prod_{1\leq A<B\leq k} \frac{z_A-z_B}{z_A-z_B-c}\frac{z_A+z_B}{z_A+z_B-c}  F(\vec{z})) \right) =\frac{1}{2^{m_1}} \\
\nonumber && \times \sum_{\sigma\in BC_k^{I(\lambda)}} \int_{ e_1^\sigma c -\I \infty}^{e_1^\sigma c + \I \infty} \frac{dw_{1}}{2\pi \I} \int_{e_2^\sigma c -\I \infty}^{e_2^\sigma c + \I \infty} \frac{dw_{2}}{2\pi \I} \cdots \Resc{\lambda} \left(\prod_{1\leq A<B\leq k} \frac{y_{\sigma(A)}-y_{\sigma(B)}}{y_{\sigma(A)}-y_{\sigma(B)}-c}\frac{y_{\sigma(A)}+y_{\sigma(B)}}{y_{\sigma(A)}+y_{\sigma(B)}-c}  F(\sigma(\vec{y})) \right),
\end{eqnarray}
where $w_i = y_{\lambda_1+\cdots +\lambda_{i-1}+1}$, and $e_1(\sigma), e_2(\sigma), \dots$ determine the contours of integration after the change of variables; they are certain parameters of $\sigma$.

Claim \ref{claim:Res-are-zero} asserts that
\begin{eqnarray*}
&& \sum_{\sigma\in BC_k^{I(\lambda)}} \int_{ e_1^\sigma c -\I \infty}^{e_1^\sigma c + \I \infty} \frac{dw_{1}}{2\pi \I} \int_{e_2^\sigma c -\I \infty}^{e_2^\sigma c + \I \infty} \frac{dw_{2}}{2\pi \I} \cdots \Resc{\lambda} \left(\prod_{1\leq A<B\leq k} \frac{y_{\sigma(A)}-y_{\sigma(B)}}{y_{\sigma(A)}-y_{\sigma(B)}-c}\frac{y_{\sigma(A)}+y_{\sigma(B)}}{y_{\sigma(A)}+y_{\sigma(B)}-c}  F(\sigma(\vec{y})) \right) \\
\nonumber && = \sum_{\sigma\in BC_k^{I(\lambda)}} \int_{-\I \infty}^{\I \infty} \frac{dw_{1}}{2\pi \I} \int_{-\I \infty}^{ \I \infty} \frac{dw_{2}}{2\pi \I} \cdots \Resc{\lambda} \left(\prod_{1\leq A<B\leq k} \frac{y_{\sigma(A)}-y_{\sigma(B)}}{y_{\sigma(A)}-y_{\sigma(B)}-c}\frac{y_{\sigma(A)}+y_{\sigma(B)}}{y_{\sigma(A)}+y_{\sigma(B)}-c}  F(\sigma(\vec{y})) \right),
\end{eqnarray*}
that is, that the contours can be deformed back to the imaginary axis, and the total contribution of arising residues is 0.

Let $BC^{I(k)}_k$ be the set of elements of $BC_k$ which correspond to one-string diagrams. Our aim is to prove the following equality
\begin{eqnarray}\label{12}
&& \sum_{\sigma\in BC_k^{I(k)}} \int_{ e_1^\sigma c -\I \infty}^{e_1^\sigma c + \I \infty} \frac{dy_{1}}{2\pi \I} \Resc{\lambda} \left(\prod_{1\leq A<B\leq k} \frac{y_{\sigma(A)}-y_{\sigma(B)}}{y_{\sigma(A)}-y_{\sigma(B)}-c}\frac{y_{\sigma(A)}+y_{\sigma(B)}}{y_{\sigma(A)}+y_{\sigma(B)}-c}  F(\sigma(\vec{y})) \right) \\
\nonumber && = \sum_{\sigma\in BC_k^{I(k)}} \int_{-\I \infty}^{\I \infty} \frac{dy_{1}}{2\pi \I} \Resc{\lambda} \left(\prod_{1\leq A<B\leq k} \frac{y_{\sigma(A)}-y_{\sigma(B)}}{y_{\sigma(A)}-y_{\sigma(B)}-c}\frac{y_{\sigma(A)}+y_{\sigma(B)}}{y_{\sigma(A)}+y_{\sigma(B)}-c}  F(\sigma(\vec{y})) \right).
\end{eqnarray}

\subsection{Initial considerations}

We start with an analysis which residues we need to sum up.

Recall that we have a sum over diagrams. We will consider the case when only one string of arrows is present in a diagram. Thus, we will assume that the variables in our string have indices from $1$ through $k$. We will denote the set of all diagrams with one string of arrows and with a structure explained in Step 1 of Section \ref{bcproof} by $\s_1$.
After Step 1, for each such diagram we have an integral over the contour $\Real (z_k)= \ep>0$, for $\ep \ll 1$ (not precisely 0 due to a possible pole at 0). In Step 2, we need to make a change of variables of the form $y_1 := z_k- \mathbf{a} c$ or $y_1 := - z_k- \mathbf{b} c$, where $\mathbf a, \mathbf b \in \mathbb N$ are certain parameters of our diagram. After this change of variables the contour of integration will be close to $\Real (y_1) = - \mathbf{a} c$ (or $- \mathbf{b} c$), and we need to deform it back to the contour $\Re(y_1)=\ep>0$. During this deformation we will pick up some residues which we are interested in.

The residues depend on the function
$$
R(z_1, \dots, z_k) = \prod_{i<j} \frac{(z_i-z_j) (z_i+z_j)}{(z_i-z_j-c) (z_i+z_j-c)},
$$
the analytic function $F(z_1, \dots, z_k)$, and a diagram $s \in \s_1$ under consideration.

Each diagram $s\in \s_1$ is of the form
$$
\qquad i_1 \aRp i_2 \aRp \dots i_{\mu-1} \aRp k \aLm i_{\mu+1} \aLp \dots \aLp i_k,
$$
and encodes the substitution of variables $z_{i_1} = z_{i_2} +c$, $\dots$, $z_{i_{\mu}-1} = z_k+ c$, $z_{i_{\mu+1}} = c - z_k$, $\dots$, $z_{i_{k}}= z_{i_{k-1}}+c$, or, if the arrow of minus-type comes to $k$ from the another direction then the substitution has a similar form, see equation \eqref{replacements}.
We refer to such a substitution as the substitution prescribed by $s$.

Let us introduce notations
$$
R_s (z_k):= \Res{s}\, R(z_1,z_2, \dots, z_k), \qquad F_s (z_k) := \Sub{s}\, F(z_1,\dots, z_k).
$$
We obtain the function $R_s (z_k)$ by taking residues of the function $R$ as specified by $s$. This amounts to removing the factors in the denominator of $R (z_1, \dots, z_k)$ which correspond to the poles at which we took residues, and and making the substitution prescribed by $s$ for what remains. The function $F_s (z_k)$ is merely a substitution, because it does not have poles. Our main concern will be the function $R_s (z_k)$, while $F_s (z_k)$ will not produce any difficulties for our analysis.

\begin{example}
Assume that $k=6$ and $s$ is given by the string
$$
1 \aRp 3 \aRp 4 \aRp 5 \aRm 6 \aLp 2.
$$
During the deformation of contours we use factors $(z_5+z_6 -c)$, $(z_4-z_5-c)$, $(z_3-z_4-c)$, $(z_2-z_6-c)$, and $(z_1-z_3-c)$; we omit these factors from the denominator of $R(z_1, \dots, z_k)$. After this we make a substitution of variables $z_5= - z_6+c$, $z_4 = -z_6 +2c$, $z_3 = -z_6 +3c$, $z_2 = z_6 +c$, $z_1 = -z_6 +4c$ (which is prescribed by $s$) into all other factors and obtain the function
$$
R_s (z_6) = \frac{15c^5 (-2 z_6+5c)(-2 z_6+7c) (2 z_6+c)}{z_6 (-z_6+c)^2}.
$$
For this diagram above we make a change of variables $y_1 := - z_6 - c$. This comes from the change of variables from $z$'s to $y$'s in step 2 along with the fact that the $z_6$ contour has real part $\ep$ and the relation between $y_1$ and $z_6$.

If instead we were working with a diagram
$$
2 \aRp 6 \aRm 5 \aLp 4 \aLp 3 \aLp 1,
$$
we need to set $y_1 := z_6 - 4c$, while the substitution and the function $R_s (z_6)$ are exactly the same.
\end{example}

As we already noticed, for a diagram $s$ we need to make a change of variables of the form $y_1 := z_k- \mathbf{a} c$ or $y_1 := - z_k- \mathbf{b} c$, where $ \mathbf{a} = \mathbf{a} (s) $ and $ \mathbf{b} = \mathbf{b} (s)$ are certain parameters of the diagram $s$. In more detail, if a diagram $s$ has a form (recall that $k$ is the largest index)
$$
q_{\mathbf{a}} \aRp q_{ \mathbf{a}-1} \aRp \dots q_{1} \aRm k \aLp p_1 \aLp \dots \aLp p_{\mathbf{b}},
$$
then we need to make a change of variables $y_1 := -z_k - \mathbf{b} c$. Let $\tilde R_s (y_1)$ be the resulting function. Let $\mathcal W (s)$ be the sum of residues which we obtain in the process of moving the contour of integration of $\tilde R_s (y_1)$ in the left-hand side of \eqref{12} from $\Real (y_1) = - \mathbf{b} c - \ep$ to $\Real (y_1) = \ep$.

If the diagram $s$ has instead the form
$$
p_{\mathbf{b}} \aRp p_{ \mathbf{b}-1} \aRp \dots p_{1} \aRp k \aLm q_1 \aLp \dots \aLp q_{\mathbf{a}},
$$
then we need to make a change of variables $y_1 := z_k - \mathbf{a} c$. Let $ \mathcal W(s)$ be the contribution of poles which we obtain in the process of moving the integral of $\tilde R_s (y_1)$ in the left-hand side of \eqref{12} from the contour $\Real (y_1) = - \mathbf{a} c + \ep$ to the contour $\Real (y_1) = \ep$.

Let
$$
\mathcal W := \sum_{s \in \s_1} \mathcal W(s)
$$
be the total contribution coming from all our diagrams from $\s$. Our goal is to prove the following proposition.
\begin{proposition}
\label{theorem-main-oneLine}
We have $\mathcal W=0$.
\end{proposition}
By definition, this proposition is equivalent to equation \eqref{12}.

We will consider a pairing on the set $\s_1$. For a diagram $s$ we can read the whole line in the opposite direction; let us denote such a diagram by $s^{\prime}$ (an example of such a diagram is given in the example above). It is clear that $R_s(z_k) = R_{s^{\prime}} (z_k)$, because the operations with multivariate integrals are the same for these diagrams. However, we need to make different changes of variables for them: In one case, we should move the $z_k$ variable to the right, and in the other case --- to the left.

Let us denote the set of all diagrams in which $k-1$ is lying to the left of $k$ by the symbol $\s$. This set includes a half of all diagrams $s$ --- another half can be obtained by taking $s^{\prime}$.

For an analytic function $G(z)$ we shall denote by $\pp_G (x_1,x_2)$ the set of poles of this function lying in the real interval $(x_1,x_2)$.

\begin{proposition}
\label{initial}
We have
$$
\mathcal W = \sum_{s \in \s} \left( - \sum_{z^{*} \in \pp_{R_s} (\ep;\mathbf{a}(s) c+\ep)} F_s (z^*) \Res{z_k = z^*} R_s (z_k) + \sum_{z^* \in \pp_{R_s} (-\mathbf{b}(s) c-\ep;\ep)} F_s (z^*) \Res{z_k = z^*} R_s (z_k) \right)
$$
\end{proposition}

\begin{proof}
The two sums correspond with the two types of change of variables relating $y_1$ to $z_k$. Each $s\in \s$ corresponds to one type of change of variable, and $s^{\prime}$ corresponds to the other type.

In the case when $y_1 := z_k- \mathbf{a}(s) c$ we pick up all poles of $\tilde R_s (y_1) = R_s (z_k- \mathbf{a}(s) c)$ between $- \mathbf{a}(s) c + \ep$ and $\ep$. This is the same as picking up all poles of $R_s (z_k)$ between $\ep$ and $ \mathbf{a}(s) c +\ep$. In the case $y_1 := -z_k - \mathbf{b}(s) c$ we obtain a sign from the Jacobian; but we also get a sign from the direction of integration over contours. These signs cancel out. Thus, we need to pick poles in the movement of the contour in $\int \tilde R_1 (y_1) = \int R_s (-z_k- \mathbf{b}(s) c)$ from $\Real (y_1) = - \mathbf{b}(s) c - \ep $ to $\Real (y_1) = \ep$. Equivalently, we need to pick poles in the movement of the contour in $\int R_s (z_k)$ from $\Real (z_k) = \ep$ to $\Real (z_k) = - \mathbf{b}(s) c -\ep$. The statement of the proposition readily follows. The difference of signs is because we move through these poles in different directions.
\end{proof}

\subsection{Pole at 0}

We know that for each diagram there exists no more than one arrow of minus-type; this arrow must end in $k$. It can connect $k-1$ and $k$ (the first line below), or a number $x$ and $k$; in the latter case $k-1$ and $k$ are connected by an arrow of plus-type (the second line below).
\begin{align}
\label{second_pairing}
\dots k-1 \aRm k \aLp x \dots \\
\dots k-1 \aRp k \aLm x \dots.
\end{align}

Let us introduce a new pairing of $\s$; diagrams in each pair are the same except for the change of plus-type and minus-type arrows leading to $k$ (as shown in \ref{second_pairing}). Note that when $k$ is at the end of the string, the pairing is between the diagrams
\begin{align*}
1 \aRp \dots (k-1) \aRp k \\
1 \aRp \dots (k-1) \aRm k
\end{align*}
Thus, we partition $\s$ into disjoint pairs. For $s \in \s$ we shall denote $\bar s \in \s$ the paired element through this plus/minus arrow switching.
%From each pair let us choose the diagram such that $k-1$ and $k$ are connected by the arrow of minus-type. We shall denote the set of all such diagrams by the symbol $\s$.

\begin{example}
Diagrams
$$
s = 1 \aRp 3 \aRp 4 \aRp 5 \aRm 6 \aLp 2
$$
and
$$
\bar s = 1 \aRp 3 \aRp 4 \aRp 5 \aRp 6 \aLm 2
$$
form a pair of the described pairing of $\s$. We have
$$
R_s (z_6) = \frac{15c^5 (-2 z_6+5c)(-2 z_6+7c) (2 z_6+c)}{z_6 (-z_6+c)^2}, \qquad R_{\bar s} (z_6) = \frac{15c^5 (2 z_6+5c)(2 z_6+7c) (-2 z_6+c)}{-z_6 (z_6+c)^2}.
$$

\end{example}

\begin{lemma}
\label{codeTranspose}
For any $s \in \s$ we have
$$
R_s (z_k) = R_{\bar s} (- z_k), \qquad F_s (0) = F_{\bar s} (0).
$$
\end{lemma}

\begin{proof}
%It follows from the rules how we should substitute the variables.
For two diagrams
\begin{gather*}
i_1 \aRp \dots \aRp k-1 \aRm k \aLp x \dots \aLp i_2 \\
i_1 \aRp \dots \aRp k-1 \aRp k \aLm x \dots \aLp i_2,
\end{gather*}
the substitutions have the form
\begin{gather*}
z_{i_1} = \mathbf{a} c -z_k;\,\, \dots \,\,z_{k-1} = c-z_k;\,\, z_k = z_k;\,\, z_x = z_k+ c;\,\, \dots\,\,z_{i_2} = z_k+ \mathbf{b} c; \\
z_{i_1} = z_k + \mathbf{a} c;\,\,  \dots \,\,z_{k-1} = z_k+ c;\,\, z_k = z_k;\,\, z_x = c - z_k;\,\, \dots \,\,z_{i_2} = \mathbf{b} c -z_k.
\end{gather*}
Note that our initial function $R(z_1, \dots, z_k)$ are stable under the transform $z_k \to (-z_k)$ (due to the fact that the index $k$ is the largest one). The statements of the lemma are thus clearly visible.
\end{proof}

\begin{proposition}
\label{prop2}
We have
$$
\mathcal W = \sum_{s \in \s} \left( - \sum_{z^* \in \pp_{F_s} (\ep;\mathbf{a}(s)c+\ep)} F_s (z^*) \Res{z_k = z^*} F_s (z_k) + \sum_{z^* \in \pp_{F_s} (- \mathbf{b}(s)c-\ep;-\ep)} F_s (z^*) \Res{z_k=z^*} F_s (z_k) \right)
$$
\end{proposition}

\begin{remark}
Note that the only change from Proposition \ref{initial} is that the residues at 0 do not enter this summation since the interval $[-\mathbf{b}(s)c,\ep]$ is now replaced by $[-\mathbf{b}(s)c,-\ep]$.
\end{remark}
\begin{proof}
Let us consider two diagrams from the same pairing of $\s$. Lemma \ref{codeTranspose} shows that the poles of $R_s(z_k)$ and $R_{\bar s} (z_k)$ are closely related: If $\{e_1, \dots, e_M \}$ are the poles of $R_s(z_k)$, then $\{-e_1, \dots, -e_M\}$ are the poles of $R_{\bar s} (z_k)$. It is clear that $\Res{z_k = e_i} R_s (z_k) = - \Res{z_k = -e_i} R_{\bar s} (z_k)$. In particular, $\Res{z_k = e_i} R_s (0) = - \Res{z_k = -e_i} R_{\bar s} (0)$. Therefore, Lemma \ref{codeTranspose} implies the statement of the proposition.
%We also note that parameters satisfy $\mathbf{a} (s) = \mathbf{b} (\bar s)$ and $ \mathbf{b} (s) = \mathbf{a} ( \bar s)$.
%It is easy to see that the poles of the functions $R_s(z_k)$ can be at the points of the form $z_k = n \frac{c}{2}$, where $n$ is an integer.
%Therefore, the contribution of the pole at 0 is canceled when we sum contributions coming from $R_s (z_k)$ and $R_{\bar s} (z_k)$, while contributions from all other poles are doubled.
\end{proof}

% ***- it's a pity that we do not have a complete cancelation on this stage.

We need to prove that $\mathcal W$ is equal to 0. We shall prove a stronger theorem --- in fact, cancellations already happen when we fix a complex number as a pole.

\begin{theorem}
\label{main}
For any complex number $z^* \ne 0$ we have
$$
\sum_{s \in \s} F_s (z^*) \Res{z_k = z^*} R_s (z_k) = 0.
$$
\end{theorem}
This theorem is nontrivial for points of the form $nc/2$, $n \in \mathbb N \backslash \{0\}$, only. We shall prove Theorem \ref{main} in next sections. Due to Proposition \ref{prop2}, this will imply Theorem \ref{theorem-main-oneLine}.

\subsection{Structure of poles}
\label{structure}

For a fixed diagram $s \in \s$ we need to understand the structure of the function $R_s (z_k)$. We shall need some notation.

A diagram $s \in \s$ gives rise to two disjoint sets $A(s), B(s) \in \{1,2, \dots, k\}$, $A(s) \sqcup B(s) = \{1, \dots, k\}$; the elements of $A(s)$ are the indices of variables such that the substitution prescribed by $s$ has a form $\mathbf{q} c - z_k$ for these variables, and the elements of $B(s)$ are $k$ and the indices of variables such that the prescribed substitution has a form $\mathbf{q} c + z_k$, for some $\mathbf{q} \in \mathbb N$.

For two numbers $x_1,x_2$  in $\{1,\ldots,k\}$ we shall write $\sss(x_1,x_2)=1$ if they both belong to $A(s)$ or both belong to $B(s)$, and we shall write $\sss(x_1,x_2)=-1$ otherwise.
For any number $x$ we shall denote by $\p(x)$ the number such that the plus-type arrow goes from $\p(x)$ to $x$. In the following, if for some $x$ the number $\p(x)$ does not exist, we assume that all statements about $\p(x)$ are true.

\begin{example}\label{examples}
For a diagram $s = 1 \aRp 3 \aRp 4 \aRp 5 \aRm 6 \aLp 2$ we have $A(s) = \{1,3,4,5\}$ and $B(s)=\{2,6\}$. We have $\sss(1,5)=1$, $\sss(2,6)=1$, $\sss(2,3)=-1$, $\p(3)=1$, $\p(6)=2$.
\end{example}

Let us fix a number $n \in \mathbb Z \backslash \{0\}$. We shall consider some combinatorial quantities which are determined by fixed parameters $n, A(s), B(s)$. We seek to figure out what order of pole the function $R_s(z_k)$ has at the point $z_k = \frac{n c}{2}$. We omit the dependence on $s$ in notation in the rest of this section because $s$ will be fixed.

\begin{definition} \label{def:func-v} Let $\tilde a$ be the index such that there is an arrow of minus-type from $z_{\tilde a}$ to $z_k$. Let us define a function $v: \{1,\dots, k\} \to \mathbb Z/2$ as follows: Let $v(k) := n/2$, and $v(\tilde a) := 1 - n/2$. All other values are defined inductively: If we have $v(x)=l$, then $v( \p(x)) := l+1$. We call a pair of numbers $x_1, x_2$ such that  $v(x_1) + v(x_2) =0$ a \textit{plus-zero}.
We call a pair of numbers $x_1, x_2$ such that $\sss(x_1,x_2) =1$ and $v(x_1) + v(x_2) =1$ a \textit{plus-pole}.
\end{definition}

\begin{example}
For the diagram $s$ as in Example \ref{examples} and $n/2=1$ we have $v(6) = 1$, $v(5)= 0$, $v(4)=1$, $v(3)=2$, $v(2)=2$, $v(1)=3$. We have no plus-zeros, $(4,5)$ is a plus-pole (but $(5,6)$ is not a plus-pole).
\end{example}

\begin{lemma}
\label{plus}
If $n/2$ is a noninteger then the number of plus-poles is equal to the number of plus-zeros. If $n/2$ is an integer then the number of plus-poles is equal to the number of plus-zeros plus one.
\end{lemma}
\begin{proof}
By definition, we have $v(\tilde a) + v(k) =1$. Therefore, plus-poles and plus-zeros can appear only in one of the sets $A$ and $B$ (if $v(\tilde a) < v(k)$ then in $A$, if $v(\tilde a) > v(k)$, then in $B$). For $l \in \mathbb N \cup \{0 \}$ note that in the set $\{-l, -(l-1), \dots, 0, \dots, l, l+1\}$ the number of pairs with sum 1 is greater than the number of pairs with sum 0 by 1; also note that in the set $\{-l+1/2, \dots, 1/2, \dots, l+1/2 \}$ the number of pairs with sum 1 is equal to the number of pairs with sum 0. The statement of the lemma follows from these observations.
\end{proof}

\begin{definition} We call a pair of numbers $x_1 < x_2$ such that $v(x_1) - v(x_2) =0$ a \textit{minus-zero}.
We call a pair of numbers $x_1 < x_2$ such that $\sss(x_1,x_2) = -1$ and $v(x_1) - v(x_2) = 1$ a \textit{minus-pole}.
Let $N_0^-$ and $N_p^-$ denote the number of minus-zeros and minus-poles, respectively.

Consider a pair of elements $x_1, x_2$ such that $\sss(x_1,x_2)=-1$, $v(x_1)= v(x_2)$ and $\p(x_1) <x_2$, $\p(x_2) <x_1$ (recall that if some of elements $\p$ do not exist, corresponding inequalities are assumed to be true); we call such a pair a \textit{pivot pair}. We shall denote by the symbol $L$ the number of pivot pairs.
\end{definition}

\begin{example}
For $s$ and $n$ as above $(2,4)$, $(3,6)$, $(2,1)$ are minus-poles (but $(5,6)$ is not a minus-pole); $(4,6)$ and $(2,3)$ are minus-zeros; also both of these minus-zeros are pivot pairs. Therefore, $L= 2$.
\end{example}

\begin{lemma}
\label{minus}
Assume that $N_p^- \ge 1$. For any such diagram and $n$ we have
$$
N_0^- +L -1 \ge N_p^-.
$$
\end{lemma}

\begin{proof}
We shall consider the set $\q_r := \{ x \in \{1, \dots, k\} : v(x) \le r \}$. Let us prove the statement of the lemma by induction in $r\geq \min\big(v(k), v(\tilde a)\big)$; on each step we consider only numbers which belong to $\q_r$.

While increasing $r$ we have the following situations (note that $v^{-1}(r)$ is the preimage under the map $v$ of $r$):
\begin{enumerate}
\item $v^{-1} (r)$ and $v^{-1} (r-1)$ (i.e. their union) consists of one or zero elements. At this point in the induction no minus-poles and minus-zeros exist in $\q_r$, and inequality clearly holds.
\item $v^{-1} (r)$ consists of two elements while $v^{-1} (r-1)$ consists of one element. The second element in $v^{-1} (r)$ must be $k$ or $\tilde a$. We know that $v(k)+v(\tilde a)=1$; therefore, the second element can give a minus-pole only if it is $\tilde a$ and $v(\tilde a) = 1$, $v(k)=0$. But we exclude the case $v(k) = n/2 =0$ from the consideration. Therefore, on this step minus-poles cannot appear. However, one minus-zero appears for sure because we have two elements in the set $v^{-1} (r)$. Thus, our inequality holds.
\item $v^{-1} (r)$ and $v^{-1} (r-1)$ each consist of two elements. On this step one minus-zero is added. Also it is easy to see that one new minus-pole appears always, and two minus-poles appear if and only if the pair of numbers forming $v^{-1} (r-1)$ is pivot. Thus, our inequality still holds.
\item $v^{-1} (r)$ consists of one element and $v^{-1} (r-1)$ consists of two elements. No minus-zero appears on this step, and one minus-pole appears if and only if the pair of numbers forming $v^{-1} (r-1)$ is pivot.
Also note that if $v^{-1} (r)$ consists of zero elements, and $v^{-1} (r-1)$ consists of two elements, then no new minus-zeros and minus-poles added on this step, while $v^{-1} (r-1)$ is a pivot pair. Therefore, in this case we obtain a stronger inequality $N_0^- +L -2 \ge N_p^-$.
\item $v^{-1} (r)$ and $v^{-1} (r-1)$ consist of one element. On these steps new minus-poles and minus-zeros do not appear.
\end{enumerate}
In light of having checked all of these case, we verify the statement of the lemma holds.
\end{proof}

\begin{proposition}
Assume that a diagram $s \in \s$ has $L= L(s)$ pivot pairs.
\begin{itemize}
\item If $n/2$ is a noninteger, then $R_s(z_k)$ has a pole of order not greater than $L-1$ at the point $z_k = z^* = n \frac{c}{2}$.
\item If $n/2$ is an integer, then $R_s(z_k)$ has a pole of order not greater than $L$ at the point $z_k = z^* = n \frac{c}{2}$.
\end{itemize}
\end{proposition}
\begin{proof}
During the transformation of
$$
R(z_1, \dots, z_k) = \prod_{i<j} \frac{(z_i-z_j) (z_i+z_j)}{(z_i-z_j- c) (z_i+ z_j -c )}
$$
into one-dimensional integral over $z_k$ encoded by a diagram $s \in \s$ some poles disappear. In more detail, in this transformation we use the poles $z_i - z_j-c$, where $i<j$ and $\sss(i,j) = 1$; also we use the pole in $z_{k-1} + z_k -c$. Therefore, the poles and zeros of the function $R_s (z_k)$ at $z^*$ should come from expressions of the other form. It is easy to see that the factor $(z_i-z_j-c)$ gives rise to a pole of $R_s (z_k)$ at $z^*=nc/2$ if and only if $(i,j)$ is a minus-pole in our terminology; in the similar vein, the same is true for factors of the form $(z_i+z_j-c)$ and plus-poles, factors $(z_i-z_j)$ and minus-zeros, factors $(z_i+z_j)$ and plus-zeros. Combining Lemmas \ref{plus} and \ref{minus} we obtain the statement of the proposition.
\end{proof}

\begin{remark}
\label{emptyCaseRemark}
If $s$ and $n$ are such that there are no new indices involved in steps 4) and 5) in the proof of Lemma \ref{minus}, then the order of the pole also does not exceed $L-1$, see step 4) of the proof of Lemma \ref{minus}.
\end{remark}

\subsection{Proof of Theorem \ref{main}}
Let us fix $z^* = nc/2$, $n \in \mathbb Z \backslash \{0\}$. Our goal is to prove that the function
$$
\sum_{s \in \s} R_s (z) F_s (z)
$$
has no pole at $z^*$. We shall give a partition of $\s$ into several disjoint sets such that the sum of functions over each set has no pole at $z^*$. Note that these sets depend on a fixed $z^*$. Moreover, the value $F_s( z^*)$ will be the same for diagrams $s$ from the same set, so our difficulties will come from the sum of the functions $R_s(z)$.

Let us choose $s \in \s$ and let us describe which diagrams are in the same set with $s$. Recall that $L(s)$ is the number of pivot pairs in $s$, and that the order of the pole at $z^*$ of the function $R_s(z_k)$ does not exceed $L(s)$.

\begin{definition} \label{def:pivot-pairs}
For the set $A(s)$ let us consider all elements of pivot pairs in this set: $w_1 < w_2 < \dots < w_{L(s)}$. For $1 \le i \le L(s)-1$ we consider the set $\mathcal A_i = \{ \p(w_i), \p (\p (w_i)), \dots, w_{i+1} \}$; we call $\mathcal A_i$ a \textit{block}. We also give the same definition for blocks $\mathcal B_i$ in $B(s)$. Note that for each $i$, $\mathcal A_i$ and $\mathcal B_i$ start and end with the elements of the same pivot pairs and have the same values of the function $v_s$ on its elements (see the definition of the function $v = v_s$ in Definition \ref{def:func-v}).
\end{definition}

We obtain the function $R_s (z_k)$ after certain substitution of variables prescribed by the diagram $s$. After this all variables $z_i$ are expressed through the variable $z_k$; let us denote by $z_i^s$ the value of variable $z_i$ after the substitution prescribed by $s$.

Let us also make a change of variables $\ep_k := z_k - nc/2$. It is easy to see that other variables can be written in the following way: if $i \in A(s)$, then $z_i^s = v(i) c -\ep_k$; if $i \in B(s)$, then $z_i^s = v(i) c + \ep_k$.

Our key transformation is a \textit{``swap''} of blocks $\mathcal A_i$ and $\mathcal B_i$. For a diagram $s \in \s$ we can consider the sets $\hat A_i := (A \backslash \mathcal A_i) \cup \mathcal B_i$ and $\hat B_i := (B \backslash \mathcal B_i) \cup \mathcal A_i$. Note that for any $x \in \{\mathcal A_i, \mathcal B_i\}$ and $y \in \{1, \dots, k\} \backslash \{\mathcal A_i, \mathcal B_i\}$ inequalities $x<y$ and $v(x) > v(y)$ are true or false simultaneously; this follows from the definition of pivot pairs and the fact that our blocks start and end with elements of the same pivot pairs. We obtain that the sets $\hat A_i$ and $\hat B_i$ give rise to the new diagram $s^i$.

\begin{example}
For a diagram
$$
1 \aRp 2 \aRp 3 \aRp 6 \aRp 7 \aRm 8 \aLp 5 \aLp 4
$$
and the pole at $z_k=c$ we have $n=2$, $z_8= c + \ep_k$, $z_7 = 0 - \ep_k$, $z_6 = c-\ep_k$, $z_5= 2c + \ep_k$, $z_4= 3c + \ep_k$, $z_3= 2c - \ep_k$, $z_2= 3c - \ep_k$, $z_1= 4c - \ep_k$. It is convenient to depict this diagram and $n$ as a diagram
\begin{gather*}
7 \ \ \ 6 \ \ \ 3 \ \ \ 2 \ \ \ 1  \\
8 \ \ \ 5 \ \ \ 4
\end{gather*}
In such a diagram the numbers with the same value of $v$-function are on the same vertical.

The pairs $(6,8)$ and $(2,4)$ are pivot pairs, and $(3,5)$ is not a pivot pair because $4$ is to the right of $3$ on the diagram above. $\mathcal A_1$ is $\{3,2\}$, and $\mathcal B_1$ is $\{5,4\}$.

The diagram $s^1$ is
$$
1 \aRp 4 \aRp 5 \aRp 6 \aRp 7 \aRm 8 \aLp 3 \aLp 2.
$$

\end{example}

It turns out that we can control how the functions $R_s (z_k)$ and $R_{s^i} (z_k)$ are related. Let us describe this relation.

Recall that the function $R(z_1, \dots, z_k)$ is obtained as a product of factors $\frac{(z_i-z_j) (z_i+z_j)}{(z_i-z_j-c) (z_i+z_j-c)}$ over all pairs of integers; when we consider a one-dimensional integral some factors disappear, and we should substitute $z_r^s = v_s (r) + \ep_k$ into all other factors.

Note that the values of variables from the blocks $\mathcal A_i$ and $\mathcal B_i$ change as $v_s (z_j)+ \ep_k \to v_s (z_j) - \ep_k$ and vice versa in $s$ and $s^i$. All other variables remain the same.

Let $I_i$ be the subset of all pairs of numbers from $1$ to $k$ such that both elements of the pair belong to $\mathcal A_i \cup \mathcal B_i$. Let
$$
\FF_{I_i} (z_1, \dots, z_k) := \prod_{(a,b) \in I_i} \frac{z_a - z_b}{z_a - z_b -c},
$$
and let $R_{s,i} (\ep_k)$ be the function obtained from $\FF_{I_i}$ after all substitutions prescribed by $s$. Note that we can write
$$
R_s \left(\frac{nc}{2}+\ep_k\right) = R_{s,0} \left(\frac{nc}{2} + \ep_k\right) R_{s,i} \left(\frac{nc}{2}+\ep_k\right),
$$
where the function $R_{s,0}$ comes from the product of all other factors.

\begin{lemma}
\label{transpose}
We have
$$
R_{s^i} \left( \frac{nc}{2} + \ep_k \right) = R_{s,0} \left( \frac{nc}{2} + \ep_k \right) R_{s,i} \left( \frac{nc}{2} - \ep_k \right).
$$
Therefore, in order to obtain the function $R_{s^i} (\ep_k)$ from $R_s (\ep_k)$ we need to change the sign of the variable $\ep_k$ in some part of the expression.
\end{lemma}

\begin{proof}
Note that the whole set of variables takes exactly the same values in the cases of $s$ and $s^i$. Therefore, the product of factors of the form $(z_i+ z_j)$ and $(z_i + z_j -c)$ over all pairs gives exactly the same in both cases.

If both indices are inside $\mathcal A_i \cup \mathcal B_i$, then the value of $(z_i- z_j)$ can be obtained by the change $\ep_i \to (-\ep_i)$, because in both variables this change happens.

Let us consider the case when one variable $\bj$ is outside of blocks, while the other variable is inside blocks. Assume that $\bj$ is less than all numbers inside the blocks (the opposite case, when $\bj$ is greater than all numbers inside the blocks, can be considered in the same way). We can group all these factors into pairs
$$
\frac{(z_{\bj} - z_a) (z_{\bj} - z_{\bar a})}{(z_{\bj} - z_a - c) (z_{\bj} - z_{\bar a} - c)},
$$
where $(a,\bar a)$ are such that $a \in \mathcal A_i$, $\bar a \in \mathcal B_i$, and $v_s(a) = v_s (\bar a)$.
Note that in both cases ($s$ and $s^i$) the expression above is equal to
$$
\frac{(z_{\bj}^s - v_s(a) - \ep_k) (z_{\bj}^s - v_s(a)+ \ep_k)}{(z_{\bj}^s - v_s(a) - c - \ep_k) (z_{\bj}^s - v_s(a) + \ep_k - c)},
$$
therefore, the product over the pairs with fixed index $\bj$ outside the blocks and another index from the blocks gives the same expression for the diagrams $s$ and $s^i$.

%A sign comes from the account of the sign of poles and zeros.
\end{proof}

Now we are able to prove Theorem \ref{main} in the case of non-integer $n/2$. We know that in this case the order of the pole does not exceed $L(s)-1$. We also know that the diagram $s$ contains $L(s)-1$ pairs of blocks $(\mathcal A_1, \mathcal B_1)$, $\dots$, $(\mathcal A_{L-1}, \mathcal B_{L-1})$. Let us ``swap'' these blocks in all possible ways. We obtain $2^{L(s)-1}$ different diagrams; denote this set by $T_s$.

It is enough to prove that the sum of $R_s (z_k)$ over $T_s$ does not have a pole at $z^*$. Indeed, the value of $F_s (z_k)$ is the same for any diagram $s$ from $T_s$, because all variables take the same values for diagrams from $T_s$.

Note that the sets of pairs of indices in which the transformation $\ep_k \to (-\ep_k)$ happens do not intersect for different parameters $i$ of blocks $(\mathcal A_i$, $\mathcal B_i)$. Therefore, one can write
$$
R_s (nc/2+\ep_k) = \hat R_{s,0} (nc/2 + \ep_k) R_{s,1} (nc/2 + \ep_k) \dots R_{s, L(s)-1} (nc/2 + \ep_k),
$$
where $\hat R_{s,0}$ does not change under our ``swaps'', and $R_{s,i}$ corresponds to the ``swap'' of $(\mathcal A_i, \mathcal B_i)$. With the use of Lemma \ref{transpose} we see that the sum of all functions corresponding to these $2^{L(s)} -1$ diagrams can be written as
$$
\hat R_0 \left( \frac{nc}{2} + \ep_k \right) \sum_{s \in T_s} R_1 \left( \frac{nc}{2} \pm \ep_k \right) \dots R_{L-1} \left( \frac{nc}{2} \pm \ep_k \right).
$$
Now we need to consider a pole of this sum at $\ep_k=0$. Consider the pairs of indices such that one of them belongs to $\mathcal A_i$ and the other one belongs to $\mathcal B_i$. It is easy to see that there are exactly $|\mathcal A_i|$ minus-zeroes among these pairs, and exactly $|\mathcal A_i|-1$ minus-poles. Therefore, there are exactly $2 |\mathcal A_i|-1$ changes of signs in the factors which give $\ep_k$ or $\ep_k^{-1}$. Thus, the sum in the previous formula can be written as

\begin{equation*}
\frac{\hat R_0 \left( \frac{nc}{2} + \ep_k \right)}{\ep_k^{deg}} \sum_{s \in T_s} (-1)^{\mbox{number of minuses}} \tilde R_1 \left( \frac{nc}{2} \pm \ep_k \right) \dots \tilde R_{L-1} \left( \frac{nc}{2} \pm \ep_k \right),
\end{equation*}
where $\hat R_0$, $\tilde R_i$ has no singularities at $\ep_k =0$ and the degree $deg$ does not exceed $L(s)-1$. But one can write such a sum in the form
$$
\frac{\hat R_0 \left( \frac{nc}{2} + \ep_k \right)}{\ep_k^{deg}} \left( \tilde R_1 \left( \frac{nc}{2} + \ep_k \right) - \tilde R_1 \left( \frac{nc}{2} - \ep_k \right) \right) \dots \left( \tilde R_{L-1} \left( \frac{nc}{2} + \ep_k \right) - \tilde R_{L-1} \left( \frac{nc}{2} - \ep_k \right) \right).
$$
It is clear that such an expression does not have a pole at $\ep_k =0$.

One can readily see that all diagrams from $\s$ can be split into such disjoint groups $T_s$. This completes the proof of Theorem \ref{main} in the case of non-integer $n/2$.

Now we need to consider the case when $n/2$ is an integer. In this case the order of the pole at $\ep_k=0$ can be equal to $L(s)$. The proof uses the same mechanism, but we need to consider one more allowed transform of the diagram $s$ (in addition to ``swaps'' of blocks $\mathcal A_i$ and $\mathcal B_i$ for $1 \le i \le L(s)-1$). Let us define new blocks $\mathcal A_L = \{ x \in A(s) : x< w_L \}$ and $\mathcal B_L = \{ x \in B(s) : x < w'_L \}$ (see the definition of $w_{L}, w'_{L}$ in Definition \ref{def:pivot-pairs}).

First, let us consider the case when both new blocks are empty. Then the results of Section \ref{structure} (see Remark \ref{emptyCaseRemark}) show that the order of the pole at $\ep_k=0$ does not exceed $L(s)-1$; therefore, it is enough to consider ``swaps'' of previously defined blocks in order to obtain a cancelation of residues.

In the general case, the difference of this pair of blocks from the previous ones is that in this pair the blocks can have different number of elements. However, we still can ``swap'' these blocks as we did above and obtain a new diagram which we denote $s^L$ (again, we obtain a correctly defined diagram due to the definition of pivot pairs).

Let us describe how the functions $R_s (z_k)$ and $R_{s^L} (z_k)$ are related.

\begin{definition}
We call the elements of $\mathcal A_L \cup \mathcal B_L$ which belong to some minus-zero \textit{regular}, and we call other numbers from $\mathcal A_L \cup \mathcal B_L$ \textit{irregular}. Let us denote by $\mathcal P$ the elements from $\{1,\dots,k\} \backslash (\mathcal A_L \cup \mathcal B_L)$ which belong to some minus-zero; the other elements from $\{1,\dots,k\} \backslash (\mathcal A_L \cup \mathcal B_L)$ we denote by $\mathcal{NP}$.
\end{definition}

\begin{example}
Consider a diagram with $n/2=-1$ represented by the diagram
\begin{gather*}
15 \ \ \ 13 \ \ \ 12 \ \ \ 10 \ \ \ \ 9 \ \ \ \ 8 \ \ \ 6 \ \ \ 5 \ \ \ 3 \ \ \ 1 \\
\ \ \ \ \ \ \ \ \ 14 \ \ \ 11 \ \ \ 7 \ \ \ 4 \ \ \ 2
\end{gather*}
For this diagram we have 1 plus-zero $(12,15)$, 2 plus-poles $(10,15)$ and $(12,13)$, 5 minus-zeros, 6 minus-poles, two pivot pairs $(7,8)$ and $(9,11)$ (other minus-zeros are not pivot pairs), $L(s)=2$. We have $\mathcal A_1=\{7 \}$ and $\mathcal B_1 = \{8\}$, $\mathcal A_2 = \{2,4 \}$ and $\mathcal B_2 = \{1,3,5,6\}$. The indices $\{1,3\}$ are irregular, the indices $\{2,5,4,6\}$ are regular. The set $\mathcal{NP}$ is $\{12,13,15\}$ and the set $\mathcal P$ is $\{7,8,9,11,10,14\}$.
\end{example}

Let $J_1$ be the set of pairs of indices such that one of the indices is irregular and another belongs to $\mathcal{NP}$. Let $J_2$ be the set of pairs of indices such that one of the indices is irregular and another belongs to $\mathcal{P}$.
Let $J_3$ be the set of pairs of indices such that both indices belong to $\mathcal A_L \cup \mathcal B_L$.

Let
$$
R_{J_1,L} (z_1, \dots, z_k) := \prod_{(i,j) \in J_1} \frac{z_i+z_j}{z_i+z_j-c} \frac{z_i-z_j}{z_i-z_j-c},
$$

$$
R_{J_2,L} (z_1, \dots, z_k) := \prod_{(i,j) \in J_2} \frac{z_i+z_j}{z_i+z_j-c} \frac{z_i-z_j}{z_i-z_j-c},
$$

$$
R_{J_3,L} (z_1, \dots, z_k) := \prod_{(i,j) \in J_3} \frac{z_i+z_j}{z_i+z_j-c} \frac{z_i-z_j}{z_i-z_j-c},
$$
let $R^1_s (nc/2+ \ep_k)$, $R^2_s (nc/2+ \ep_k)$, and $R^3_s (nc/2+ \ep_k)$ be the functions which are obtained from the functions above after the substitution of variables prescribed by $s$, and let
$$
R_{s,L} \left( \frac{nc}{2} + \ep_k \right) := R^1_s \left( \frac{nc}{2} + \ep_k \right) R^2_s \left( \frac{nc}{2} + \ep_k \right) R^3_s \left( \frac{nc}{2} + \ep_k \right).
$$
Note that we can write
$$
R_s \left( \frac{nc}{2} + \ep_k \right) = \hat R_{s,0} \left( \frac{nc}{2} + \ep_k \right) R_{s,L} \left( \frac{nc}{2} + \ep_k \right),
$$
where $\hat R_0$ comes as a product over other pairs of indices.

\begin{lemma}
We have for $j=1,2,3$
$$
R^j_{s^L} \left( \frac{nc}{2}+ \ep_k \right) = R^j_s \left( \frac{nc}{2} - \ep_k \right).
$$
We also have the equality
$$
R_{s^L} \left( \frac{nc}{2}+ \ep_k \right) = \hat R_0 \left( \frac{nc}{2} + \ep_k \right) R_L \left( \frac{nc}{2} - \ep_k \right),
$$
similar to the case of other blocks.
\end{lemma}

\begin{proof}

Let $i$ be an irregular index; it is less than indices from $\mathcal P$ and $\mathcal{NP}$ due to the definition of a pivot pair. Then $z^{s^L}_i = v_s(i) c \mp \ep_k$.

Let $N := \max (n/2, - n/2)$; after the substitution prescribed by $s$ the variables from $\mathcal{NP}$ take values $\{-Nc+\ep_k, \dots, 0+\ep_k, \dots, Nc+\ep_k\}$ or the same expressions with $(-\ep_k)$ instead of $(+\ep_k)$. Assume that the $+$ sign appears here; the opposite case can be considered in the same way. For an index $j \in \mathcal{NP}$ with $z^s_j = v_s (j) c + \ep_k$ let $\bar j$ denote the index such that $z^s_{\bar j} = - v_s(j) c +\ep_k$.
%We obtain that all numbers from $\mathcal{NP}$ are split into pairs $(j, \bar j)$

For a fixed irregular index $i$ and for any $j \in \mathcal{NP}$ one can directly verify that the sets $\{ z_i^s - z_j^s, z_i^s+z_j^s\}$ and $\{z_i^{s^L} - z_{\bar j}^{s^L}, z_i^{s^L} + z_{\bar j}^{s^L} \}$ are obtained from each other by $\ep_k \to (-\ep_k)$. Also the same is true for the sets $\{ z_i^{s^L} - z_j^{s^L}, z_i^{s^L}+z_j^{s^L} \}$ and $\{z_i^{s} - z_{\bar j}^{s}, z_i^{s} + z_{\bar j}^{s} \}$. The statement of the lemma about the function $R^1_{s^L}$ follows from this.

For any $a \in \mathcal P$ denote by $b$ the index such that $(a,b)$ is a minus-zero. We have $z_a^s = v_s(a) + \ep_k$ and $z_b^s = v_s(a) - \ep_k$ (or vice versa). One can directly verify that the sets $\{ z_i^s - z_a^s, z_i^s- z_b^s\}$ and $\{z_i^{s^L} - z_a^{s^L}, z_i^{s^L} - z_b^{s^L} \}$ are obtained from each other by $\ep_k \to (-\ep_k)$. Also the same is true for the sets $\{ z_i^s + z_a^s, z_i^s + z_b^s\}$ and $\{z_i^{s^L} + z_a^{s^L}, z_i^{s^L} + z_b^{s^L} \}$. The statement of the lemma about the function $R^2_{s^L}$ follows from this.

If both indices belong to $\mathcal A_{L(s)} \cup \mathcal B_{L(s)}$, then in both indices the change $\ep_k \to (-\ep_k)$ happens. The statement of the lemma about the function $R^3_{s^L}$ follows from this.

If both indices do not belong to $\mathcal A_{L(s)} \cup \mathcal B_{L(s)}$, then obviously corresponding factors do not change.

The only remaining case is that one index is regular, and the other is from $\{1, \dots, k\} \backslash \left( \mathcal A_L \cup \mathcal B_L \right)$; the product over corresponding factors does not change --- this can be shown in the same way as in the proof of Lemma \ref{transpose}.

\end{proof}

Note that the sets of pairs of indices which give rise to the change $\ep_k \to (-\ep_k)$ in corresponding factors do not intersect for the pairs of blocks $(\mathcal A_1, \mathcal B_1)$, $\dots$, $(\mathcal A_L, \mathcal B_L)$. Let us swap these blocks in all possible ways, obtaining $2^{L}$ different diagrams. Recall that the order of the pole at $\ep_k=0$ does not exceed $L$. In exactly the same way as before we can show that the sum over these diagrams does not have a pole at $\ep_k=0$ which implies Theorem \ref{main} in the general case.

\begin{example}
Consider the diagram $s$
$$
1 \aRp 2 \aRp 3 \aRp 6 \aRp 7 \aRm 8 \aLp 5 \aLp 4.
$$
One can compute that
$$
R_s (z_8) = \frac{315(-2 z_8+5c)(2 z_8+c) c^7 (-2 z_8+9c)(-2 z_8+7c)(-4z_8^2+9c^2)}{(-2z_8+c)^2 (-z8+c)^2 z_8^2}.
$$
and the residue at point $z_8= c$ is equal to $-4110750 c^8$. In this case our blocks are $\mathcal A_1 = \{2,3\}$, $\mathcal B_1 = \{4,5\}$, $\mathcal A_2 = \{1\}$, and the block $\mathcal B_2$ is empty. Transposing these blocks in all possible ways we obtain three new diagrams:
\begin{align*}
1 \aRp 4 \aRp 5 \aRp 6 \aRp 7 \aRm 8 \aLp 3 \aLp 2, \\
2 \aRp 3 \aRp 6 \aRp 7 \aRm 8 \aLp 5 \aLp 4 \aLp 1, \\
4 \aRp 5 \aRp 6 \aRp 7 \aRm 8 \aLp 3 \aLp 2 \aLp 1.
\end{align*}
Computations of residues of functions corresponding to these diagrams at the point $z_8=c$ give numbers $-850500 c^8$, $3827250 c^8$, and $1134000 c^8$; we have
$$
-4110750 -850500 + 3827250 + 1134000 =0.
$$
\end{example}


\begin{thebibliography}{alpha}

\bibitem{Alhfors}
L.~Alhfors.
\newblock Complex analysis.
\newblock McGraw-Hill, 1979.

\bibitem{AKQ}
T. Albert, K. Khanin, J. Quastel.
\newblock The intermediate disorder regime for directed polymers in dimension $1+1$.
\newblock {\it Ann. Probab.} {\bf 42}:1212--1256, 2014.

\bibitem{ACQ}
G.~Amir, I.~Corwin, J.~Quastel.
\newblock Probability distribution of the free energy of the continuum directed random polymer in $1+1$ dimensions.
\newblock {\em Commun. Pure Appl. Math.},{\bf 64}:466--537, 2011.

\bibitem{AAR}
G.~Andrews, R.~Askey, R.~Roy.
\newblock {\it Special functions.}
\newblock Cambridge University Press, 2000.

\bibitem{BBSSS}
E.~Bachmat, D.~Berend, L.~Sapir, S.~Skiena, N.~Stolyarov.
\newblock Analysis of airplane boarding via space-time geometry and random matrix theory. \newblock {\it J. Phys. A}, {\bf 39}:453--459, 2006.

\bibitem{BaikRains}
J.~Baik, E.~Rains.
\newblock Symmetrized random permutations.
\newblock Random Matrix Models and Their Applications, {\it MSRI volume 40} (ed. P. Bleher and A. Its.):1--19, 2001.

\bibitem{BQS}
M.~Bal\'{a}zs, J.~Quastel, T.~Sepp\"{a}l\"{a}inen.
\newblock Scaling exponent for the Hopf-Cole solution of KPZ/stochastic Burgers.
\newblock {\it J. Amer. Math. Soc.}, {\bf 24}:683--708, 2011.

\bibitem{BarCorHahn}
G. Barraquand, I. Corwin.
\newblock The $q$-{H}ahn asymmetric exclusion process.
\newblock {\it Ann. Appl. Probab.}, to appear.

\bibitem{BarCorBeta}
G. Barraquand, I. Corwin.
\newblock Random-walk in Beta-distributed random environment.
\newblock arXiv:1503.04117, 2015.

\bibitem{BG}
L.~Bertini, G. Giacomin.
\newblock  Stochastic Burgers and KPZ equations from particle systems.
\newblock {\em Commun. Math. Phys}, {\bf 183}:571--607, 1997.

\bibitem{BC}
L.~Bertini, N.~Cancrini.
\newblock  The stochastic heat equation: Feynman-Kac formula and intermittence.
\newblock {\em J. Stat. Phys.}, {\bf 78}:1377--1401, 1995.

\bibitem{Bethe}
H.~Bethe.
\newblock Zur Theorie der Metalle. I. Eigenwerte und Eigenfunktionen der linearen Atomkette. (On the theory of metals. I. Eigenvalues and eigenfunctions of the linear atom chain)
\newblock {\it Zeitschrift fur Physik}, {\bf 71}:205--226, 1931.

\bibitem{BFG}
G.~Blatter,M.V.~Feigel'man, V.B.~Geshkenbein, A.I.~Larkin, V.M.~Vinokur.
\newblock Vortices in high-temperature superconductors.
\newblock {\it Rev. Mod. Phys.}, {\bf 66}:1125--1388, 1994.

\bibitem{BorodinR}
A.~Borodin.
\newblock On a family of rational symmetric functions.
\newblock arXiv:1410.0976.


\bibitem{BorCor}
A.~Borodin, I.~Corwin.
\newblock Macdonald processes.
\newblock {\it Probab. Theor. Rel. Fields},  {\bf 158}:225--400, 2014.


\bibitem{BorCordiscrete}
A.~Borodin, I.~Corwin.
\newblock Discrete time $q$-TASEP.
\newblock {\it Int. Math. Res. Not.}, rnt206, 2013.

\bibitem{BCLyapunovpaper}
A.~Borodin, I.~Corwin.
\newblock On moments of the parabolic Anderson model.
\newblock {\it Ann. Appl. Probab.}, {\bf 24}:1171--1197, 2014.

\bibitem{BCF}
A.~Borodin, I.~Corwin, P.~L.~Ferrari.
\newblock Free energy fluctuations for directed polymers in random media in $1+1$ dimension.
\newblock {\it Commun. Pure Appl. Math.}, {\bf 67}:1129--1214, 2014.

\bibitem{BCG}
A.~Borodin, I.~Corwin, V. Gorin.
\newblock Stochastic six-vertex model.
\newblock {\it Duke Math. J.}, to appear.

\bibitem{BCPS1}
A.~Borodin, I.~Corwin, L. Petrov, T. Sasamoto.
\newblock Spectral theory for the $q$-{B}oson particle system.
\newblock {\it Compositio Math.} {\bf 151}:1--67, 2015.

\bibitem{BCPS2}
A.~Borodin, I.~Corwin, L. Petrov, T. Sasamoto.
\newblock Spectral theory for interacting particle systems solvable by coordinate Bethe ansatz.
\newblock {\it Commun. Math. Phys.} {\bf 339}:1167--1245, 2015.

\bibitem{BorCorRem}
A.~Borodin, I.~Corwin, D.~Remenik.
\newblock Log-Gamma polymer free energy fluctuations via a Fredholm determinant identity.
\newblock {\it Commun. Math. Phys.},  {\bf 324}:215--232, 2013.

\bibitem{BCS}
A.~Borodin, I.~Corwin, T.~Sasamoto.
\newblock From duality to determinants for q-TASEP and ASEP.
\newblock {\it Ann. Probab.}, {\bf 42}:2314--2382, 2014.

\bibitem{BorodinPetrovhigher}
\newblock A.~Borodin, L.~Petrov.
\newblock Higher spin six vertex model and rational symmetric functions.
\newblock In preparation, Lectures in Ecole de Physique des Houches 2015.


\bibitem{Bra}
M.~Bramson.
\newblock Maximal displacement of branching brownian motion.
\newblock {\it Commun. Pure Appl. Math.}, {\bf 31}:531--581, 1978.

%\bibitem{CalCaux}
%P.~Calabrese, J.~S.~Caux.
%\newblock Dynamics of the attractive 1D Bose gas: analytical treatment from integrability.
%\newblock {\it J. Stat. Mech.}, P08032, 2007.

\bibitem{CDprl}
P.~Calabrese, P.~Le Doussal.
\newblock An exact solution for the KPZ equation with flat initial conditions.
\newblock {\it Phys.Rev.Lett.} {\bf 106}:250603, 2011.

\bibitem{CDlong}
P.~Calabrese, P.~Le Doussal.
\newblock The KPZ equation with flat initial condition and the directed polymer with one free end.
\newblock {\it J. Stat. Mech.} P06001, 2012.

\bibitem{CDR}
P.~Calabrese, P.~Le Doussal, A.~Rosso.
\newblock Free-energy distribution of the directed polymer at high temperature.
\newblock {\it Euro. Phys. Lett.}, {\bf 90}:20002, 2010.

%\bibitem{Cod}
%E.~A.~Coddington, N.~Levinson.
%\newblock {\em Theory of Ordinary Differential Equations}.
%\newblock McGraw Hill, 1955.

\bibitem{ICreview}
I. Corwin.
\newblock The {K}ardar-{P}arisi-{Z}hang equation and universality class.
\newblock {\em Random Matrices Theory Appl.}, {\bf 1}, 2012.

\bibitem{Corhahn}
I. Corwin.
\newblock The $q$-{H}ahn Boson process and $q$-{H}ahn {TASEP}.
\newblock {\it Int. Math. Res. Not.} rnu094, 2015.

\bibitem{ICICM}
I.~Corwin.
\newblock Macdonald processes, quantum integrable systems and the Kardar-Parisi-Zhang universality class.
\newblock {\em Proceedings of the International Congress of Mathematicians 2014}.

\bibitem{COSZ}
I.~Corwin, N.~O'Connell, T.~Sepp\"{a}l\"{a}inen, N.~Zygouras.
\newblock Tropical combinatorics and Whittaker functions.
\newblock {\it Duke J. Math.} {\bf 163}:513--563, 2014.

\bibitem{CorPetpush}
I. Corwin, L. Petrov.
\newblock {The q-PushASEP: a new integrable model for traffic in $1+1$ dimension}
\newblock {\it J. Stat. Phys.} {\bf 160}:1005--1026, 2015.

\bibitem{CorPetdual}
I. Corwin, L. Petrov.
\newblock Stochastic higher spin vertex models on the line
\newblock {\it Commun. Math. Phys.} to appear.


\bibitem{CQ}
I.~Corwin, J.~Quastel.
\newblock Crossover distributions at the edge of the rarefaction fan.
\newblock {\em Ann. Probab.}, {\bf 41}:1243--1314, 2013.

\bibitem{CSS}
I. Corwin, T. Sepp\"{a}l\"{a}inen, H. Shen.
\newblock The strict-weak lattice polymer.
\newblock {\it J. Stat. Phys.} {\bf 160}:10027--1053.

\bibitem{CorTsai}
I. Corwin, L. Tsai.
\newblock KPZ equation limit of higher-spin exclusion processes.
\newblock arXiv:1505.04158.

\bibitem{Deift}
P.~Deift.
\newblock Applications of a commutation formula.
\newblock {\it Duke Math. J.} {\bf 45}:267--310, 1978.

\bibitem{DeiftG}
P.~Deift, D.~ Gioev.
\newblock {\it Random matrix theory: Invariant ensembles and universality.}
\newblock AMS, 2009.

%\bibitem{Dorlas}
%T.~C.~Dorlas.
%\newblock Orthogonality and completeness of the Bethe ansatz eigenstates of the nonlinear Schroedinger model.
%\newblock  {\it Commun. Math. Physics.}, {\bf 154}:347--376, 1993.

\bibitem{Dot}
V.~Dotsenko.
\newblock Bethe ansatz derivation of the Tracy-Widom distribution for one-dimensional directed polymers.
\newblock {\it Euro. Phys. Lett.}, {\bf 90}:20003, 2010.


\bibitem{Dot2}
V.~Dotsenko.
\newblock Replica Bethe ansatz derivation of the GOE Tracy-Widom distribution in one-dimensional directed polymers with free boundary conditions.
\newblock {\it J. Stat. Mech.}, P11014, 2012.

\bibitem{Dot3}
V.~Dotsenko.
\newblock Distribution function of the endpoint fluctuations of one-dimensional directed polymers in a random potential.
\newblock {\it J. Stat. Mech.}, P02012, 2013.

\bibitem{Dot4}
V.~Dotsenko.
\newblock Two-time free energy distribution function in $(1+1)$ directed polymers.
\newblock {\it J. Stat. Mech.}, P06017, 2013.

\bibitem{Dot5}
V.~Dotsenko.
\newblock Two-point free energy distribution function in $(1+1)$ directed polymers.
\newblock {\it J. Phys. A}, {\bf 46} 355001, 2013.


\bibitem{FH}
D.S.~Fisher, D.A.~Huse.
\newblock Directed paths in random potential.
\newblock {\it Phys. Rev. B.}, {\bf 43}:10728--10742, 1991.

\bibitem{Forrester}
P.~Forrester
\newblock {\it Log-gases and Random matrices}.
\newblock Princeton university press, 2010.

\bibitem{FNS}
D.~Forster, D.R.~Nelson, M.J.~Stephen.
\newblock Large-distance and long-time properties of a randomly stirred fluid.
\newblock {\it Phys. Rev. A}, {\bf 16}:732--749, 1977.

%\bibitem{G}
%J.~G\"artner.
%\newblock Convergence towards Burgers equation and propagation of chaos for weakly asymmetric exclusion process.
%\newblock {\em Stoch. Proc. Appl.}, {\bf 27}:233--260, 1988.

%\bibitem{Hahn}
%W.~Hahn.
%\newblock Beitr\"{a}ge zur Theorie der Heineschen Reihen. Die 24 Integrale der hypergeometrischen q-Differenzengleichung. Das q-Analogon der Laplace-Transformation
%\newblock {\it Mathematische Nachrichten}, {\bf 2}:340--379, 1949.

\bibitem{Gaudin}
M.~Gaudin.
\newblock Boundary energy of a Bose gas in one dimension.
\newblock {\it Phys. Rev. A}, {\bf 4}:386--394, 1971.

\bibitem{GutkinSuther}
E.~Gutkin, B.~Sutherland.
\newblock Completely integrable systems and groups generated by reflections.
\newblock {\it Proc. Natl. Acad. Sci. USA}, {\bf 76}:6057--6059, 1979.

\bibitem{HHRN}
O.~Hallatschek, P.~Hersen, S.~Ramanathan, D.R.~Nelson
\newblock Genetic drift at expanding frontiers promotes gene segregation.
\newblock {\it Proc. Natl. Acad. Sci.}, {\bf 104}:19926--19930, 2007.


\bibitem{HHZ}
T.~Halpin Healy, Y.Z.~Zhang.
\newblock Kinetic Roughening, Stochastic Growth, Directed Polymers and all that.
\newblock {\it Phys. Rep.}, {\bf 254}:215--415, 1995.

\bibitem{HHR}
A.~Hansen, E.~L.~Hinrichsen, S.~Roux.
\newblock Roughness of crack interfaces.
\newblock {\it Phys. Rev. Lett.}, {\bf 66}:2476--2479, 1991.


\bibitem{HO}
G.~J.~Heckman, E.~M.~Opdam.
\newblock Yang's system of particles and Hecke algebras.
\newblock {\it Ann. Math.}, {\bf 145}:139--173, 1997.

\bibitem{Hel}
S.~Helgason.
\newblock {\it Groups and geometric analysis: Integral geometry, invariant differential operators, and spherical functions.}
\newblock Academic Press, 1984.

\bibitem{HuHe}
D.A.~Huse, C.~L.~Henley.
\newblock Pinning and roughening of domain walls in Ising systems due to random impurities.
\newblock {\it Phys. Rev. Lett.}, {\bf 54}:2708--2711, 1985.

\bibitem{HL}
T.~Hwa, M.~Lassig.
\newblock Similarity-detection and localization.
\newblock {\it Phys. Rev. Lett.}, {\bf 76}:2591--2594, 1996.


%\bibitem{HobsonTribe}
%T.~Hobson, R.~Tribe.
%\newblock On the duality between coalescing Brownian particles and the heat equation driven by Fisher-Wright noise.
%\newblock {\em Elect. Comm. Probab.}, {\bf 10}:136--145, 2005.

\bibitem{IS}
T.~Imamura, T.~Sasamoto.
\newblock Current moments of 1D ASEP by duality.
\newblock {\em J. Stat. Phys.},  {\bf 142}:919--930, 2011.

\bibitem{ImSa}
T.~Imamura, T.~Sasamoto.
\newblock  Replica approach to the KPZ equation with half Brownian motion initial condition.
\newblock {\it J. Phys. A: Math. Theor.} {\bf 44}:385001, 2011.

\bibitem{ImSaKPZ}
T.~Imamura, T.~Sasamoto.
\newblock  Exact solution for the stationary Kardar-Parisi-Zhang equation.
\newblock {\it Phys. Rev. Lett.}, {\bf 108}:190603, 2012.


\bibitem{ImSaKPZ2}
T.~Imamura, T.~Sasamoto.
\newblock  Stationary correlations for the 1D KPZ equation.
\newblock {\it J. Stat. Phys.}, {\bf 150}:908--939, 2013.

\bibitem{ISS}
T.~Imamura, T.~Sasamoto, H.~Spohn.
\newblock On the equal time two-point distribution of the one-dimensional KPZ equation by replica.
\newblock arXiv:1305.1217.

\bibitem{KJ}
K.~Johansson.
\newblock Shape fluctuations and random matrices.
\newblock {\em Commun. Math. Phys.}, {\bf 209}:437--476, 2000.

\bibitem{K}
M.~Kardar.
\newblock Replica-Bethe Ansatz studies of two-dimensional interfaces with quenched random impurities.
\newblock {\it Nucl. Phys. B}, {\bf 290}:582--602, 1987.

\bibitem{KPZ}
K.~Kardar, G.~Parisi, Y.Z.~Zhang.
\newblock  Dynamic scaling of growing interfaces.
\newblock {\em Phys. Rev. Lett.}, {\bf 56}:889--892, 1986.



%\bibitem{KC}
%V.~Kac, P.~Cheung.
%\newblock {\em Quantum Calculus.}
%\newblock Springer-Verlag, 2002.

%\bibitem{Lax}
%P.~D.~Lax.
%\newblock {\it Functional analysis}.
%\newblock Wiley-Interscience, 2002.

%\bibitem{Lig}
%T.~Liggett.
%\newblock {\it Interacting particle systems}.
%\newblock Spinger-Verlag, Berlin, 2005.

\bibitem{LD}
T.~Gueudre, P.~Le Doussal.
\newblock Directed polymer near a hard wall and KPZ equation in the half-space.
\newblock {\it Euro. Phys. Lett}, {\bf 100}:26006, 2012.

\bibitem{LDexpand}
T.~Gueudre, P.~Le Doussal.
\newblock In preparation, 2015.


\bibitem{LeDoussalThierry1}
\newblock P. Le Doussal, T. Thierry.
\newblock Log-Gamma directed polymer with fixed endpoints via the replica Bethe Ansatz.
\newblock {\it J. Stat. Mech.} P10018, 2014.

\bibitem{LeDoussalThierry2}
\newblock P. Le Doussal, T. Thierry.
\newblock On integrable directed polymer models on the square lattice
\newblock arXiv:1506.05006.


\bibitem{LFC}
S.~Lemerle, J.~Ferr\'{e}, C.~Chappert, V.~Mathet, T.~Giamarchi, P.~Le Doussal.
\newblock Domain wall creep in an Ising ultrathin magnetic film.
\newblock {\it Phys. Rev. Lett.}, {\bf 80}:849--852, 1998.

%\bibitem{LiebPerComm}
%E.H. Lieb.
%\newblock Personal Communication, December 2014.

\bibitem{LL}
E.H.~Lieb, W.~Liniger.
\newblock Exact Analysis of an Interacting Bose Gas. I. The General Solution and the Ground State.
\newblock {\it Phys. Rev. Lett.}, {\bf 130}:1605--1616, 1963.

\bibitem{M}
I.G.~Macdonald.
\newblock {\it Symmetric Functions and Hall Polynomials.}
\newblock 2nd ed. Oxford University Press, New York. 1999.

\bibitem{SM}
S.N.~Majumdar.
\newblock Random matrices, the Ulam problem, directed polymers and growth models, and sequence matching. In: Les Houches Summer School Proceeding 85, 179-216, 2007.


\bibitem{MMN}
S.N.~Majumdar, K.~Mallick, S.~Nechaev.
\newblock Bethe Ansatz in the Bernoulli matching model of random sequence alignment.
\newblock {\it Phys. Rev. E}, {\bf 77}:011110, 2008.

\bibitem{MWI}
M.~Matsushitaa, J.~Wakitaa, H.~Itoha, I.~R\'{a}folsa, T.~Matsuyamab, H.~Sakaguchic, M.~Mimurad.
\newblock Interface growth and pattern formation in bacterial colonies.
\newblock {\it Phys. A}, {\bf 249}:517--524, 1998.

\bibitem{McGuire}
J.~B.~McGuire.
\newblock Study of exactly soluble one-dimensional N-body problems.
\newblock {\it J. Math. Phys.}, {\bf 5}:622, 1964.

\bibitem{Mol}
S.~Molchanov
\newblock {\it Lectures on random media. Lectures on Probability Theory. Ecole d'Et\'{e} de Probabilit\'{e}s de Saint-Flour XXII--1992.}
\newblock Lecture Notes in Math. {\bf 1581}:242--411, Springer, Berlin, 1994.

\bibitem{Gregorio}
G.~Moreno Flores.
\newblock On the (strict) positivity of solutions of the stochastic heat equation.
\newblock  arXiv:1206.1821.

%\bibitem{MGP}
%J.~MacDonald, J.~Gibbs, A.~Pipkin.
%\newblock Kinetics of biopolymerization on nucleic acid templates.
%\newblock {\it Biopolymers}, {\bf 6}, 1968.

%\bibitem{QRMF}
%G.~Moreno Flores, D.~Remenik, J.~Quastel.
%\newblock In preparation.


\bibitem{OCon}
N.~O'Connell.
\newblock Directed polymers and the quantum Toda lattice.
\newblock  {\em Ann. Probab.}, {\bf 40}:437--458, 2012.

%\bibitem{OConOrth}
%N.~O'Connell.
%\newblock Tracy-Widom asymptotics for a random polymer model with gamma-distributed weights.
%\newblock {\it Elect. J. Probab.} {\bf 20}:1-18, 2015.
%
%\bibitem{OConPei}
%N.~O'Connell, Y.~Pei.
%\newblock A $q$-weighted version of the Robinson-Schensted algorithm.
%\newblock arXiv:1212.6716.
%
\bibitem{OSZ}
N.~O'Connell, T.~Sepp\"{a}l\"{a}inen, N.~Zygouras.
\newblock Geometric RSK correspondence, Whittaker functions and symmetrized random polymers.
\newblock {\bf Inventiones}, {\bf 197}:361--416, 2014.

\bibitem{OY}
N.~O'Connell, M.~Yor.
\newblock Brownian analogues of Burke's theorem.
\newblock {\it Stoch. Proc. Appl.}, {\bf 96}:285--304, 2001.


%\bibitem{OpdamPerComm}
%E. Opdam.
%\newblock Personal Communications, May 2014.


%\bibitem{OkRes}
%A.~Okounkov, N.~Reshetikhin.
%\newblock Correlation function of Schur process with application to local geometry of a random 3-dimensional Young diagram.
%\newblock {\it J. Amer. Math. Soc.}, {\bf 16}:581--603, 2003.

\bibitem{Oxford}
S.~Oxford.
\newblock {\it The Hamiltonian of the quantized nonlinear Schr\"{o}dinger equation}.
\newblock Ph.D. thesis, UCLA, 1979.


\bibitem{ProS1}
S.~Prolhac, H.~Spohn.
\newblock Two-point generating function of the free energy for a directed polymer in a random medium.
\newblock {\it J. Stat. Mech.} P01031, 2011.

\bibitem{ProS2}
S.~Prolhac, H.~Spohn.
\newblock The one-dimensional KPZ equation and the Airy process.
\newblock {\it J. Stat. Mech.} P03020, 2011.

\bibitem{ProSpoComp}
S.~Prolhac, H.~Spohn.
\newblock The propagator of the attractive delta-Bose gas in one dimension.
\newblock {\it J. Math. Phys.}, {\bf 52}:122106, 2011.

\bibitem{Pov}
A.~Povolotsky.
\newblock Bethe ansatz solution of zero-range process with nonuniform stationary state.
\newblock {\it Phys. Rev. E} {\bf 69}:061109, 2004.

%\bibitem{SanSch}
%S.~Sandow, G.~Sch\"{u}tz.
%\newblock On $U_q[su(2)]$-symmetric driven diffusion.
%\newblock {\em Euro. Phys. Lett.} {\bf 26}:7, 1994.

%\bibitem{Spohn}
%H.~Spohn.
%\newblock {\em Large Scale Dynamics of Interacting Particles}.
%\newblock Springer, 1991.

\bibitem{SAY}
M.E.~Sardiu, G.~Alves, Y.K.~Yu.
\newblock Score statistics of global sequence alignment from the energy distribution of a modified directed polymer and directed percolation problem.
\newblock {\it Phys. Rev. E} {\bf 72}:061917, 2005.


\bibitem{SaSp}
T.~Sasamoto, H.~Spohn.
\newblock One-dimensional KPZ equation: an exact solution and its universality.
\newblock {\em Phys. Rev. Lett.}, {\bf 104}:23, 2010.

\bibitem{SaSpBM}
T.~Sasamoto, H.~Spohn.
\newblock Point-interacting Brownian motions in the KPZ universality class.
\newblock {\it Elect. J. Probab.}, {\bf 20}:1--28, 2015.

\bibitem{SasWad}
T.~Sasamoto, M.~Wadati.
\newblock Exact results for one-dimensional totally asymmetric diffusion models.
\newblock {\it J. Phys. A}, {\bf 31}:6057--6071, 1998.


\bibitem{SeppLog}
T.~Sepp\"{a}l\"{a}inen.
\newblock Scaling for a one-dimensional directed polymer with boundary conditions.
\newblock {\it Ann. Probab.}, {\bf 40}:19--73, 2012.

%\bibitem{Schutz}
%G.~M.~Sch\"{u}tz.
%\newblock Duality relations for asymmetric exclusion processes.
%\newblock {\em J. Stat. Phys.}, {\bf 86}:1265--1287, 1997.

%\bibitem{Spitzer}
%F.~Spitzer.
%\newblock Interaction of Markov processes.
%\newblock {\em Adv. Math.}, {\bf 5}:246--290, 1970.

%\bibitem{SpohnStochasticIntegrability}
%H.~Spohn.
%\newblock Stochastic integrability and the KPZ equation.
%\newblock IAMP News Bulletin, April 2012.

%\bibitem{TW}
%C.~Tracy and H.~Widom.
%\newblock Level-spacing distributions and the Airy kernel.
%\newblock {\em Comm. Math. Phys.}, 159:151--174, 1994.

%\bibitem{Tak}
%Y.~Takeyama
%\newblock A discrete analogue of periodic delta Bose gas and affine Hecke algebra.
%\newblock arXiv:1209.2758.

\bibitem{TW1}
C.~Tracy, H.~Widom.
\newblock Integral formulas for the asymmetric simple exclusion process.
\newblock {\em Commun. Math. Phys.}, {\bf 279}:815--844, 2008.

\bibitem{TW2}
C.~Tracy, H.~Widom.
\newblock A Fredholm determinant representation in ASEP.
\newblock {\em J. Stat. Phys.}, {\bf 132}:291--300, 2008.

\bibitem{TW3}
C.~Tracy, H.~Widom.
\newblock Asymptotics in ASEP with step initial condition.
\newblock {\em Commun. Math. Phys.}, {\bf 290}:129--154, 2009.

\bibitem{TWhalfspace}
C.~Tracy, H.~Widom.
\newblock The Bose Gas and Asymmetric Simple Exclusion Process on the Half-Line.
\newblock {\it J. Stat. Phys.}, {\bf 150}:1--12, 2013.

%\bibitem{TW4}
%C.~Tracy, H.~Widom.
%\newblock Formulas for ASEP with two-sided Bernoulli initial condition.
%\newblock {\em J. Stat. Phys.}, {\bf 140}:619--634, 2010.

%\bibitem{Varadhan}
%S.~R.~S.~Varadhan.
%\newblock {\it Stochastic processes.}
%\newblock AMS, 2007.


\bibitem{Vidya}
V.~Venkateswaran.
\newblock Symmetric and nonsymmetric Hall--Littlewood polynomials of type BC.
\newblock arXiv:1209.2933.

\bibitem{W}
J.~Walsh.
\newblock An introduction to stochastic partial differential equations. In: Ecole d'Ete de Probabilites de
Saint Flour XIV, Lecture Notes in Mathematics n. 1180. Berlin: Springer, 1986.


\bibitem{Yang1}
C.~N.~Yang.
\newblock Some exact results for the many body problem in one dimension with repulsive delta function interaction.
\newblock {\it Phys. Rev. Lett.}, {\bf 19}:1312--1314, 1967.

\bibitem{Yang2}
C.~N.~Yang.
\newblock S matrix for the one dimensional N-body problem with repulsive or attractive delta-function interaction.
\newblock {\it Phys. Rev.}, {\bf 168}:1920--1923, 1968.


\bibitem{Yudson}
V.~I.~ Yudson.
\newblock Dynamics of integrable quantum systems.
\newblock {\it Zh. Eksp. Teor. Fiz.}, {\bf 88}: 1757--1770, 1985.

\end{thebibliography}
\end{document}